\documentclass[a4paper,11pt]{article}

\usepackage[top=1in, bottom=1in, left=1.25in, right=1.25in]{geometry}

\usepackage[utf8]{inputenc}
\usepackage{xspace}
\usepackage[french,english]{babel}
\usepackage[T1]{fontenc}

\usepackage{csquotes}

\usepackage{graphicx,color} 
\usepackage[svgnames]{xcolor}

\usepackage[sortcites,sorting=nyt,backend=biber,hyperref=true,maxbibnames=99]{biblatex} 
\renewbibmacro*{byeditor+others}{%
  \setunit{\addcomma\space}%
  \ifnameundef{editor}
    {}
    {\usebibmacro{byeditor+othersstrg}%
     \setunit{\addspace}%
     \printnames[byeditor][1-1]{editor}%
     \clearname{editor}%
     \newunit}%
  \usebibmacro{byeditorx}%
  \usebibmacro{bytranslator+others}}

\renewbibmacro*{doi+eprint+url}{%
    \printfield{doi}%
    \newunit\newblock%
    \iftoggle{bbx:eprint}{%
        \usebibmacro{eprint}%
    }{}%
    \newunit\newblock%
    \iffieldundef{doi}{%
        \usebibmacro{url+urldate}}%
        {}%
    }

\usepackage{tikz}
\usetikzlibrary{automata}
\tikzset{
		>=latex,
		blk/.style={state,fill=black},
		wht/.style={state,fill=white},
		node distance=0.5\columnwidth
}

\usepackage{algorithm}
\usepackage[noend]{algorithmic}

\algsetup{indent=2em}

\usepackage{mathptmx}
\usepackage{amsmath}
\usepackage{amsfonts}
\usepackage{amssymb}
\usepackage{graphicx}
\usepackage{float}
\usepackage{amsthm}
\usepackage{enumitem}
\usepackage{adjustbox}
\usepackage{xcolor}
\usepackage{tabularx}
\usepackage{alltt}
\usepackage{xspace}
\usepackage{caption}
\usepackage[subrefformat=parens]{subcaption}
\usepackage{changepage}
\usepackage{mdframed}
\usepackage{lipsum}
\usepackage{listings}

\usepackage[normalem]{ulem}

\usepackage{tikz}
\usetikzlibrary{fit}
\usetikzlibrary{automata}
\usetikzlibrary{calc}
\usetikzlibrary{positioning}
\usetikzlibrary{matrix}
\usepackage{bookmark,hyperref}

\usepackage[acronym,toc]{glossaries}

\usepackage{dirtytalk}

\pgfdeclarelayer{background}
\pgfsetlayers{background,main}

\tikzstyle{lcm}=[rectangle,draw=black!50,fill=white,text width=3.5mm,align=center]
\tikzstyle{schedot}=[coordinate]
\tikzstyle{sched}=[coordinate,rectangle,text width=3.5mm]
\tikzstyle{cycle}=[rectangle,rounded corners=5pt,fill=black!10]
\tikzset{
		>=latex,
		node distance=0.4\columnwidth
}

\numberwithin{algorithm}{section}

\numberwithin{lemma}{section}

\numberwithin{corollary}{section}

\newtheorem{theorem}{Theorem}
\numberwithin{theorem}{section}

\numberwithin{property}{section}

\numberwithin{definition}{section}

\numberwithin{observations}{section}

\newtheorem*{observations*}{Observations}

\newtheorem{condition}{Condition}
\numberwithin{condition}{section}

\numberwithin{conjecture}{section}

\newcommand\LOOK{\mbox{\textsc{Look}}\xspace}
\newcommand\COMPUTE{\mbox{\textsc{Compute}}\xspace}

\newcommand\MOVE{\mbox{\textsc{Move}}\xspace}

\newcommand\FSYNC{\mbox{FSYNC}\xspace}
\newcommand\SSYNC{\mbox{SSYNC}\xspace}
\newcommand\ASYNC{\mbox{ASYNC}\xspace}

\newcommand\STAY{\mbox{\textsc{Stay}}\xspace}

\newcommand\HALF{\mbox{\textsc{M2H}}\xspace}

\newcommand\BLACK{\robcol{Black}\xspace}
\newcommand\WHITE{\robcol{White}\xspace}

\newcommand\Black{\BLACK}
\newcommand\White{\WHITE}

\newcommand\robcol[1]{\mbox{{\textsc{#1}}}}
\newcommand\robstate[1]{\mbox{\textbf{\textsc{#1}}}}

\newcommand\SEC{\mbox{SEC}\xspace}

\newcommand\probName[1]{\mbox{\textbf{\textsf{#1}}}}

\newcommand\RENDEZVOUS{\probName{Rendez}\-\probName{vous}\xspace}
\newcommand\LEADEL{\probName{Leader} \probName{Election}\xspace}
\newcommand\GEOLEADEL{\probName{Geoleader} \probName{Election}\xspace}
\newcommand\CONVERGENCE{\probName{Convergence}\xspace}
\newcommand\GATHERING{\probName{Gathering}\xspace}

\newcommand\OBLOT{\mbox{$\mathcal{OBLOT}$}\xspace}

\newcommand\luminous{\mbox{$\mathcal{LUMINOUS}$}\xspace}

\DeclareMathAlphabet{\mathcalbf}{OMS}{cmsy}{b}{n}
\DeclareMathAlphabet{\mathcal}{OMS}{cmsy}{m}{n}

\graphicspath{{png/}}

\title{Unreliable Sensors for \\ Reliable Efficient Robots}
\author{Adam Heriban$^\star$ \and Sébastien Tixeuil$^\star$}
\date{$^\star$Sorbonne University, CNRS, LIP6, France}

\begin{document}

\selectlanguage{english}

\maketitle

\begin{abstract}
The vast majority of existing Distributed Computing literature about mobile robotic swarms considers \emph{computability} issues: characterizing the set of system hypotheses that enables problem solvability.
By contrast, the focus of this work is to investigate \emph{complexity} issues: obtaining quantitative results about a given problem that admits solutions. Our quantitative measurements rely on a newly developed simulation framework to benchmark pen and paper designs.

First, we consider the maximum traveled distance when gathering robots at a given location, not known beforehand (both in the two robots and in the $n$ robots settings) in the classical $\mathcal{OBLOT}$ model, for the \FSYNC, \SSYNC, and \ASYNC schedulers. This particular metric appears relevant as it correlates closely to what would be real world fuel consumption.
Then, we introduce the possibility of errors in the vision of robots, and assess the behavior of known rendezvous (\emph{aka} two robots gathering) and leader election protocols when sensors are unreliable. We also introduce two new algorithms, one for fuel efficient convergence, and one for leader election, that operate reliably despite unreliable sensors. 
\end{abstract}

\section{Introduction}

Since the seminal work of Suzuki and Yamashita~\cite{DBLP:journals/siamcomp/SuzukiY99}, much research on cooperative mobile robots was aimed at identifying the minimal assumptions (in terms of synchrony, sensing capabilities, environment, etc.) under which basic problems can be solved. A recent state of the art was recently proposed by Flocchini et al.~\cite{DBLP:series/lncs/Flocchini19}.

Robots are modeled as mathematical points in the 2D Euclidean plane and independently execute their own instance of the same algorithm.
In the model we consider, robots are anonymous (\emph{i.e.}, they are indistinguishable from each-other), oblivious (\emph{i.e.}, they have no persistent memory of the past is available), and disoriented (\emph{i.e.}, they do not agree on a common coordinate system). The robots operate in Look-Compute-Move cycles. In each cycle, a robot "Looks" at its surroundings and obtains (in its own coordinate system) a snapshot containing the locations of all robots. Based on this visual information, the robot "Computes" a destination location (still in its own coordinate system), and then "Moves" towards the computed location. Since the robots are identical, they all follow the same deterministic algorithm. The algorithm is oblivious if the computed destination in each cycle depends only on the snapshot obtained in the current cycle (and not on stored previous snapshots). The snapshots obtained by the robots are not consistently oriented in any manner (that is, the robots' local coordinate systems do not share a common direction nor a common chirality\footnote{Chirality denotes the ability to distinguish left from right.}).

The execution model significantly impacts the ability to solve collaborative tasks. Three different levels of synchronization have been commonly considered. The strongest model is the fully-synchronous (\FSYNC) model~\cite{DBLP:journals/siamcomp/SuzukiY99}, where each phase of each cycle is performed simultaneously by all robots. The semi-synchronous (\SSYNC) model~\cite{DBLP:journals/siamcomp/SuzukiY99} considers that time is discretized into rounds, and that in each round an arbitrary yet non-empty subset of the robots are active. The robots that are active in a particular round perform exactly one atomic \LOOK-\COMPUTE-\MOVE cycle in that round. The weakest model is the asynchronous (\ASYNC) model~\cite{DBLP:series/synthesis/2012Flocchini,DBLP:journals/tcs/FlocchiniPSW05}, which allows arbitrary delays between the \LOOK,\COMPUTE and \MOVE phases, and the movement itself may take an arbitrary amount of time. It is assumed that the scheduler (seen as an adversary) is fair in the sense that in each execution, every robot is activated infinitely often.

\subsection{Previous works and Motivations}

An important shortcoming of the robot model introduced by Suzuki and Yamashita \cite{DBLP:journals/siamcomp/SuzukiY99} with respect to real-world implementation of mobile robot algorithms is the assumption that both the vision sensors and the actuation motors are perfect. More specifically, the model assumes that robots have an infinite vision range, and can sense the position of other robots relatively to theirs with infinite accuracy. Robots are also usually able to reach their target with infinite movement precision (with respect to the angle to the target).

Several attempts have been made to make the \OBLOT model more realistic, \emph{e.g.} by limiting the range of sensors through the limited visibility model~\cite{DBLP:journals/trob/AndoOSY99,DBLP:conf/antsw/GordonEB08,DBLP:conf/antsw/GordonWB04}, by allowing the sensors to miss other robots~\cite{DBLP:conf/sirocco/HeribanT19}, by using inaccurate sensors~\cite{DBLP:journals/siamcomp/CohenP08,DBLP:conf/antsw/GordonEB08,DBLP:conf/antsw/GordonWB04,DBLP:journals/automatica/Martinez09,DBLP:journals/tcs/YamamotoIKIW12}, or by discarding the hypothesis that robots are transparent~\cite{DBLP:conf/cccg/LunaFPSV14,DBLP:journals/tcs/HonoratPT14}.

However, many attempts are hindered by increased complexity due to manually proving algorithms in those more complex settings. For instance, to our knowledge, the consequences of error-prone vision have only been studied through very simple problems: \GATHERING and \CONVERGENCE~\cite{DBLP:journals/siamcomp/CohenP08,DBLP:conf/antsw/GordonEB08,DBLP:conf/antsw/GordonWB04,DBLP:journals/automatica/Martinez09,DBLP:journals/tcs/YamamotoIKIW12}.

To allow more complex problems to be studied considering more realistic settings, it appears necessary to favor an machine-helped approach.

Formal methods encompass a long-lasting path of research that is meant to overcome errors of human origin. Unsurprisingly, this mechanized approach to protocol correctness was used in the context of mobile robots~\cite{DBLP:conf/srds/BonnetDPPT14,DBLP:conf/sss/DevismesLPRT12,DBLP:journals/dc/BerardLMPTT16,DBLP:conf/sss/AugerBCTU13,DBLP:conf/sss/MilletPST14,DBLP:journals/ipl/CourtieuRTU15, berard:hal-01238784,DBLP:conf/prima/RubinZMA15,DBLP:conf/sss/DevismesLPRT12,DBLP:conf/fmcad/SangnierSPT17,DBLP:conf/icdcn/BalabonskiPRT18,DBLP:journals/mst/BalabonskiDRTU19,DBLP:conf/netys/BalabonskiCPRTU19,DBLP:journals/fmsd/SangnierSPT20,DBLP:conf/srds/DefagoHTW20,}.

When robots move freely in a continuous two-dimensional Euclidean space, to the best of our knowledge, the only formal framework available is Pactole\footnote{\url{http://pactole.lri.fr}}.
Pactole relies on higher-order logic to certify impossibility results~\cite{DBLP:conf/sss/AugerBCTU13,DBLP:journals/ipl/CourtieuRTU15,DBLP:conf/icdcn/BalabonskiPRT18}, as well as the correctness of algorithms~\cite{DBLP:conf/wdag/CourtieuRTU16,DBLP:journals/mst/BalabonskiDRTU19} in the \FSYNC and \SSYNC models, possibly for an arbitrary number of robots (hence in a scalable manner). Pactole was recently extended by Balabonski~\emph{et al.}~\cite{DBLP:conf/netys/BalabonskiCPRTU19} to handle the \ASYNC model, thanks to its modular design. However, in its current form, Pactole lacks automation; that is, in order to prove a result formally, one still has to write the proof (that is automatically verified), which requires expertise both in Coq (the language Pactole is based upon) and about the mathematical and logical arguments one should use to complete the proof. 

On the other hand, model checking and its derivatives (automatic program synthesis, parameterized model checking) hint at more automation once a suitable model has been defined with the input language of the model checker. 
In particular, model-checking proved useful to find bugs (usually in the \ASYNC setting)~\cite{DBLP:journals/dc/BerardLMPTT16,DBLP:conf/sofl/DoanBO16,DBLP:conf/opodis/Doan0017} and to formally check the correctness of published algorithms~\cite{DBLP:conf/sss/DevismesLPRT12,DBLP:journals/dc/BerardLMPTT16,DBLP:conf/prima/RubinZMA15,DBLP:conf/srds/DefagoHTW20}. Automatic program synthesis~\cite{DBLP:conf/srds/BonnetDPPT14,DBLP:conf/sss/MilletPST14} was used to obtain automatically algorithms that are "correct-by-design". However, those approaches are limited to instances with few robots. Generalizing them to an arbitrary number of robots with similar models is doubtful as Sangnier \emph{et al.}~\cite{DBLP:journals/fmsd/SangnierSPT20} proved that safety and reachability problems are undecidable in the parameterized case with default models. Another limitation of the above approaches is that they \emph{only} consider cases where mobile robots \emph{evolve in a \underline{discrete} space} (\emph{i.e.}, graph). This limitation is due to the model used, that closely matches the original execution model by Suzuki and Yamashita~\cite{DBLP:journals/siamcomp/SuzukiY99}. As a computer can only model a finite set of locations, a continuous 2D Euclidean space cannot be expressed in this model.

Overall, the only way to obtain automated proofs of correctness in the continuous space context through model checking is to use a more abstract model~\cite{DBLP:conf/wdag/DefagoHTW19,DBLP:conf/srds/DefagoHTW20}, which commands writing additional handwritten theorems to assess its relevance in the original model.
Overall, using formal methods for complex algorithms in realistic settings requires a substantial effort that may be out of reach when one simply wants to asses the feasibility of an algorithmic design.

Furthermore, these approaches currently only address whether the added constraints enable the construction of counter-examples for a given task, and, to the best of our knowledge, do not address the important issue of performance degradation, or, in the cases where counter-examples do appear, the likelihood of their appearance and their impact.

In fact, an overwhelming majority of the research on mobile robotic swarms has focused on proving, under a given set of conditions, whether there exists a counter example to a given solution proposal for a problem. On the other hand, the practical efficiency of a given algorithm (with respect to real-world criteria such as fuel consumption) was rarely studied by the Distributed Computing community, albeit being of paramount importance to the Robotics community~\cite{DBLP:conf/gecco/AroraMDB19,DBLP:journals/ijrr/YooFS16}.
Fuel-constrained robots have been considered in the discrete graph context, for both exploration \cite{DBLP:conf/arcs/DyniaKS06} and distributed package delivery \cite{DBLP:conf/algosensors/Chalopin0MPW13}, but, to our knowledge, no study considered the two-dimensional Euclidean space model that was promoted by Suzuki and Yamashita~\cite{DBLP:journals/siamcomp/SuzukiY99}. A possible explanation for this situation is that the more complex the algorithm (or the system settings), the more difficult it becomes to rigorously find the worst possible executions.

We investigate another approach: since our goal is to bridge the gap between theoretical mobile robots, and actual robotics, we move one step towards robotics and use a very common tool: simulation. First, robot simulators, such as Gazebo\footnote{\texttt{http://gazebosim.org/}}, are industry standard tools for designing physical robots. Then, simulating mobile robots is not a new idea, and has been tried since the very beginning of mobile robots~\cite{DBLP:journals/trob/AndoOSY99}.

Our goal is to design and implement a practical simulator for networks of mobile robots that is focused on finding counter-examples and monitoring network behavior, rather than proving algorithms or providing a visual representation. Our vision is that this tool is especially useful in the early stages of algorithm design to eliminate obviously wrong paths, and detect anomalies. It should not be seen as a replacement for formal tools, but as a replacement for researcher intuition when working on a mobile robot network model or algorithm.

As such, the simulator should be easy to use, understand and modify by any Distributed Computing researcher in order to include any new algorithm or model. 
It should also be capable of monitoring network behavior and output quantitative data points to assess real world performance, according to a given set of metrics, as well as enabling comparison with previously proven algorithms in a given setting.

We first focus on the known limitations of this approach and highlight the difficulty of encoding victory and defeat conditions for the computed executions, and how it impacts our ability to reliably detect counter-examples, as well as the expected consequences of working in a discretized Euclidean space, such as the impossibility to distinguish \CONVERGENCE and \GATHERING. 

\subsection{Our Contribution}

In this paper, we design and implement a practical simulator for mobile robotic swarms evolving in a two-dimensional Euclidean space. To circumvent the obvious problem of an infinite number of initial positions, our simulation framework is based on the Monte Carlo method for choosing initial configurations~\cite{MonteCarlo49}. 

We first benchmark our simulation framework using a well known problem in the domain: rendezvous. Rendezvous mandates that two robots gather in finite time at the same location, not known beforehand. There exists a number of rendezvous solution for various settings, yet our simulation framework enables fair quantitative comparison. We choose the fuel metric (\emph{a.k.a.} total traveled distance) under various system conditions: \FSYNC, \SSYNC, and \ASYNC schedulers with or without rigid motion.

We then assess the impact of inaccurate visibility sensors on two milestone algorithms: the Center-of-Gravity convergence algorithm~\cite{DBLP:journals/siamcomp/SuzukiY99} for two robots, and the Geoleader election algorithm~\cite{DBLP:conf/sss/CanepaP07}. It turns out that their behavior is significantly impacted by even small inaccuracies. 

To address the shortcoming identified by our simulations in the literature, we design a new two-color, fuel-efficient, convergence algorithm for the ASYNC scheduler, and an improved leader election algorithm that is resilient to inaccurate vision. Both proposal are similarly benchmarked with our simulation framework. 

The rest of the paper is organized as follows. Section~\ref{chap:MonteCarlo} presents the core technicalities underlying our simulation framework, and its limitations through the problems of \OBLOT \FSYNC \CONVERGENCE, and \GEOLEADEL. Section~\ref{chap:performance} demonstrates how the framework can be used for the purpose of performance evaluation, while Section~\ref{chap:realistic} show how realistic error models can be integrated into the entire evaluation process. Section~\ref{chap:improved} introduces two new algorithms, one for fuel efficient convergence, and one for leader election with unreliable sensors. Finally, Section~\ref{sec:conclusion} provides concluding remarks.

\section{Monte-Carlo Simulation of Mobile Robots}

\label{chap:MonteCarlo}

\subsection{Overview of the Framework}

Our simulation framework is written from scratch using Python 3, ensuring a large compatibility across executing platforms. Our design goal is to remain as close as possible to the theoretical model of Suzuki and Yamashita~\cite{DBLP:journals/siamcomp/SuzukiY99}, in order to maximize readability and usability by the mobile robot distributed computing community.

Each mobile entity is thus encapsulated as an instance of the \texttt{Robot} class.
In the case of the basic \OBLOT model, robots have the following properties:

\begin{itemize}
    \item A unique \texttt{name}.
    \item \texttt{x} and \texttt{y} coordinates in the Euclidean plane.
    \item A \texttt{snapshot} list of \texttt{Robot}s that contains visible \texttt{Robot}s.
    \item A \texttt{target}, which is a 2-tuple of the \texttt{x} and \texttt{y} coordinates of the target destination.
\end{itemize}

\noindent The \texttt{Robot} class also provides three methods: 

\begin{itemize}
    \item The \texttt{LOOK} method uses the network as an input. It creates a list of the visible \texttt{Robot}s in the network and assigns it to \texttt{snapshot}.
    \item The \texttt{COMPUTE} method uses \texttt{snapshot} to compute and assign \texttt{target}, according to the algorithm we want to evaluate.
    \item The \texttt{MOVE} method updates \texttt{x} and \texttt{y} according to \texttt{target}.
\end{itemize}

This is summarized in figure~\ref{fig:robotClass}

\begin{figure}[htb]
    \centering
    \includegraphics[width=0.35\linewidth]{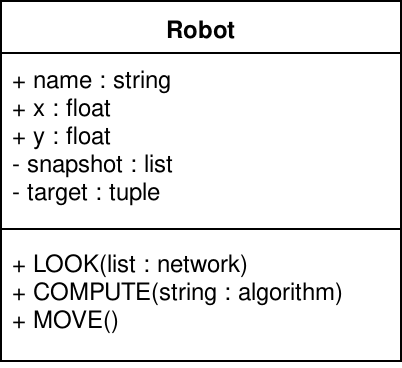}
    \caption{Robot Class}
    \label{fig:robotClass}
\end{figure}

Because robots are anonymous, \texttt{name} cannot be used for computing purposes, and is simply a way for the scheduler to reliably monitor the robots in the network. Similarly, robots cannot use \texttt{x} and \texttt{y} directly as they are disoriented. 
The simulation consists of two parts: an initializing sequence and a loop.
The \emph{initializing sequence} creates a \texttt{network} list, which contains all robots, according to simulation parameters. 
To circumvent the problem of the infinite number of possible initial positions, our simulation framework is based on the Monte-Carlo method for choosing initial configurations~\cite{MonteCarlo49}. So, unless otherwise specified, the initial location of each robot is chosen uniformly at random within the bounds of the type used to represent positions. Using the Monte-Carlo method allows us to both minimize biases in the initial parameters, and arbitrarily increase the precision of the simulation by simply increasing the number of simulations. 
For each iteration of the \emph{main loop}, a scheduling function is executed once. In the case of \FSYNC, for each loop iteration, all robots in the network simultaneously perform a \texttt{LOOK}, then simultaneously perform a \texttt{COMPUTE}, and then simultaneously perform a \texttt{MOVE}. Using different schedulers, such as \SSYNC or \ASYNC only requires changing the scheduling function: \SSYNC creates a non empty list of robots to be activated for a whole cycle, and \ASYNC picks a single robot to be activated for a single phase.
The loop terminates whenever a \emph{victory} condition holds, which confirms the algorithm completed its intended task. In the case where an algorithm may fail, a \emph{defeat} condition can also be used.
For practical reasons, the loop has a maximum number of iterations. However, reaching this maximum should not be interpreted as either a failure or a success.

\subsection{Scheduling}

Modeling the \FSYNC scheduler can be trivially done by performing all \LOOK operations, then all \COMPUTE operations, then all \MOVE operations.
For the \ASYNC and \SSYNC schedulers, we rely on randomness to test as many executions as possible.
To model the \SSYNC scheduler, for each time step, we chose a non-empty subset of the network uniformly at random and perform a full cycle.

To model the \ASYNC scheduler, we chose one robot uniformly at random and perform its next operation\footnote{Note that this model does not explicitly include simultaneous operations: we consider that the output of two simultaneous events $E_1$ and $E_2$ can be either the output of $E_1$ then $E_2$ or the output of $E_2$ then $E_1$.}.
In the case of the \ASYNC scheduler, we must also consider what happens if a robot performs a \LOOK operation while another robot is moving. the \OBLOT model usually considers that an adversary can chose the perceived location of the second robot to be anywhere between its initial position and its destination (on a straight line). Modeling this behavior could be easily done by changing the perceived coordinates in the \LOOK operation uniformly at random between the location and the target of the perceived robot (on a straight line). However, existing literature about the \ASYNC model shows that the most problematic scenarios appear when the outdated position perceived for a robot is its initial location. With our simulation framework, we also observed that always choosing the initial location when observing a given robot while it is in its \MOVE phase yielded the most adversarial results, so, while our framework is able to simulate both perceptions, we assume this adversarial behavior in the sequel. 

For all schedulers, our simulation framework supports both the rigid and the non-rigid settings. The rigid setting mandates a robot that selected a distinct target in the \COMPUTE phase to always reach it in the \MOVE phase. The non-rigid setting partially removes this condition: the robot may be stopped by the scheduler before it reaches the target, but not before it traverses a distance of at least $\delta$, for some $\delta>0$.

\subsection{Simulation Conditions}

Our framework uses Monte-Carlo simulation for both the initial conditions and the scheduling. This means we can perform an arbitrarily large number of simulations, which in turn induces an arbitrarily more precise simulation. 
Therefore any criterion on either time, number of iterations, or precision is equivalent. 
Unless specified otherwise, 4 simulation threads are run in parallel, for one hour, on a modern quad-core CPU, after which results are merged and analyzed. We use the PyPy3 JIT compiler instead of the CPython interpreter, for better performance. Results of the 4 simulations are then compiled and analyzed.

\subsection{Comparison with Existing Simulators}

We found two noteworthy simulators for mobile robots: Sycamore and JBotSim.

\emph{Sycamore} is a Java program focused explicitly on mobile robots. However, it appears to be far more complex to build, use and modify than our proposal. Moreover, the latest version we could find seems to date back from 2016, and requires versions of Java that are no longer supported. 

\emph{JBotSim}\footnote{\url{https://jbotsim.io}} is a Java library for simulating distributed networks in general. While it appears to be able to simulate \OBLOT robots, it is not designed to do so. So, one has to dig into the intricacies of the simulator to emulate basic mobile robot settings.

We also found a third Java-based simulator, named oblot-sim\footnote{\url{https://github.com/werner291/oblot-sim}}. 
We are, however, unsure of its provenance and design goals.

All three simulators emphasize real-time visualization of executed algorithms through a complete graphical interface.
Our proposal focuses on extremely simple quantitative simulation. In its current version, a complete instance of the simulator requires only five separate files for a total of less than 30KB of code (The sources for JBotSim and Sycamore weigh 3MB and 4.8MB, respectively). 
We also believe that using Python instead of Java greatly improves portability and ease of understanding, which in turns allows researchers to more easily implement and test unusual settings.

In short, our goal is not to visualize executions in real-time, but to simulate as many executions as possible to process their outcome.

\subsection{Limitations of the Simulation}

While the initial approach described in previous sections may seem sound and simple enough to work with, it results in two distinct problems. As stated previously, our objective with robot simulation is to reliably provide counter-examples whenever they may occur. This requires reliably detecting problematic executions, which is difficult for two reasons. First, success and defeat conditions for most mobile robot algorithms are written in a way that might not be directly usable in a computer simulation. Then, we show that issues predictably arise due to the nature of discretized floating point numbers compared to "true" real numbers used in mathematical models.

\paragraph{Halting the Simulation: \emph{Victory} and \emph{Defeat} Conditions:}

One of the goals of our simulation framework is to find counter-examples for a given algorithm and setting. To do so, we need to simulate the evolution of the network until one of two things happen:

\begin{itemize}
    \item A sufficient condition has been met. This implies that the current execution is successful, and a new simulation with a different initial configuration should begin. This is called a \emph{victory condition}.
    \item A necessary condition has been violated. This implies that the current execution constitutes a counter-example. This is called a \emph{defeat condition}.
\end{itemize}

We illustrate the difficulty of using and defining such conditions in practice through the example of one of the most fundamental problems in the context of mobile robots: \GATHERING. 

The common victory condition for \GATHERING is the following, for two robots $r_1$ and $r_2$:

\begin{condition}[Theoretical \GATHERING Victory]
\label{cond_rdv}
\GATHERING is achieved if and only if, for any pair of robots in the network, the distance between the two robots is eventually always zero. This can also be written more formally as $\exists t_0 \in \mathbb{R}_{\ge 0} : \forall t_1 \ge t_0 , \forall(r_1,r_2) |r_1r_2|_{t_1} = 0$
\end{condition}
In the previous condition, $|r_1r_2|_{t}$ denotes the distance between $r_1$ and $r_2$ at time $t$ in the current execution.

However, this particular condition would require the ability for the simulator to infinitely simulate the future of the network, which is obviously impossible.
Moreover, the matching defeat condition is unusable for similar reasons: \[\nexists t_0 \in \mathbb{R}_{\ge 0} : \forall t_1 \ge t_0 , \forall(r_1,r_2) |r_1r_2|_{t_1} = 0\] 
\begin{center}
or
\end{center}
\[\forall t_0 \in \mathbb{R}_{\ge 0} : \exists t_1 \ge t_0 , \exists(r_1,r_2) |r_1r_2|_{t_1} \neq 0\]

We instead define a more practical defeat condition:
\begin{condition}[Practical \GATHERING Defeat]\
\label{defeat}
$\exists (t_0,t_1) \in (\mathbb{R}_{\ge 0})^2 : t_1>t_0, inputs(t_0) = inputs(t_1), \exists t \in [t_0,t_1] / \exists (r_1,r_2) |r_1r_2|_{t} \neq 0$
\end{condition}

Where $inputs(t)$ is the set of all input parameters relevant to the algorithm. This is different from the configuration, which would contain \emph{all} parameters of the network at a given point of the execution.

This input set is used as a practical way to detect cycles in the execution. For a deterministic algorithm, if all inputs of the algorithm are identical to a previously encountered set of inputs, then a cycle has been found.

The input set we use must be chosen such that for two sets $S_1$ and $S_2$, $S_1(t) = S_2(t) \implies \forall S_1(t+1), \exists S_2(t+1) : S_1(t+1) = S_2(t+1)$. In other words, regardless of the scheduling, two identical sets should not be able to generate different sets.

\begin{theorem}
For two robots executing a deterministic algorithm, if condition~\ref{defeat} is true then condition~\ref{cond_rdv} is false.
\end{theorem}

\begin{proof}
For a deterministic algorithm, if condition~\ref{defeat} is true, there exists a scheduling starting from the initial configuration that reaches $inputs(t_0)$ and $inputs(t_1)$. Because $inputs(t_0) = inputs(t_1)$, there exists a cycle containing non-gathered configurations. Then the adversary scheduler can repeat this cycle infinitely, and condition~\ref{cond_rdv} is false.
\end{proof}

\begin{theorem}
If the number of input sets is finite, then for two robots executing a deterministic algorithm, if condition~\ref{cond_rdv} is false, then condition~\ref{defeat} is true.
\end{theorem}

\begin{proof}
Any scheduling is infinite. So, if the total number of input sets is finite, then every scheduling contains at least one cycle. If condition~\ref{cond_rdv} is false, then there are no non-gathered cycles, so there is at least one gathered cycle that must be repeated, and condition~\ref{defeat} is true.
\end{proof}

One may naively want to use a similar reasoning to define a sufficient victory condition:

\begin{condition}[Naive \GATHERING Victory]
$\exists (t_0,t_1) \in \mathbb{R}_{\ge 0}^2 : t_1>t_0, inputs(t_0) = inputs(t_1), \forall t \in [t_0,t_1], \forall(r_1,r_2) |r_1r_2|_{t} = 0$
\end{condition}

However, this condition ignores the fact that the scheduler may be able to not repeat this cycle by carefully choosing the activation order of the robots.
A proper condition that is usable regardless of the scheduler is the following:

\begin{condition}[Practical \GATHERING Victory]
\label{vict_rdv}
$\forall(r_1,r_2) \exists t_0 \in \mathbb{R}_{\ge 0} : |r_1r_2|_{t_0} = 0 \land \forall \mathcal{S}, \exists t_1 > t_0 : inputs(t_0) = inputs(t_1),\forall t \in [t_0,t_1], |r_1r_2|_{t} = 0$

With $\mathcal{S}$ a scheduling.
In other words, there exists a time after which all robots are stuck in gathered cycles.
\end{condition}

Analyzing configurations and finding cycles in the execution is not an issue for our simulator. The main difficulty lies in our ability to properly model the configuration using the input set. If the set is too restrictive and omits relevant parameters, then we find cycles that do not actually exist. Similarly, a set that is not restrictive enough may hide actual cycles. This depends on both the robot model and the algorithm used to solve the problem.

In the case of \RENDEZVOUS or \GATHERING for two robots, the standard algorithm~\cite{DBLP:journals/siamcomp/SuzukiY99} for the \FSYNC scheduler targets the midpoint between the two robots and is described in algorithm~\ref{algo_rdv}.

\begin{algorithm}[H]
\caption{Basic \FSYNC \RENDEZVOUS} 
\label{algo_rdv}
\begin{algorithmic}
\STATE target[0] = (x + snapshot[0].x)/2
\STATE target[1] = (y + snapshot[0].y)/2
\end{algorithmic}
\end{algorithm}

In the Euclidean space, the number of configurations appears to be infinite. 
Because robots are disoriented, the algorithm uses no information on distance, or coordinate systems, so that all configurations are identical. Then, the input set is actually empty.
This implies that an algorithm succeeds if and only if the network is gathered after the first activation of both robots. Otherwise, the defeat condition is immediately true for rigid movement.

For the sake of providing a second example, let us consider that robots are endowed with weak local multiplicity detection, meaning that they can distinguish a non-gathered configuration from a gathered configuration. This allows us to modify the initial algorithm to algorithm~\ref{algo_rdv2}.

\begin{algorithm}[htb]
\caption{\FSYNC \RENDEZVOUS with Multiplicity Detection} 
\label{algo_rdv2}
\begin{algorithmic}
\IF{$\neg gathered$}
\STATE target[0] = (x + snapshot[0].x)/2
\STATE target[1] = (y + snapshot[0].y)/2
\ENDIF
\end{algorithmic}
\end{algorithm}

In this case, the gathered state is a relevant input parameter, and should be included in the input set. Now, all gathered configurations are considered identical and all non-gathered configurations are considered identical. This means that the robots must still gather after the first activation. However, while this was already considered a cycle with the empty set, if robots are now gathered, the input set is different and no cycle has yet been reached. The first cycle is reached after the second activation. If the robots remain gathered, then this is a gathered cycle and should not trigger the defeat condition. However, if for some reason the robots were to separate after the second activation, this would constitute a non-gathered cycle with the first input set, and the defeat condition would be triggered.

Using this reasoning, we check our simulator against our two-color \ASYNC algorithm~\cite{DBLP:conf/icdcn/HeribanDT18} and the two-color \SSYNC algorithm from Viglietta~\cite{DBLP:conf/algosensors/Viglietta13}. For Heriban two-color, we accurately find no counter-example, and all executions lead to the victory condition in \ASYNC, \SSYNC and \FSYNC. For Viglietta two-color, we accurately find no counter-example and all executions lead to the victory condition in \SSYNC and \FSYNC, and we find counter-examples that trigger the defeat condition in \ASYNC.

We perform a similar study for a weaker version of \GATHERING, called \CONVERGENCE.
The common condition for \CONVERGENCE is the following:

\begin{condition}[Theoretical \CONVERGENCE Victory]\label{conv}
\CONVERGENCE is achieved if and only if, for any distance $\epsilon$ greater than zero, the distance between any pair of robots is eventually always smaller than $\epsilon$.

This can also be written more formally as  $\forall \epsilon \in \mathbb{R}_{> 0}, \exists t_0 \in \mathbb{R}_{\ge 0} : \forall t_1 \ge t_0, \forall(r_1,r_2) |r_1r_2|_{t_1} \le \epsilon$
\end{condition}

Note that, as we expect, \GATHERING implies \CONVERGENCE, but \CONVERGENCE does not imply \GATHERING.
In this case, the distance between the two robots is a relevant parameter to check whether or not the problem is solved. However, since it does not change the behavior of the algorithm, it is still not part of the input set.

We define the following defeat condition:

\begin{condition}[Practical \CONVERGENCE Defeat]\
\label{def_conv}
$\exists(r_1,r_2) : \exists (t_0,t_1) \in (\mathbb{R}_{\ge 0})^2 : t_1>t_0 \land inputs(t_0) = inputs(t_1) \land 0 < |r_1r_2|_{t_0} \le |r_1r_2|_{t_1}$
\end{condition}

\begin{theorem}
For a deterministic algorithm, if condition~\ref{def_conv} is true, then condition~\ref{conv} is false.
\end{theorem}

\begin{proof}
Similarly to \GATHERING, this condition implies a cycle where distance does not decrease, so the adversary scheduler can repeat it infinitely and prevent \CONVERGENCE.
\end{proof}

This does \emph{not} imply that the distance between the two robots must always be strictly decreasing in the general case, as this would neither be a sufficient nor a necessary condition.
Because $\epsilon$ can be infinitely small, we cannot chose the 'right' $\epsilon$ to properly define a victory condition.

\paragraph{The Consequences of the Discretized Euclidean Plane:}
\label{sssec:NRN}

While it is tempting to define a victory condition similar to that of \GATHERING, the question of $\epsilon$ remains.
Floating point numbers are obviously incapable of infinite precision.
So, because any number greater that zero is a valid choice, if $\epsilon$ is smaller than the minimum positive number that can be represented in the chosen floating point precision, it cannot be distinguished from a true zero. This implies that small enough distances between two robots cannot be distinguished from a gathered state.
So, it is intrinsically impossible to distinguish \CONVERGENCE from actual \GATHERING.

Let us modify algorithm~\ref{algo_rdv} so that both robot move towards the midpoint, but only move a distance of $\dfrac{|r_1r_2|}{2} - \dfrac{\delta}{2}$ instead of $\dfrac{|r_1r_2|}{2}$. In theory, this algorithm does not lead to \RENDEZVOUS, as robots reach a distance of $\delta$ after their first activation. However, if $\delta$ is small enough, the precision of floating point numbers is such that $\dfrac{|r_1r_2|}{2} - \dfrac{\delta}{2}$ and $\dfrac{|r_1r_2|}{2}$ appear identical, and the distance $|r_1r_2|$ appears to be zero. This is essentially a \CONVERGENCE algorithm that is fast enough to be mistaken for a \RENDEZVOUS algorithm.
In practice, there is very little that can be done against this sort of behavior and \uline{conditions for \GATHERING should not be considered reliable.}

On the other hand, under different circumstances, the discrete nature of the simulation can instead lead theoretically good executions to fail in practice.
Let us consider a network of two robots $r_1$ and $r_2$ such that $r_2$ does not move, and $r_1$ moves to the midpoint. This should trivially lead to \CONVERGENCE. Let us now assume that $r_1.y = r_2.y$, and that $r_1.x$ and $r_2.x$ are such that $r_2.x$ is the smallest float greater than $r_1.x$. This possibly leads to $\dfrac{r_1.x+r_2.x}{2} = r_1.x$, so $r_1$ stops moving and the defeat condition for \CONVERGENCE is wrongly activated.

We test this by setting $r_1.y = r_2.y = 0$, picking $r_1.x$ at random in $[0,1]$ and picking $r_2.x$ at random in $[2,3]$ so that $r_1.x < r_2.x$.

In the first case, $r_1$ moves to the midpoint and $r_2$ does not move. This results in approximately 37.5\% of one million attempts wrongly failing \CONVERGENCE.

In the second case, $r_2$ moves to the midpoint and $r_1$ does not move. This results in approximately 25.0\% of one million attempts wrongly failing \CONVERGENCE.

This asymmetry may be explained by biases in the binary64 approximation. Regardless, this is a real, hard to predict problem with a non-negligible chance of happening and requires careful analysis of found counter-examples.

Problems with limited float precision also appear when simulating \GEOLEADEL.

\GEOLEADEL is successful if, given a set of robots, each with their own coordinate system, robots can all deterministically agree on a same robot, called the \robstate{Geoleader}.

\GEOLEADEL is known to be impossible in the general case~\cite{DBLP:journals/tcs/DieudonneP12} because of possible symmetries in the network. In practice, this impossibility is circumvented using randomized algorithms to break such symmetries.
Let us consider the state-of-the-art algorithm~\ref{algCan3} by Canepa and Gradinariu Potop-Butucaru~\cite{DBLP:conf/sss/CanepaP07} for three robots.

\begin{algorithm}[H]
\caption{Original \LEADEL Algorithm by Canepa and Gradinariu Potop-Butucaru~\cite{DBLP:conf/sss/CanepaP07} for Three Robots}         
\label{algCan3}         
\begin{algorithmic}
\STATE Compute the angles between two robots
\IF{$my\_angle$ is the smallest} 
\STATE Become \robstate{Leader}
\STATE Exit
\ELSIF{$my\_angle$ is not the smallest, but the other two are identical}
\STATE Become \robstate{Leader}
\STATE Exit
\ELSIF{All angles are identical}
\STATE Perform a Bernoulli trial with a probability of winning of $p = \dfrac{1}{3}$
\IF{Trial won}
\STATE Move perpendicular to the opposite side of the triangle in opposite direction
\ENDIF
\ENDIF
\end{algorithmic}
\end{algorithm}

For this particular algorithm, there are three cases:

\begin{enumerate}
    \item The common case, where one angle is greater than the two others.
    \item A rare case where two angles are identical, and the third one is smaller.
    \item The rarest case where all angles are identical. In that case, a Bernoulli trial is required to degrade to the other cases.
\end{enumerate}

Let us assume a network of three robots, $[r_1,r_2,r_3]$, such that $r_1$ is placed at coordinates $(-0.5,0)$, and $r_2$ at $(0.5,0)$.

We show where each case appears in figure~\ref{fig:Lead_theor}). The third case occurs if $r_3$ is at $(0,\pm \dfrac{\sqrt{3}}{2})$, which are noted as points $eq1$ and $eq2$. Positions of $r_3$ that lead to the second case are noted as $iso1$, $iso2$, and $iso3$.

However, it is \emph{not} possible, using floating point numbers, to have $x$ such that $x^2 = 3$. It is then impossible, regardless of the quality of the simulation, to place $r_3$ on $eq1$ or $eq2$, despite being possible in theory.

Similarly, an infinitely large number of points mathematically located on the circular arcs of the second case cannot be represented properly using floating point numbers.

To test this, each robot is given a new property 'Leader', which is a string containing the name of the \robstate{Leader} robot. We perform the simulation and display the results in figure~\ref{fig:Simu1}. 

\begin{figure}[htb]
    \centering
    \begin{subfigure}[b]{0.49\textwidth}
        \centering
        \includegraphics[width=\textwidth]{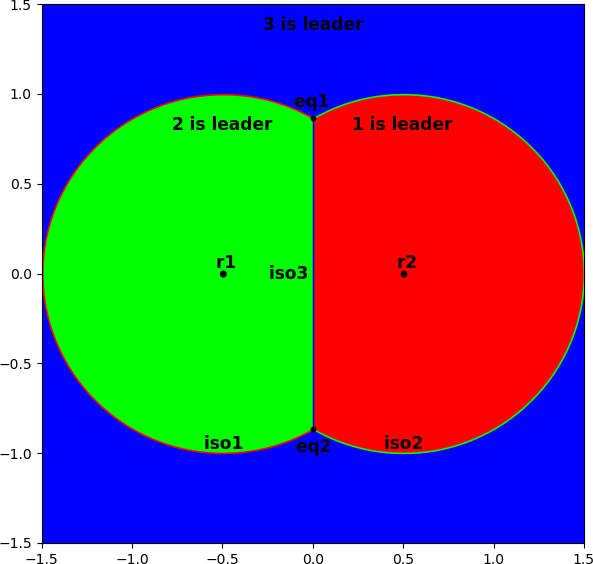}
        \captionsetup{justification=centering}
        \caption{\robstate{Leader} Depending on the Location of $r_3$. Red, green and blue represent $r_1$, $r_2$ and $r_3$, respectively.}
        \label{fig:Lead_theor}
    \end{subfigure}
    \hfill
    \begin{subfigure}[b]{0.49\textwidth}
        \centering
        \includegraphics[width=\textwidth]{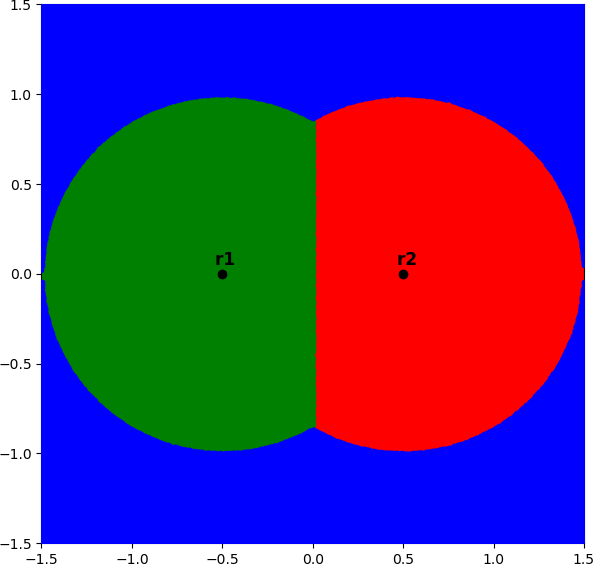}
        \captionsetup{justification=centering}
        \caption{Simulation for 3-robot \LEADEL with Perfect Vision Sensors. No isosceles or equilateral point was found.}
        \label{fig:Simu1}
    \end{subfigure}
\end{figure}

As we predicted, the fact that real numbers cannot be properly represented in our discrete, floating point space prevents the simulator from finding the known counter-example in the case of 3-robots \LEADEL. Furthermore, the three circular arcs on which the second case occurs have a combined surface theoretically equal to zero. Therefore, they are statistically impossible to find using our Monte-Carlo simulation. 

However, it should be noted that, even in a world of perfect sensors, building an equilateral triangle would require placing the third robot with physically impossible precision. So, while this counter-example exists from a mathematical standpoint, it could never occur in a more realistic setting. So, when considering practical robots, this could be considered a minor issue.

On the contrary, the use of a discretized Euclidean space could be viewed as massive advantage compared to the regular, continuous model, as it makes the inherently unrealistic hypothesis of robots being able to store and process snapshots of infinite precision. In this approximated context, snapshots have a known, maximum size, depending on the chosen precision for the coordinates of other robots.

So, in this context, storing a snapshot for a full cycle becomes a trivial matter, and using algorithm \textbf{SyncSim} described by Das \emph{et al.}~\cite{DBLP:conf/icdcs/DasFPSY12,DBLP:journals/tcs/0001FPSY16}, to simulate an \FSYNC scheduling under an \ASYNC scheduler, becomes possible without additional unrealistic hypotheses.

As a result, we believe designing algorithms that properly solve problems in the context of a discretized Euclidean space should be a priority, as it would allow mobile robots to only need to function using the \FSYNC scheduler, and would remove the unrealistic requirement of infinite precision. One such algorithm is shown in Section~\ref{chap:improved}.

\section{Fuel Efficiency in the Usual Settings}
\label{chap:performance}

\noindent The overwhelming majority of the mobile robots research has focused on proving, under a given set of conditions, whether there exists a counter example to a given problem. On the other hand, the practical efficiency of a given algorithm (with respect to real-world criteria such as fuel consumption) was rarely studied by the distributed computing community, albeit commanded by the robotics community~\cite{DBLP:conf/gecco/AroraMDB19,DBLP:journals/ijrr/YooFS16}.

Fuel-constrained robots have been considered in the discrete graph context, for both exploration~\cite{DBLP:conf/arcs/DyniaKS06} and distributed package delivery~\cite{DBLP:conf/algosensors/Chalopin0MPW13}. However, to our knowledge, no study considered the two-dimensional Euclidean space model that was promoted by Suzuki and Yamashita~\cite{DBLP:journals/siamcomp/SuzukiY99}. A possible explanation for this situation is that the more complex the algorithm (or the system setting), the more difficult it becomes to rigorously find the worst possible execution.

\subsection[\textsf{Rendezvous} Algorithms]{\RENDEZVOUS Algorithms}

\noindent We first quantify the maximum traveled distance and the average traveled distance for several known \RENDEZVOUS algorithms. We consider the \emph{Center Of Gravity algorithm}~\cite{DBLP:journals/siamcomp/SuzukiY99}, the two-color \ASYNC algorithm (\emph{Her2}) by Heriban et al.~\cite{DBLP:conf/icdcn/HeribanDT18}, the two-color algorithm (\emph{Vig2}) by Viglietta~\cite{DBLP:conf/algosensors/Viglietta13}, which is known to solve \RENDEZVOUS in \SSYNC, and \CONVERGENCE in \ASYNC, the three-color algorithm (\emph{Vig3}) by Viglietta~\cite{DBLP:conf/algosensors/Viglietta13}, the four-color algorithm (\emph{Das4}) by Das \emph{et al.}~\cite{DBLP:conf/icdcs/DasFPSY12,DBLP:journals/tcs/0001FPSY16}.
We also investigate the algorithms assuming unreliable compasses by Izumi \emph{et al.}~\cite{DBLP:journals/siamcomp/IzumiSKIDWY12}: the \SSYNC static-error compass algorithm (\emph{Stat \SSYNC}), which, despite its name, works in \ASYNC, the \SSYNC dynamic-error compass algorithm (\emph{Dyn \SSYNC}), which does not work in \ASYNC, and the \ASYNC static-error compass algorithm (\emph{Dyn \ASYNC}).

We take advantage of the modularity of our simulator. The \texttt{Robot} class now carries several new properties: \texttt{color}, the color a robot presently displays ; \texttt{compass}, the type of compass and error, \emph{i.e.} 'none', 'static' or 'dynamic' ; \texttt{compass\_error}, the maximum error allowed for the compass ; and \texttt{compass\_offset}, the current compass error. The color is changed at the end of the \texttt{COMPUTE} method. Depending on the value of \texttt{compass}, \texttt{compass\_offset} is either chosen during the initialization, or at the beginning of every \texttt{LOOK} method.

Each algorithm is first carefully analyzed on paper to find the worst possible execution. Simulations are then run according to the aforementioned protocols. Due to limitations described in Section~\ref{sssec:NRN}, we actually assess those protocols for a degraded notion of \CONVERGENCE rather than \GATHERING. The distance traveled is expressed relatively to the initial distance between the two robots. In practice, the first robot is always located at $\{0,0\}$ and the second robot is placed at random on the circle of radius 1 centered on $\{0,0\}$. Algorithms are tested with no initial pending moves, as arbitrary pending moves would render fuel efficiency mostly impossible to reliably monitor.

Results are summed up in Table~\ref{table_res}. The red color denotes cases where the simulation was stuck in non-gathered cycles, and had to be manually unstuck. Details as to why this happened are provided below.

For scale, running 4 instances of \emph{Vig3} for one hour under the \ASYNC scheduler resulted in $\simeq$ 14 million total individual executions.
\clearpage
\begin{table}[htb]
    \centering
    \begin{subfigure}[b]{\linewidth}
        \centering
        \captionsetup{justification=centering}
        \includegraphics[width=\linewidth]{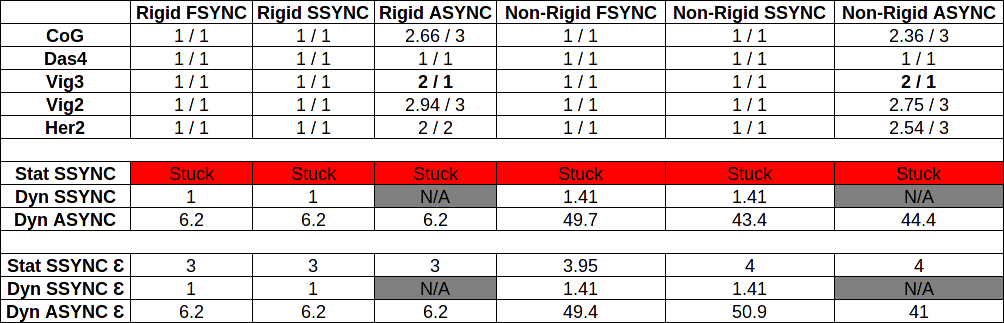}
        \caption{Maximum Traveled Distance \\ Found / Predicted}
    \end{subfigure}
    \begin{subfigure}[b]{\linewidth}
        \centering
        \captionsetup{justification=centering}
        \includegraphics[width=\linewidth]{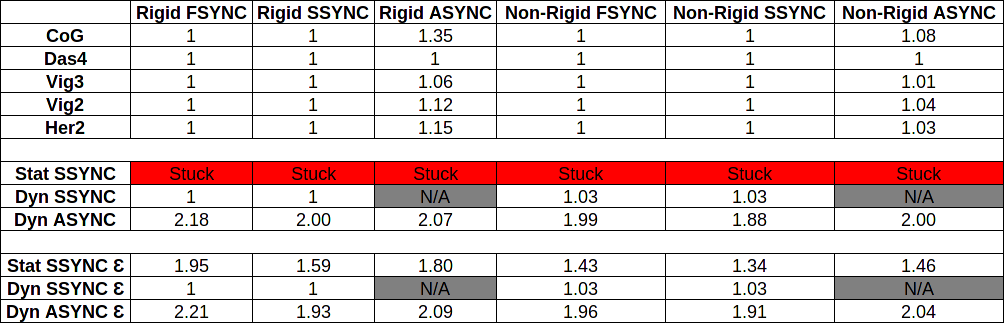}
        \caption{Average Traveled Distance}
    \end{subfigure}
    \caption{Maximum and Average Traveled Distances}
    \label{table_res}
\end{table}

While most results match the predictions, our pen and paper analysis missed a worst case execution for \ASYNC \emph{Vig3}, which was found by the simulator (highlighted in bold in Figure~\ref{table_res}). This highlights the difficulty of manually finding the maximum distance even with simple algorithms and settings.
It should be noted that rigid motion yields worst results than non-rigid. This is normal because increasing the traveled distance relies on picking a target outside of the $[r_1,r_2]$ segment, and when this is the case, performing the full motion increases the traveled distance more than performing it partially. Thus, unless stated otherwise, all further simulations assume rigid motion.

The difference between \SSYNC and \ASYNC with respect to efficiency becomes apparent, as under the \ASYNC scheduler, optimal fuel consumption mandates using four colors, while a simple oblivious algorithm is sufficient in \SSYNC.

The algorithms using compasses yield the most interesting results. First, numerous simulations of the \SSYNC static algorithm became stuck.
These failures are due to the fact that the sine and cosine operations used in the algorithms tend to sum errors, and there is a possibility that a robot moves in a way that results in an angle of exactly 0, which actually randomly yields an angle of either $0-\epsilon$ or $0+\epsilon$, where $\epsilon$ is a very small positive number. This in turn results in unsolvable cycles that prevent \CONVERGENCE. As $\epsilon$ was never larger than $10^{-9}$, we chose to prevent this behavior by slightly enlarging the interval of the condition that should be triggered on an angle of zero to an angle in $[-10^{-6},10^{-6}]$. We do the same for all conditions for consistency. So any condition that should be true for angles in $[A,B[$ are now true for angles in $[A-10^{-6},B-10^{-6}[$, in $[A,B]$ now in $[A-10^{-6},B+10^{-6}]$, in $]A,B]$ now in $]A+10^{-6},B+10^{-6}]$ and in $]A,B[$ now in $]A+10^{-6},B-10^{-6}[$.

Interestingly, this new condition only had notable impact on the static error algorithm. Indeed, these errors could be seen as small dynamic random angle errors. Since the static error algorithm is not designed to be resilient against dynamic errors, it fails whenever they appear. This also demonstrate the resilience of the dynamic error algorithms.

\subsection[\textsf{Convergence} For \textit{n} Robots]{\CONVERGENCE For \textit{n} Robots}

\noindent Cohen and Peleg~\cite{DBLP:journals/siamcomp/CohenP05} proved the Center of Gravity (CoG) algorithm solves \CONVERGENCE for $n$ robots under the \ASYNC scheduler. We analyze the fuel consumption of the algorithm under both the \SSYNC and \ASYNC schedulers. Results for the minimum, maximum, and average distance traveled are show in table~\ref{NCoG}. We use the sum of the distances to the CoG in the initial configuration as a baseline unit of distance, \emph{i.e.} the distance traveled in \FSYNC.

\begin{table}[htb]
    \centering
    \captionsetup{justification=centering}
    \includegraphics[scale=0.39]{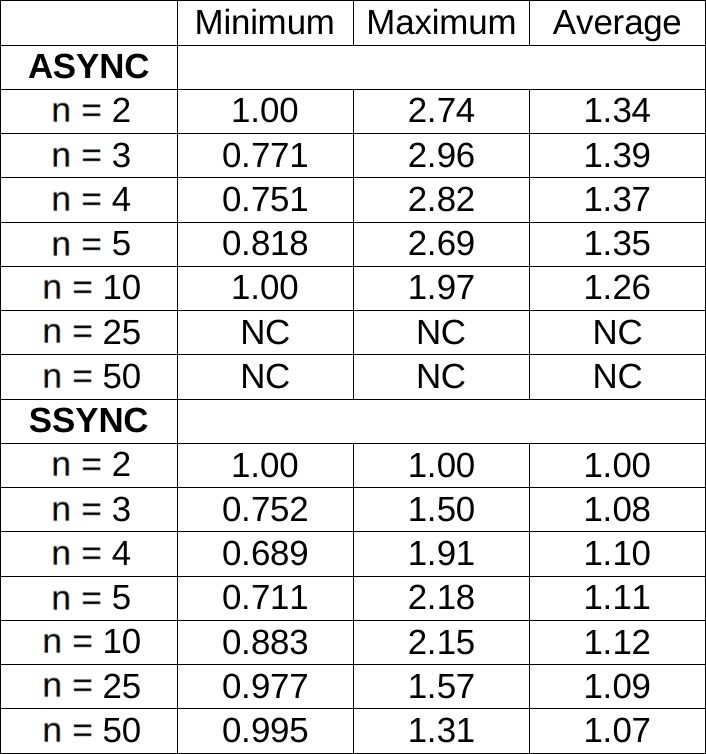}
    \caption{Traveled Distances for CoG}
    \label{NCoG}
\end{table}

It should be noted that, while previous results are based on at least hundreds of thousands of simulations, due to the increase in simulation complexity, in \ASYNC, for $n=25$, only 31 simulations could be computed under an hour. So they were discarded. Similarly, for $n=50$, no simulation could be finished under an hour.

Looking at the results, one element immediately jumps out: for $n \geq 3$, the CoG algorithm wastes movements. This is easy to understand: robots move towards the center of gravity, which for 3 or more robots is different from the geometric median (\emph{a.k.a.} the Weber point), which would actually minimize movement. Our tests seem to indicate that aiming for the median instead of the CoG can reduce traveled distance by up to 30\%. However, it is a known result that no explicit formula for the geometric median exists.
As a result, in practice, when trying to minimize traveled distance, \CONVERGENCE for $n$ robots should rely on an approximation of the geometric median rather than the center of gravity.

\section{Analyzing Algorithms in Realistic Settings}
\label{chap:realistic}

\noindent In Section~\ref{chap:performance}, the simulation of inaccurate compasses yielded extremely interesting results. To follow this track, we now focus in this section on the setting where sensors are inaccurate. In more details, we analyze the Center of Gravity (CoG) algorithm for \RENDEZVOUS in this setting, as well as the \GEOLEADEL algorithm by Canepa and Gradinariu Potop-Butucaru~\cite{DBLP:conf/sss/CanepaP07}.

\subsection{Visibility Sensor Errors}

\noindent To study the impact of inaccurate sensors, we consider three different models for vision error. For a robot $r_1$ looking at a robot $r_2$ located in $(x,y)$ in the Cartesian coordinate system centered at $r_1$, and located at $(r,\theta)$ in the polar coordinate system centered at $r_1$, we define:
\begin{itemize}
    \item The \emph{absolute} error model~\cite{DBLP:journals/automatica/Martinez09} uses a constant value $err$. A first number $R_{err}$ is picked uniformly at random in $[0,err]$, and a second $\theta_{err}$ in $[0,2\pi]$. The perceived position of $r_2$ is then $(x+R_{err} cos(\theta_{err}),y+R_{err} sin(\theta_{err}))$.
    \item The \emph{relative} error model~\cite{DBLP:journals/siamcomp/CohenP08} uses two constants $err_{dist}$ and $err_{angle}$. Two numbers $R_{err}$ and $\theta_{err}$ are picked uniformly at random in $[-err_{dist},err_{dist}]$ and $[-err_{angle},err_{angle}]$. The polar coordinates of $r_2$ are then perceived to be $(r + r*R_{err}, \theta + \theta_{err})$
    \item The \emph{absolute-relative} error model is similar to relative error, but the perceived polar coordinates are $(r + R_{err}, \theta + \theta_{err})$
\end{itemize}

These error models are depicted in Figure~\ref{fig:errors}.

\begin{figure}[htb]
    \centering
    \begin{subfigure}[b]{\linewidth}
        \centering
        \includegraphics[width=0.6\textwidth]{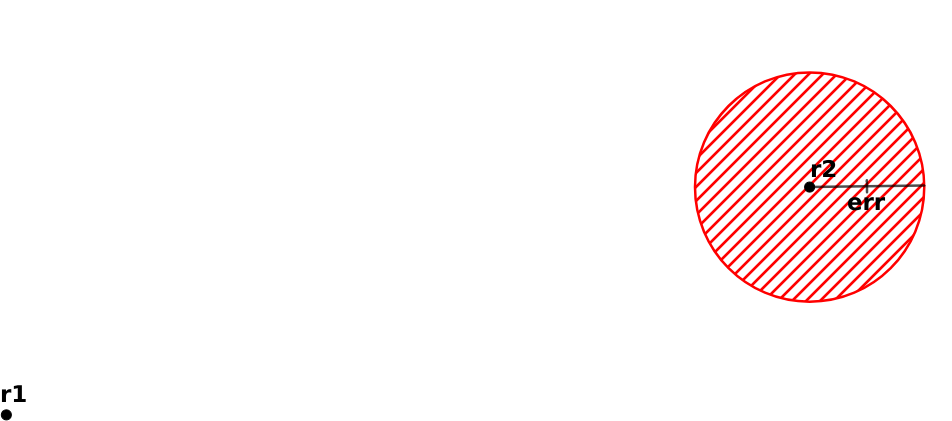}
        \caption{Absolute error}
        \label{fig:err_abs}
    \end{subfigure}

    \begin{subfigure}[b]{\linewidth}
        \centering
        \includegraphics[width=0.6\textwidth]{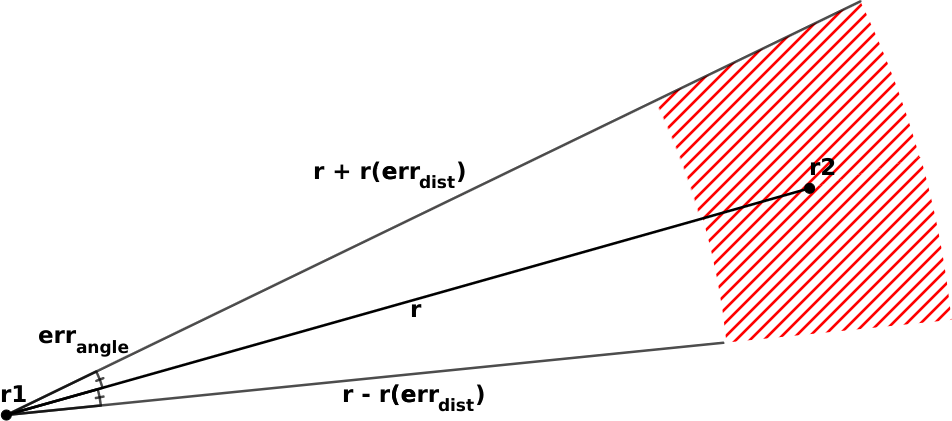}
        \caption{Relative error}
        \label{fig:err_rel}
    \end{subfigure}

    \begin{subfigure}[b]{\linewidth}
        \centering
        \includegraphics[width=0.6\textwidth]{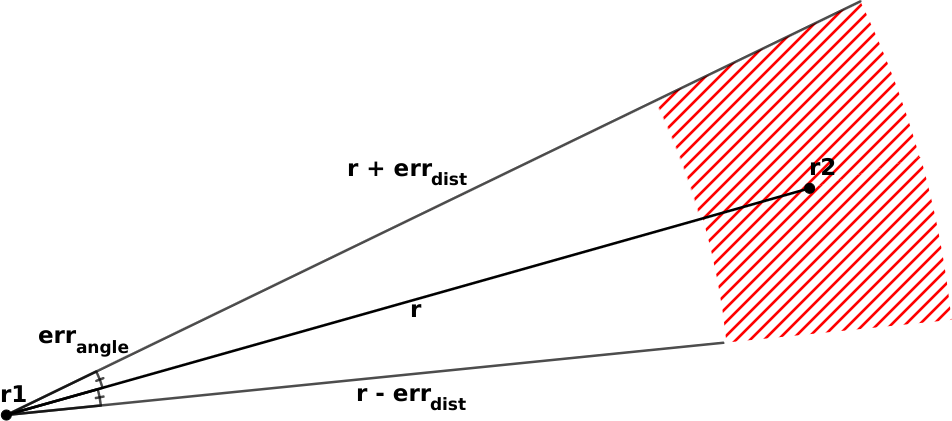}
        \caption{Absolute-relative error}
        \label{fig:err_relabs}
    \end{subfigure}
    \captionsetup{justification=centering}
    \caption{Types of Errors \\ The $r_2$ point is the actual location of robot $r_2$. The red hashed area represents possible detected positions by robot $r_1$.}
    \label{fig:errors}
\end{figure}

It should be noted that each model could be used to accurately model errors for different types of sensors.
The absolute error model is interesting because it is simple to compute, requires no change of coordinate system, uses a single parameter, and closely matches the behavior of robots where the \LOOK phase is an abstraction of GPS-type coordinates exchanges~\cite{DBLP:journals/jnw/YaredDIW07}.
The two relative models are more complex from a computing perspective, but closely match the use of either computer vision or telemetry sensors. Both carry an angular error matched with either proportional or absolute distance error. Which type of distance error is more appropriate would depend on the exact type of sensor.

These new error models drive adding three properties to the \texttt{Robot} class:

\begin{itemize}
    \item \texttt{LOOK\_error\_type}, a string that defines the type of error and can be either \texttt{'none'}, \texttt{'relative'}, \texttt{'absolute'}, or \texttt{'abs-rel'}.
    \item \texttt{LOOK\_distance\_error}, a float that matches either $err$ or $err_{dist}$, depending on the type of error.
    \item \texttt{LOOK\_angle\_error}, a float that  matches $err_{angle}$.
\end{itemize}

Robots then chose the corresponding error (with parameters chosen uniformly at random) when performing their \LOOK operation.

\subsection[\textsf{Convergence} for \textit{n}=2 Robots]{\CONVERGENCE for \textit{n}=2}

\noindent \CONVERGENCE with vision error using the CoG algorithm has already been studied by Cohen and Peleg~\cite{DBLP:journals/siamcomp/CohenP08}. The error model they considered is identical to our relative error model. Their paper states that \CONVERGENCE with distance error using the CoG algorithm is impossible in the general case. This is, however, only true for $n\geq3$, which the authors omit to mention. In the case $n=2$, it appears to be theoretically impossible to make the algorithm diverge for a distance error smaller than a $100\%$, or $err = 1$. We can reasonably ignore the case of an error greater than $100\%$, as it would allow for a robot to perceive another one directly behind itself.

To our knowledge, no formal result exists regarding the angle error. In theory, the maximum angle error is $\pi$. We simulate \CONVERGENCE for $n=2$ robots using the CoG algorithm for the relative error model. The error for each robot is chosen uniformly at random at the beginning of the execution.  

\begin{figure}[htb]
    \centering
    \begin{subfigure}[b]{0.49\textwidth}
        \centering
        \includegraphics[scale=0.28]{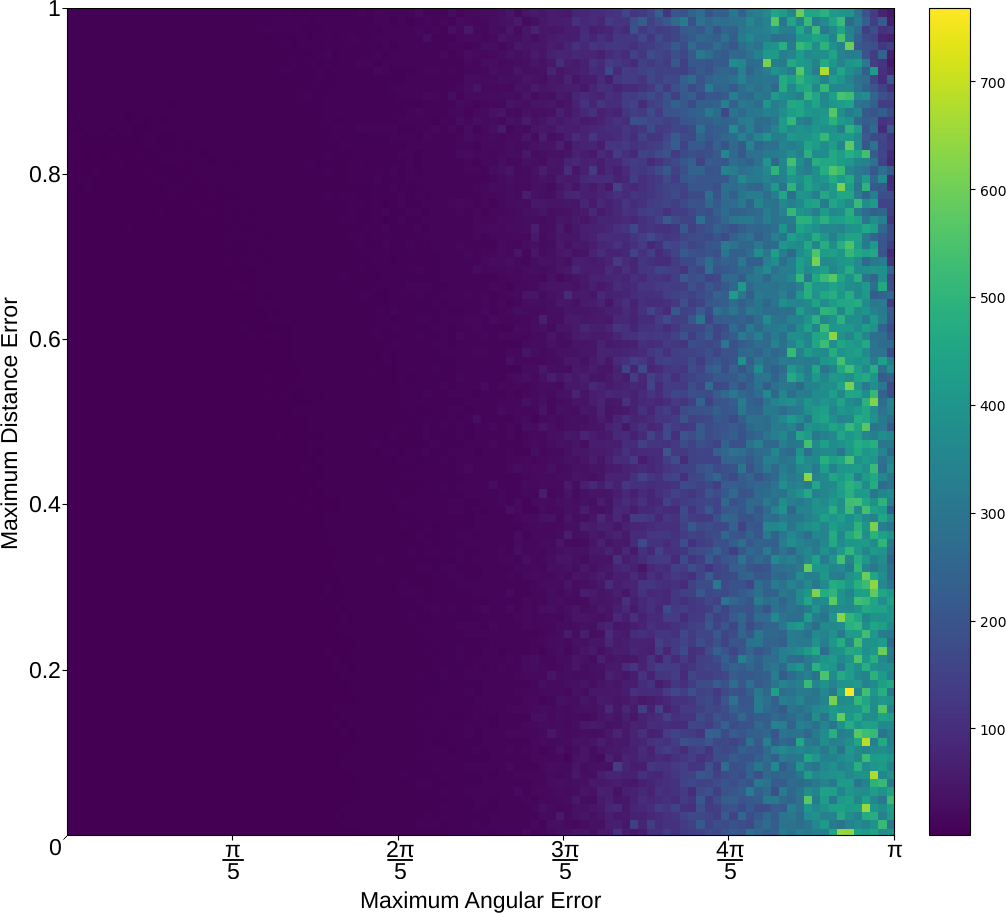}
        \caption{Maximum Traveled Distance}
        \label{fig:Max_dist}
    \end{subfigure}
    \hfill
    \begin{subfigure}[b]{0.49\textwidth}
        \centering
        \includegraphics[scale=0.28]{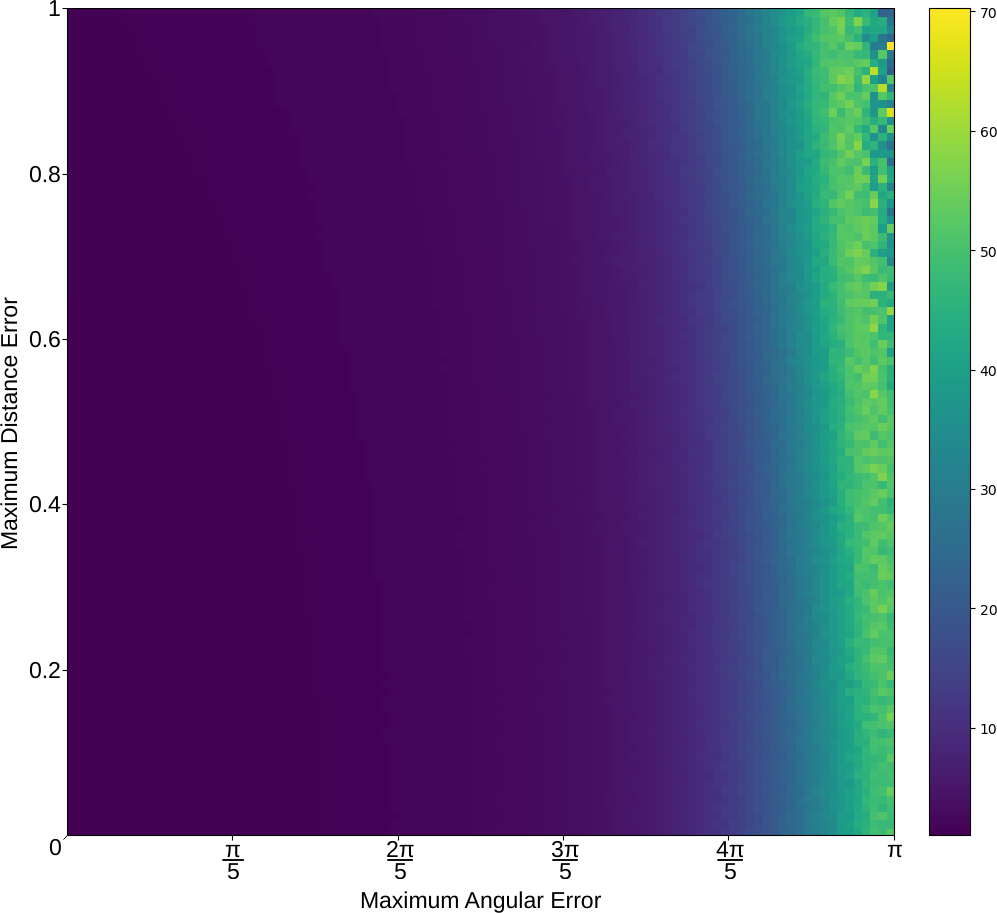}
        \caption{Average Traveled Distance}
        \label{fig:Avg_Dist}
    \end{subfigure}
    \begin{subfigure}[b]{0.49\textwidth}
        \centering
        \includegraphics[scale=0.28]{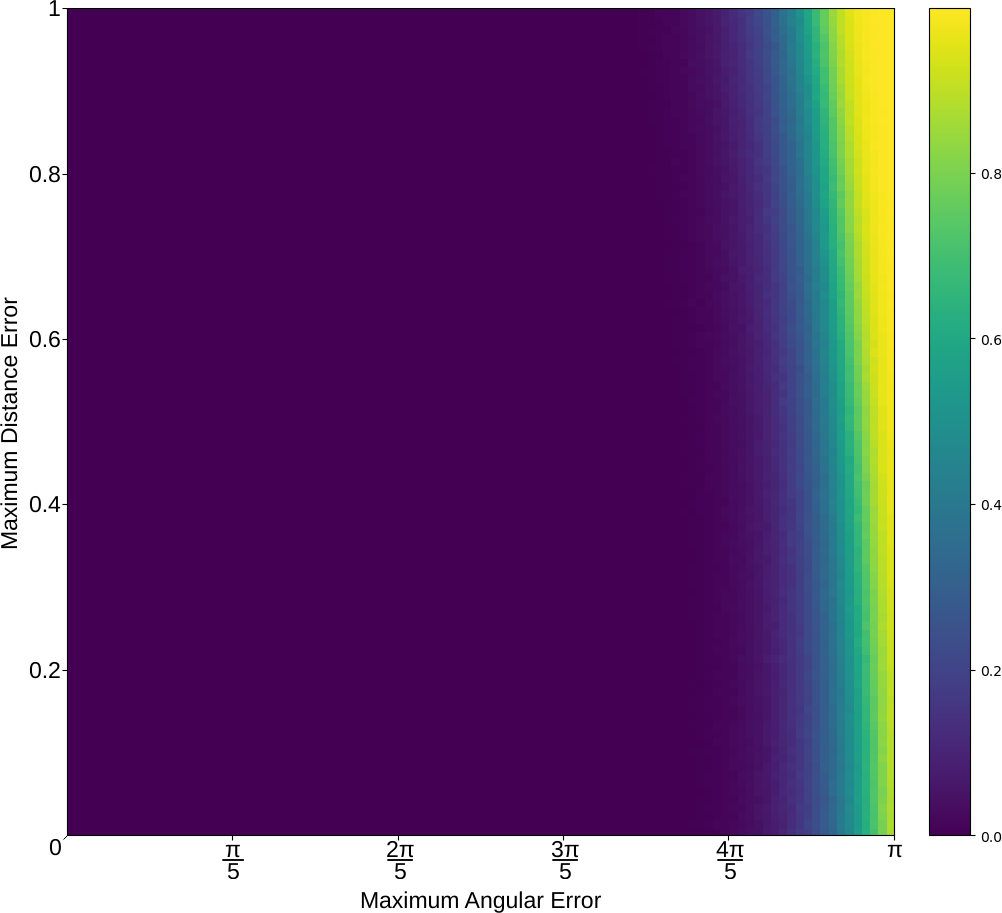}
        \caption{Proportion of Diverging Executions}
        \label{fig:Fail_exec}
    \end{subfigure}
    \captionsetup{justification=centering}
    \caption{Movement and Divergence of the CoG Algorithm for Two Robots with Inaccurate Visibility Sensors}
    \label{fig:CoG_Dist_Ang}
\end{figure}

We must also consider the now possible case of a diverging algorithm. Since the execution is random, any setting should \emph{eventually} converge. However, we must put a reasonable stopping condition in case the execution is clearly diverging. We chose to activate the defeat condition if the distance between the two robots becomes ten times larger than the distance in the initial configuration.

Note that the apparent decrease in maximum and average traveled distance for higher angle error is most likely due to the increase of diverging executions (fewer executions converge, but the traveled distance for those is shorter).

It appears clearly that the angular error has a much greater potential for both preventing\linebreak\CONVERGENCE, and making robots waste fuel. Indeed, when the angular error remains below $3\pi/5$, a distance error up to 100\% can be tolerated with no performance loss. \\To give some perspective, the realistic setting of a $10\%$ vision error with a $1^\circ$ angle error yields a maximum traveled distance of 1.221 and an average of 1.036, with no divergent executions out of more than 500 million data points. 

\pagebreak

\subsection{Compass Errors}

\noindent In the particular case of compass based algorithms of Izumi \emph{et al.}~\cite{DBLP:journals/siamcomp/IzumiSKIDWY12}, rendezvous is possible when the compasses are inaccurate. More specifically, the maximum tolerated errors are $\frac{\pi}{2}$, $\frac{\pi}{4}$ and $\frac{\pi}{6}$ for the static \SSYNC, dynamic \SSYNC, and dynamic \ASYNC algorithms, respectively. In our simulation we chose static errors, for consistency, with values up to $\frac{49\pi}{100}$, $\frac{24\pi}{100}$ and $\frac{16\pi}{100}$, to avoid possible edge cases. 
Results of maximum and average traveled distances for these algorithms are detailed in Table~\ref{table_comp_err}.

\begin{table}[htb]
    \centering
    \begin{subfigure}[b]{\linewidth}
        \centering
        \captionsetup{justification=centering}
        \includegraphics[scale=0.35]{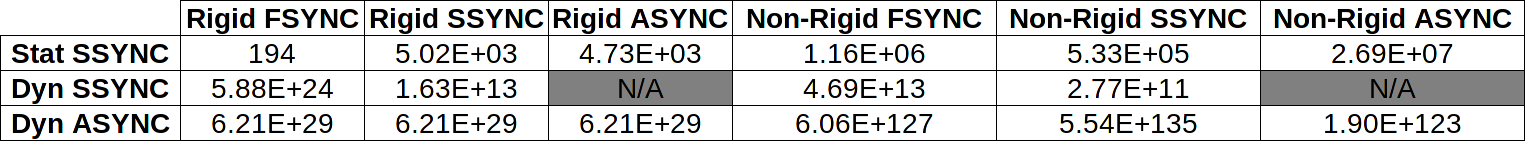}
        \caption{Maximum Traveled Distance}
    \end{subfigure}
    \begin{subfigure}[b]{\linewidth}
        \centering
        \captionsetup{justification=centering}
        \includegraphics[scale=0.35]{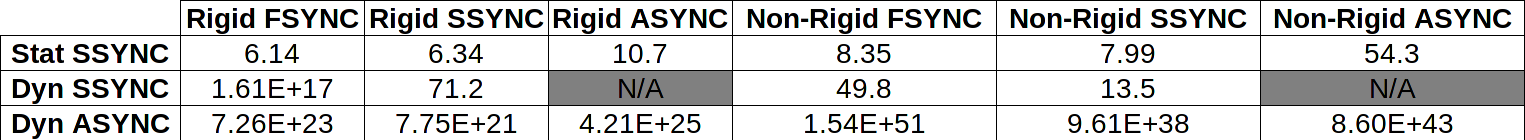}
        \caption{Average Traveled Distance}
    \end{subfigure}
    \captionsetup{justification=centering}
    \caption{Maximum and Average Traveled Distances for \RENDEZVOUS \\ with Inaccurate Compasses}
    \label{table_comp_err}
\end{table}

We observe that the unreliable compasses are used in a way that makes robots rotate around each other until they are oriented in such a way that one robot moves while the other stays, regardless of the error. However, there are no provisions in these algorithms to limit distance increases during the rotating phases, which explains the results. Detailed observation shows the distance between the two robots can gradually diverge towards infinity during rotation and then converge to zero in a single cycle. This also demonstrated a problem for our \CONVERGENCE criterion: robots could converge at rather large coordinates such that the coordinates of robots are in succession, but, since the accuracy of floating point numbers decreases as the number increase, the distance between the two robots was greater than $10^{-10}$. As a result, we modified the criterion to $|r_1r_2|<max(10^{-10},|Or_1|*10^{-10})$, with $O$ the point of coordinates $\{0,0\}$.

\subsection[\textsf{Geoleader} \textsf{Election}]{\GEOLEADEL}

\noindent Let us now consider \GEOLEADEL algorithm \ref{algCan3} by Canepa and Gradinariu Potop-Butucaru~\cite{DBLP:conf/sss/CanepaP07}, for $n=3$.
Looking at our previous results from Section~\ref{sssec:NRN}, we notice that the borders between each zone should be an issue for imperfect sensors, as different errors for different robots may lead to robots electing different \robstate{Leader} robots.

\begin{figure}[htb]
    \centering
    \includegraphics[width=0.6\linewidth]{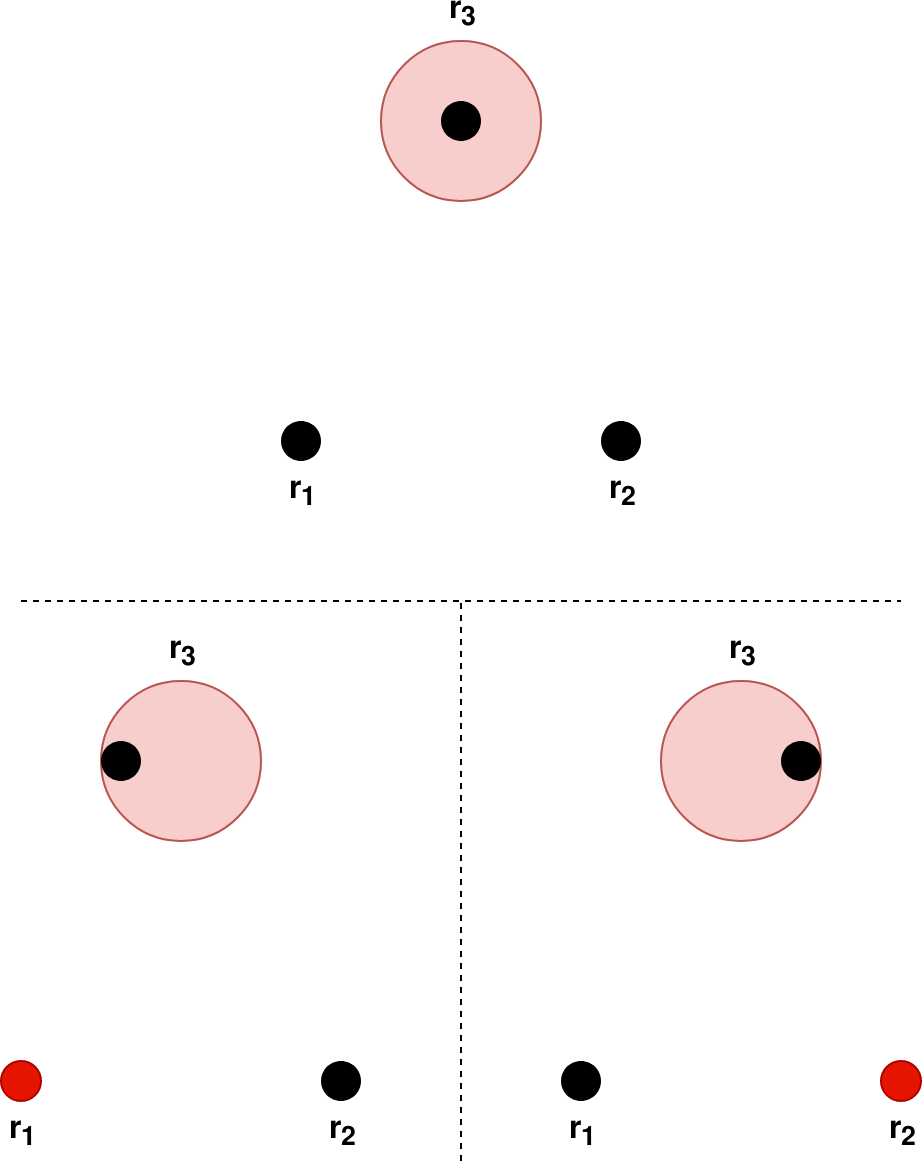}
    \caption{Example of \LEADEL Failure Due to Imperfect Vision}
    \label{fig:ex1}
\end{figure}

We demonstrate how this phenomenon can occur in Figure~\ref{fig:ex1} for the case of absolute vision error. On top is the actual configuration, where angles $\widehat{r_1r_2r_3}$ and $\widehat{r_2r_1r_3}$ are equal\footnote{Because robots have no chirality, angles cannot reliably be distinguished from their opposite. So, two opposite angles may always be considered equal.}, and angle $\widehat{r_1r_3r_2}$ is smaller than both, so $r_3$ should be elected. The red circle shows the possible perceived position of $r_3$ by $r_1$ and $r_2$ due to vision error. In the bottom left case, we show a possible perception by $r_1$ where $r_1$ should be elected \robstate{Leader}, as $\widehat{r_2r_1r_3}$ is now greater than $\widehat{r_1r_2r_3}$. On the lower right, $r_2$ similarly thinks it should be elected. Now, two different robots consider themselves \robstate{Leader} and the election process fails.

We now use the absolute model to simulate \GEOLEADEL with $err = 0.001$, for $n=3$.
This simulation yields $\simeq 0.1\%$ of errors in total, where two robots compute different \robstate{Leader} robots, and is shown in figure~\ref{fig:Simu2}. 

\begin{figure}[htb]
    \centering
    \captionsetup{justification=centering}
    \includegraphics[width=0.675\linewidth]{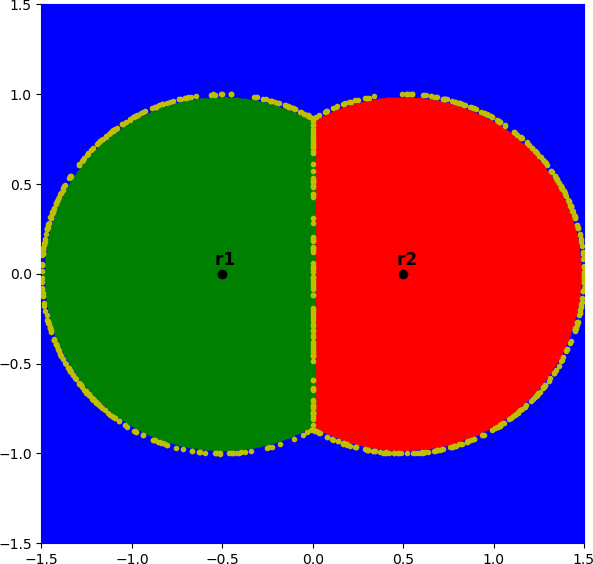}
    \caption{Simulation for 3-robot \LEADEL with Absolute Vision Error \\ Yellow points represent configurations where the error generates two different \robstate{Leader} robots.}
    \label{fig:Simu2}
\end{figure}

\section[Improved \textsf{Convergence} and \textsf{Leader} \textsf{Election}]{Improved \CONVERGENCE and \LEADEL for Faulty Visibility Sensors}
\label{chap:improved}

\noindent Following our observations of problematic behaviors in Sections~\ref{chap:performance} and~\ref{chap:realistic}, we provide two new algorithms: a fuel efficient \CONVERGENCE algorithm for two robots, and a \GEOLEADEL algorithm that is resilient to faulty visibility sensors.

\subsection[Fuel Efficient \textsf{Convergence}]{Fuel Efficient \CONVERGENCE}

\noindent We provide a new algorithm (Algorithm~\ref{alg:efficient}) for the \ASYNC \CONVERGENCE of two robots. Our algorithm is a simplified version of the two-color algorithm by Viglietta~\cite{DBLP:conf/algosensors/Viglietta13}, which does \emph{not} solve \GATHERING (while Viglietta's algorithm does solve \GATHERING in \SSYNC). Our algorithm however ensures that no target can ever be outside of the segment between the two robots, ensuring no wasted moves, and that there exists a scheduling such that convergence is eventually achieved. It is denoted by \textsc{FEC} (Fuel Efficient \CONVERGENCE, presented in Figure~\ref{fig:Efficient2}). Our algorithm still uses two colors (\Black and \White), and when observing the other robot's color, the observing robot either remains still (the 'Self' target) or goes to the computed midpoint between the two robots (the 'Midpoint' target), possibly switching its color to the opposite one.

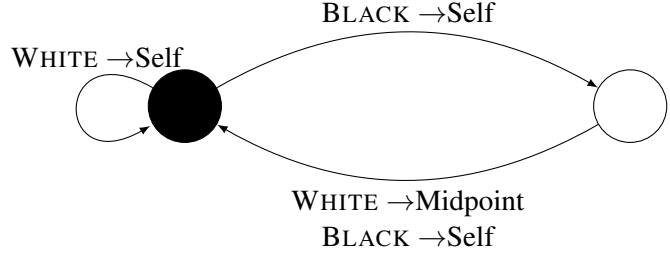
\begin{figure}[htb]
	\centering
	\begin{tikzpicture}
		\node[blk] (B) {};
		\node[wht] (W) [right of=B] {};	
		\path[->] (B) edge[bend left] node[above,align=center]{\Black$\rightarrow$Self} (W);
		\path[->] (W) edge[bend left] node[below,align=center]{\White$\rightarrow$Midpoint \\ \Black$\rightarrow$Self} (B);
		\path[->] (B) edge[out=150,in=210,loop] node[near start,above,align=right]{\White$\rightarrow$Self} (B);
	\end{tikzpicture}
	\caption{FEC: Fuel Efficient \CONVERGENCE Algorithm for Two Robots}
	\label{fig:Efficient2}
\end{figure}

\begin{algorithm}
\caption{FEC: Fuel Efficient \CONVERGENCE Algorithm for Two Robots} 
\label{alg:efficient}         
    \begin{algorithmic}                   
        \STATE 
        \IF{me.color = \White}
            \STATE me.color $\Leftarrow$ \Black
            \IF{other.color = \White}
    		    \STATE me.destination $\Leftarrow$ other.position/2
    		\ENDIF
        \ELSIF{me.color = \Black}
        	\IF{other.color = \Black}
            	\STATE me.color $\Leftarrow$ \White
    		\ENDIF
        \ENDIF
    \end{algorithmic}
\end{algorithm}

As a sanity check, we ran this algorithm through our simulator for one hour ($\simeq 30$ million data points) under a randomized \ASYNC scheduler and could not find a single execution where the traveled distance was greater than the initial distance.

\begin{theorem}
The Fuel Efficient \CONVERGENCE Algorithm (\ref{alg:efficient}) guarantees the distance traveled for \CONVERGENCE is never greater than the initial distance between the two robots under the \ASYNC scheduler, assuming no pending moves in the initial configuration.
\end{theorem}

\begin{proof}
First, we see that to achieve \CONVERGENCE with an optimal distance, robots should always be moving towards each other. So, for robots to converge using more than the initial distance, it is required that, at one point in the execution, one robot moves \emph{not towards} the other robot.
We note that a network of two disoriented robots can be simplified as a line. In that sense, the only movement that can increase the maximum \CONVERGENCE distance is when a robot moves opposite the other robot. In other words, when robots 'switch sides'.
Let us now prove that no robot can target a robot while it is in its \MOVE phase: 
Only the $\{$\WHITE,\WHITE $\}$ snapshot can trigger a \MOVE phase. Since this transition implies a change of color to \BLACK at the end of the \COMPUTE phase, robots that move can only be \BLACK.
So, if a robot is moving, it is \BLACK and the other robot, regardless of color, cannot start moving because its snapshot is different from $\{$\WHITE,\WHITE $\}$.
Furthermore, because robots switch to \BLACK after moving, and can only switch to \WHITE if the other robot is \BLACK, no robot can execute multiple \MOVE in sequence unless the other robot has executed at least a full cycle in between. So a robot cannot move multiple times while the other has pending moves.
We look at what happens after each robot completes at least one full cycle. We assume $r_1$ performs a \LOOK, and $r_2$ performs $k$ cycles before $r_1$ finishes its \MOVE. The distance after $r_1$ finishes its cycle is presented in table~\ref{tab:movCycles-bis}.

\begin{table}[htb]
\centering
\begin{adjustbox}{max width=\textwidth}
\begin{tabular}{|c|c|c|}\hline
& $r_1$ has a pending \STAY & $r_1$ has a pending \HALF \\ \hline
$r_2$ executes $k$ \STAY & $X$ & $\left[\dfrac{X}{2} , X - \delta \right]$  \\ \hline
$r_2$ executes $1$ M2H\footnotemark[2]& $\left[ \dfrac{X}{2} , X - \delta \right]$ & $\left[ 0 , X - 2\delta \right]$ \\ \hline

\end{tabular}
\end{adjustbox}
\caption{Distance after a full cycle of $r_1$ and $k$ full cycles of $r_2$ with an initial distance of $X$}
\label{tab:movCycles-bis}
\end{table}

\footnotetext[2]{As explained above, moving a second time requires at least a full cycle from the other robot.}

In the case of simultaneous \HALF, the distance can be reduced down to zero, but robots cannot switch sides.
In both other cases where a \MOVE happens, the distance is reduced at most down to half, and robots cannot switch sides.
Overall, in no cases can the robots move not towards one another, so the maximum distance traveled is always the initial distance between the two robots.
\end{proof}

However, while the randomized scheduler we use for the simulator ensures convergence is always achieved, a rapid analysis of the algorithm shows that this algorithm ensures fuel efficiency, but does not actually ensures convergence. In fact, a simple \SSYNC scheduling can infinitely prevent robots from moving. This further highlights that simulations and formal proofs are complementary. We conjecture that Fuel Efficient Convergence is not actually possible for two colors, and that algorithms using three colors may even yield Fuel Efficient Rendezvous (not just Convergence).

We also compare the resilience of this algorithm against vision errors with the center of gravity algorithm as a baseline in Figures~\ref{MAX_err} and~\ref{AVG_err}. Our results show that this algorithm is slightly more resilient to vision errors than CoG. 

\begin{figure}[htb]
    \centering
    \begin{subfigure}[b]{0.45\textwidth}
        \centering
        \includegraphics[width=\textwidth]{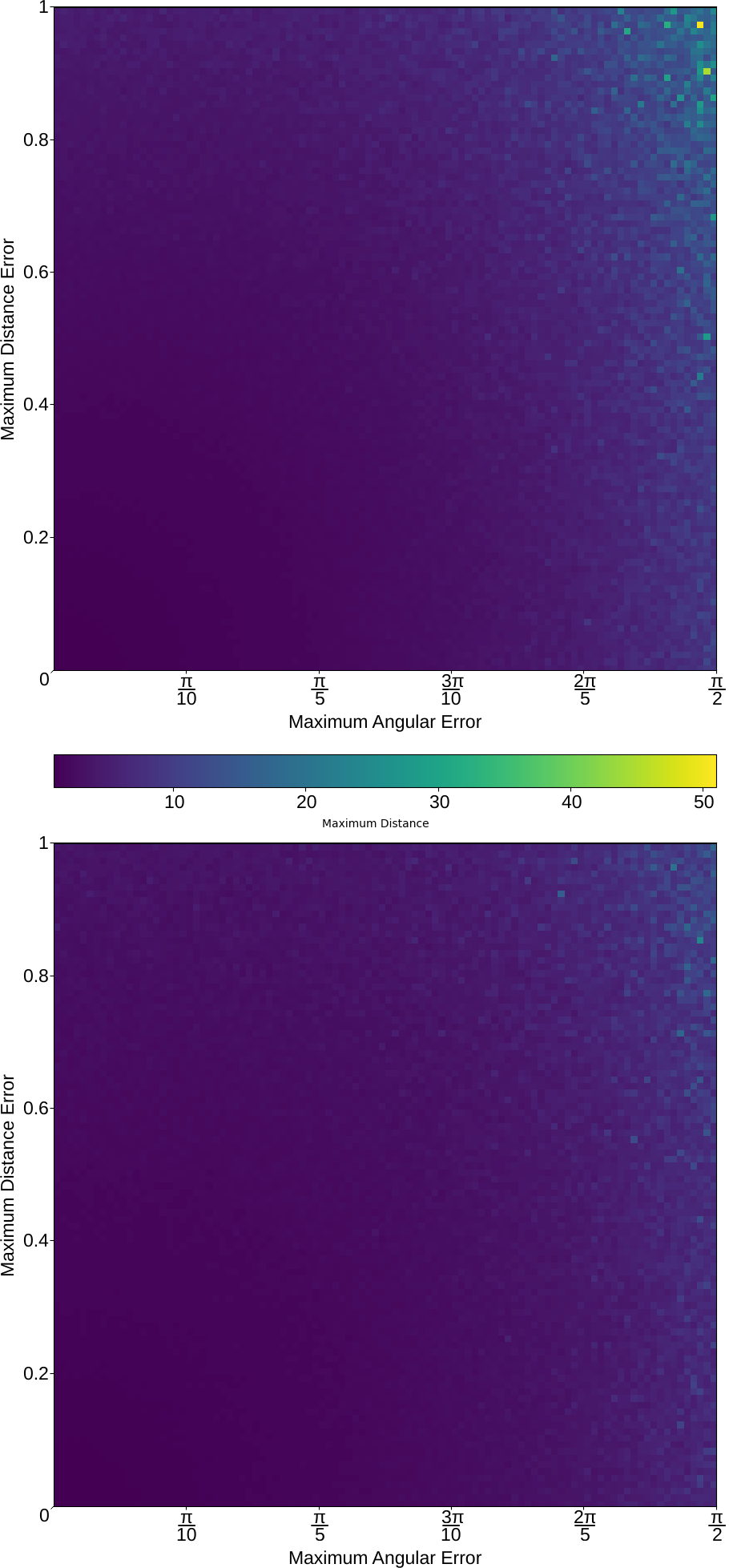}
        \captionsetup{justification=centering}
        \caption{Maximum Distance Traveled by \textsc{CoG} (top) and \textsc{FEC} (bottom)}
        \label{MAX_err}
    \end{subfigure}
    \hfill
    \begin{subfigure}[b]{0.45\textwidth}
        \centering
        \includegraphics[width=\textwidth]{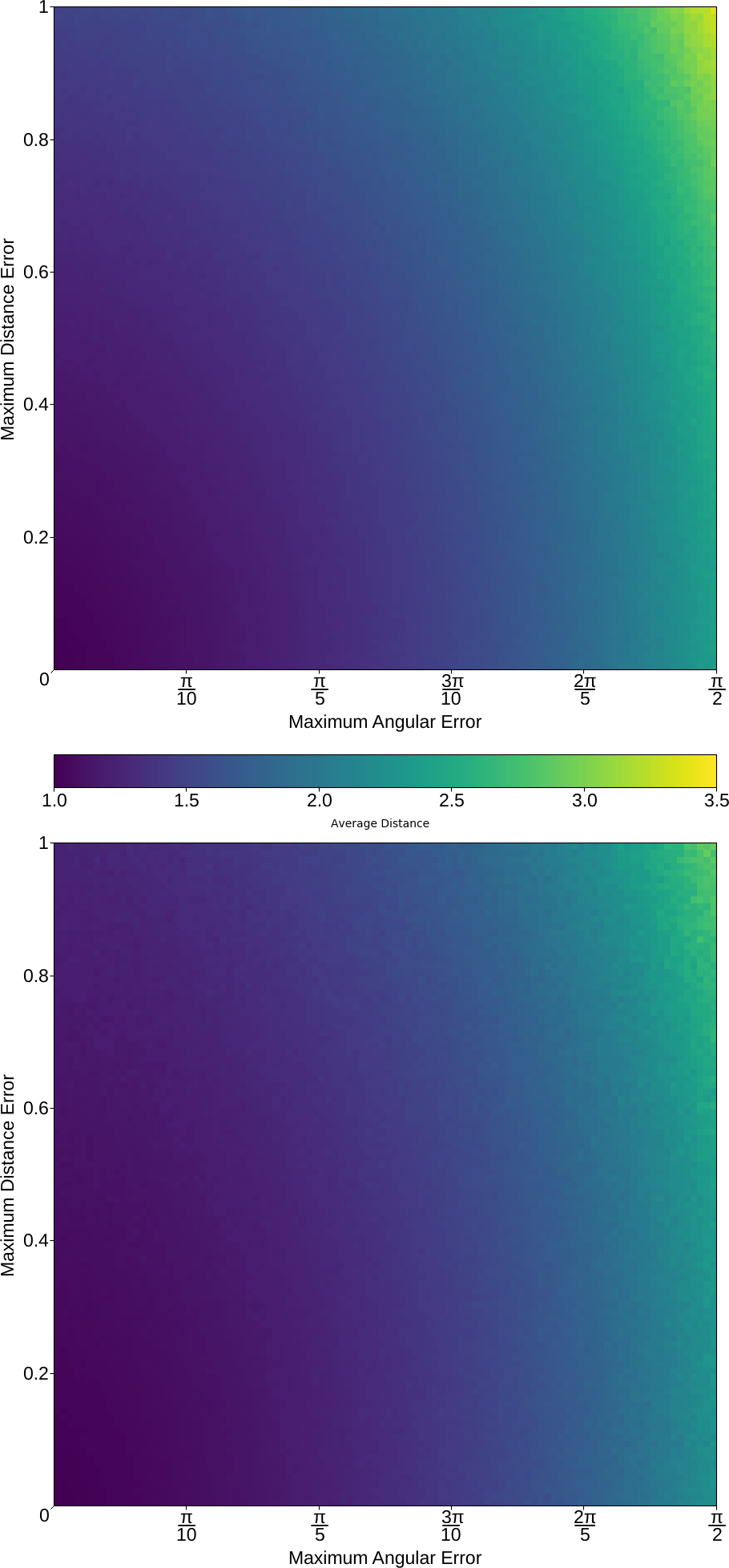}
        \captionsetup{justification=centering}
        \caption{Average Distance Traveled by \textsc{CoG} (top) and \textsc{FEC} (bottom)}
        \label{AVG_err}
    \end{subfigure}
\end{figure}

\subsection[Error Resilient \textsf{Geoleader} \textsf{Election}]{Error Resilient \GEOLEADEL}
\label{ssec:errLead}

\noindent The \GEOLEADEL algorithm by Canepa and Gradinariu Potop-Butucaru~\cite{DBLP:conf/sss/CanepaP07} was \emph{not} designed under the assumption that the visibility sensors could be prone to errors. In this subsection, we use this awareness to create a new, error-resilient, version of this algorithm, using our simulation framework.

\paragraph{\textsf{Geoleader} \textsf{Election} for Four Robots}
One intuitive way of building a fully resilient algorithm for \LEADEL could be based on robots computing the bounds of the error zone. While this seems feasible for a 3-robot election, it becomes far less trivial for four robots or more. We present the results of a leader election for four robots in the appendix.

\paragraph{Proposed Algorithm}
In section~\ref{chap:realistic}, we used the simulation framework to detect failed elections caused by visibility sensor errors. Since mobile robots are able to run any algorithm during their \COMPUTE phase, then they can also run the simulation framework to do precisely that.

The improved algorithm relies on the knowledge of the vision error model and its upper bounds to simulate random errors in a robot's position and snapshot and determine whether there exists a possibility of the other robots electing different \robstate{Leader} robots.

Note that absolutely knowing that the election cannot fail (\emph{i.e.}, the election cannot yield two different \robstate{Leader} robots for two different robots) would require checking the entire surface of possible errors, which is not feasible in practice. So, we assume that robots perform a finite number of trials and decide accordingly.

Each robot internally simulates a position error for each robot in its snapshot within the known margins, performs a simulated election for each robot in its snapshot, and checks for discrepancies in the resulting \robstate{Leader} robots. This is repeated with new random errors for a given number of tries, similar to a Monte-Carlo approach.

Once a robot believes the election process can succeed, it chooses the \robstate{Leader} as in the original algorithm.
Otherwise, it picks a random direction and distance, and performs a \MOVE to "scramble" the network.
This process repeats until all robots believe the election can succeed.
This process is detailed in Algorithm~\ref{algR}.

\begin{algorithm}
\caption{Reliable \LEADEL algorithm} 
\label{algR}
\begin{algorithmic}
\STATE $L = self.$COMPUTE$('Leader Election')$
\STATE $my\_network = self.snapshot \cup self$
\STATE $counter = 0$
\WHILE{$counter < nb\_tries$}
\FOR{$r_1$ in $my\_network$}
\STATE $r_v = r_1$
\STATE Change $r_v.x$ and $r_v.y$ randomly according to error parameters
\STATE $r_v.snapshot = my\_network/\{r_1\}$
\FOR{$r_2$ in $r_v.snapshot$}
\STATE Change $r_2.x$ and $r_2.y$ randomly according to error parameters
\ENDFOR
\STATE $L_v = r_v.$COMPUTE$('LeaderElection')$
\IF{$L \neq L_v$} 
    \STATE Move randomly
    \STATE Exit
\ENDIF
\ENDFOR
\STATE $counter \mathrel{+}= 1$
\ENDWHILE
\STATE L is elected \robstate{Leader}
\end{algorithmic}
\end{algorithm}

\noindent We now perform simulations using this algorithm. Each point is sorted according to the following:

\begin{itemize}
    \item If no robot detects a possible error, it is a valid point.
    \item If at least one robot has detected a possible error, and decided to move as a result, it is a detected possible error point.
    \item If no robot moves, but two robots have different \robstate{Leader} robots, it is an undetected error point.
\end{itemize}

We measure the proportion of undetected error and possible error points for $nb_{tries}$ between 0 and 30. Results are presented in Figure~\ref{fig:perf}.

\begin{figure}[htb]
    \centering
    \includegraphics[width=0.7\linewidth]{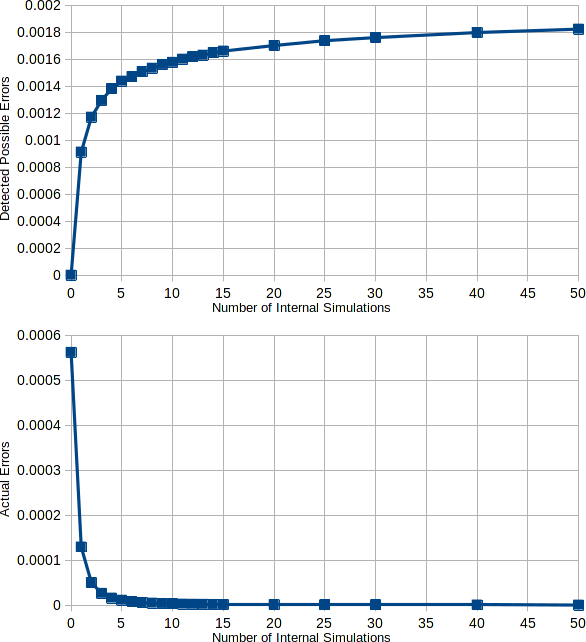}
    \caption{Performance of the Error-Resilient Election Algorithm \\ $err = 0.001$}
    \label{fig:perf}
\end{figure}

Note that the number of undetected error points, while decreasing, does not reach zero under our testing conditions. Also, using a single internal simulation typically results in a $\sim 80\%$ reduction in the number of undetected error points. Using 10 internal simulations resulted in a reduction of $99.5\%$ of undetected error points. Which number of internal simulations is the best suited would depend on both the speed of the leader election process and reliability of the obtained solution requirements.

Importantly, we notice that, were we to choose an error model and error bounds such that it models the possible errors of representing real numbers using limited precision floats, then this particular algorithm, when used with an infinitely large, similar to $\mathbb{R}^2$ number of random tries, can be made to reliably detect anomalies due to the errors induced by evolving in the continuous plane, yet only perceiving a discretized plane. Actually, me make the conjecture that this algorithm can be adapted to allow any algorithm that makes decisions based on robot locations to operate properly in a perceived discretized plane.

Furthermore, using this algorithm allows us to reduce the size of a snapshot to a finite, storable amount to realistically use the \textbf{SyncSim} protocol~\cite{DBLP:journals/tcs/0001FPSY16}, and fully simulate the \FSYNC scheduler in \luminous \ASYNC.

\section{Conclusion}
\label{sec:conclusion}

\noindent In this paper, we introduce a modular framework designed to simulate mobile robots for any given setting.

We discuss the limitations and constraints of this approach, and use it to compute the maximum distance traveled, or fuel efficiency, of multiple algorithms in several settings, with interesting results. In particular, we note that the algorithm by Izumi \emph{et al.}~\cite{DBLP:journals/siamcomp/IzumiSKIDWY12} can lead to an unbounded increase in distance before eventually gathering. Similarly, the center of gravity algorithm is inherently sub-optimal for $n>2$ robots, and robots should use an algorithm based on the geometric median instead.

We then use this framework to simulate inaccurate sensors for mobile robots and verify the behavior of \CONVERGENCE and motion based \LEADEL under this new model. We also introduce errors in the perception of colors for \luminous robots performing state-of-the-art two-robot \GATHERING.

Finally, we designed two new algorithms. The first one is designed to perform two-robot \CONVERGENCE under the \ASYNC scheduler with optimal fuel efficiency.
The second algorithm uses the simulator itself to allow robots to solve motion based \LEADEL with inaccurate sensors. The latter can be adapted to allow for decision making algorithm, such as \LEADEL, to function using discretized snapshots, and so, to use the \textbf{SyncSim} protocol to simulate the \FSYNC scheduler in \luminous \ASYNC.

Overall, this framework achieves its planned objective of being both easy to use and able to produce useful results for researchers. As a test, we timed the full implementation, and testing in \FSYNC, \SSYNC and \ASYNC, of the two color \RENDEZVOUS algorithm from Viglietta to require less than half an hour, including basic network monitoring and testing.

The source code and instructions for our simulator are provided in the appendix and at the following repository: \url{https://github.com/UberPanda/PyBlot-Sim}

\subsection*{Future Work}

\noindent As we already discussed, our simulator is modular to allow for use for any given algorithm and model. So it seems logical that it should, ideally, implement every existing model and test all major algorithms in the literature, such as mutual visibility for opaque robots.

Furthermore, while interesting for researchers, our simulator is not a tool for formal proofs. However, one could also argue that in its current state, we have not proven that the simulator actually simulates mobile robots, even within our degraded hypotheses.
We believe that the simulator itself should be formally proven to match the model of mobile robots it claims to simulate.
Note that the usefulness of this proof would be limited, as the addition of any new module may require proving the entire simulator again.

Finally, our \LEADEL algorithm for errors in vision is able to function in a continuous setting using discretized snapshots. The design philosophy behind this algorithm of using randomized tries to simulate sensor errors is not specific to the \LEADEL problem, and it could be used for other algorithms that rely on making decisions based on the locations of robots in the network and that are sensitive to errors in perception.
Building new algorithms that can use these finite snapshots allows us to use the \textbf{SyncSim} protocol~\cite{DBLP:journals/tcs/0001FPSY16} and simulate a \FSYNC scheduler in \luminous \ASYNC and would be a major advantage for resilience to asynchrony.

\printbibliography

\newpage

\appendix

\section{Appendix: example of an Instance of the Simulator}

\label{AP:pyblot}

We present a minimum working example of an instance of the simulator. It simulates executions of the Vig2 algorithm~\cite{DBLP:conf/algosensors/Viglietta13} in the standard \OBLOT model with rigid motion, under the \FSYNC scheduler. It monitors the number of cycles needed to complete degraded gathering\footnote{i.e. robots are closer than $10^{-10}$ with the initial distance being $1$.}. Results predictably show the possible need of two full cycles in the case where both robots start in the \BLACK color. For better readability, more advanced features of the simulator are not included.

\lstset{breaklines=true,
        language=Python,
        postbreak=\mbox{\textcolor{red}{$\hookrightarrow$}\space},
        frame=lines,
    	tabsize=2,
    	basicstyle={\footnotesize\ttfamily},
    	basewidth={.5em,0.55em},
    	commentstyle={\rmfamily\color{orange!60!black}\bfseries\itshape\ttfamily},
    	keywordstyle={\rmfamily\color{blue!60!black}\bfseries\ttfamily},
    	stringstyle={\rmfamily\color{green!50!black}\bfseries\itshape\ttfamily},
    	showstringspaces=false,
	    }

\begin{lstlisting}[caption={Robot Class File: Common/lib\_robot.py\label{lst:robot}}]
import Common.lib_algorithms as lib_algos

class Robot:
	def __init__(self, name, x, y, color=0):
		self.name = name
		self.x = x ## Real position in the network
		self.y = y 
				
		self.phase = 'WAITING'
		
		# Available info for COMPUTE :
		
		self.snapshot = [] ## snapshot
		self.color = color ## color
		self.target = ()


	def LOOK(self,network,scheduler):
		self.snapshot = []
		for R2 in network: 
			if self != R2:
				R2x = R2.x
				R2y = R2.y
				self.snapshot.append(Robot(R2.name, R2x, R2y, R2.color))
		
		self.phase = 'COMPUTING'



	def COMPUTE(self,algo):
		try:
			result = getattr(lib_algos, algo)(self)
		except:
			raise AttributeError('This algorithm does not exist in the current version (or you made a typo ?)')
		
		return result


	def MOVE(self):
			
		self.x = self.target[0]
		self.y = self.target[1]
			
		self.target = ()
		self.phase = 'WAITING'
\end{lstlisting}

\begin{lstlisting}[caption={Miscellaneous Functions File: Common/lib\_misc\_functions.py\label{lst:miscfunc}}]
from math import sqrt

def rob_dist(R1,R2): # Using Robots
	return dist((R1.x,R1.y),(R2.x,R2.y))

def dist(T1,T2): # USing Tuples
	return sqrt((T1[0] - T2[0])**2+(T1[1] - T2[1])**2)
\end{lstlisting}

\begin{lstlisting}[caption={Algorithms File: Common/lib\_algorithms.py\label{lst:alg}}]

def vig2(R1):
	R2 = R1.snapshot[0]
	if R1.color == 0:
		if R2.color == 0:
			R1.target = ((R1.x+R2.x)/2, (R1.y+R2.y)/2) 
			R1.color = 1
			R1.phase = 'MOVING'
		else:
			R1.target = (R2.x, R2.y) 
			R1.phase = 'MOVING'
	else:
		R1.phase = 'WAITING'
		if R2.color == 1:
			R1.color = 0
\end{lstlisting}

\begin{lstlisting}[caption={Scheduler File: Common/lib\_schedulers.py\label{lst:sched}}]


def scheduler(sched,network,algo):
	
	if sched == 'FSYNC':
		for R1 in network:
			R1.LOOK(network,'FSYNC')
		for R1 in network:
			R1.COMPUTE(algo)
		for R1 in network:
			if R1.phase == 'MOVING':
				R1.MOVE()

	else:
		raise Exception('Unknown Scheduler')

\end{lstlisting}

\begin{lstlisting}[caption={Simulation Functions File: Common/lib\_sim\_functions.py\label{lst:simfunct}}]
# None Needed for this Example
\end{lstlisting}

\begin{lstlisting}[caption={Simulation File: Vig2.py\label{lst:simulation}}]
from random import SystemRandom, choice, uniform
from math import sqrt, pi, cos, sin
from Common.lib_robot import Robot
from Common.lib_schedulers import scheduler
from Common.lib_misc_functions import rob_dist

_sysrand = SystemRandom()

######################################################

SSTAB = True
COLOR_RANGE = range(2) # Total number of colors.

######################################################

simus = 0
act_list = []

while simus < 1000000:
	simus += 1

	# Network initialization

	network = []
	network.append(Robot('r1',0,0,0))
	
	ang = uniform(-pi,pi)
	network.append(Robot('r2',cos(ang),sin(ang),0))
	
	for R1 in network:
		if SSTAB == True:
			R1.color = choice(COLOR_RANGE)

	## Beginning of the simulation
	
	steps = 0
	while True:
		
		scheduler('FSYNC',network,'vig2')
		
		## Live state monitoring
		
		steps += 1
		
		## Victory condition

		if rob_dist(network[0],network[1]) < 10**(-10):
			break
		
		## Defeat condition

	act_list.append(steps)
	
	
print('FSYNC R Vig2')
print('Min : ' + str(min(act_list)))
print('Max : ' + str(max(act_list)))
print('Avg : ' + str(sum(act_list) / len(act_list)))
\end{lstlisting}

\newpage

\section{Appendix: leader Election for 4 Robots}

Figures~\ref{fig:lead4_0} through~\ref{fig:lead4_5} show the result of attempting to elect a \robstate{Geoleader} using algorithm \ref{algCan4} by Canepa and Gradinariu Potop-Butucaru~\cite{DBLP:conf/sss/CanepaP07}, which should fail whenever two robots are identically close to the center of the smallest enclosing circle. However, as we have shown, such cases are statistically impossible with perfect sensors and simply become a small subset of the error points of error-prone sensors. 

\begin{algorithm}[H]
\caption{Original \LEADEL Algorithm by Canepa and Gradinariu Potop-Butucaru~\cite{DBLP:conf/sss/CanepaP07} for Four or More Robots}         
\label{algCan4}         
\begin{algorithmic}
\STATE Compute the smallest enclosing circle \SEC
\STATE Compute the distance $d_k$ to the center of \SEC, for all robots $1\leq k \leq n$
\IF{$d_{myself} < d_k \forall k \neq myself$, where $1\leq k \leq n$} 
    \STATE Become \robstate{Leader}
    \STATE Exit
\ENDIF
\IF{$d_{myself} \leq d_k \forall k \neq myself$, where $1\leq k \leq n$} 
    \STATE Perform a Bernoulli trial with a probability of winning of $p = \dfrac{1}{n}$
    \IF{Trial won}
    \STATE move a distance of $d_{myself} \cdot p$ towards the center of the \SEC
    \ENDIF
\ENDIF
\end{algorithmic}
\end{algorithm}

Robots $r_1$ and $r_2$ are fixed at coordinate $(-0.5,0)$ and $(0.5,0)$, respectively. Robot $r_3$ has a fixed location for each image, on a grid in the lower left quarter of the image. Symmetries of the network allow us to extrapolate results for the remainder of the positions of $r_3$. The position of $r_4$ is chosen at uniformly at random, and each point show the result for a given position. The error is absolute with $err = 0.001$.

As before, colors red, green, and blue denote that robot $r_1$, $r_2$ and $r_3$ are the chosen\linebreak\robstate{Geoleader}, respectively. Color cyan denotes that robot $r_4$ is the chosen \robstate{Geoleader}, and yellow denotes that two different \robstate{Geoleader} robots have been elected due to sensor error.
Similarly to previous simulations, each image contains one million points.

As can be seen from the following figures, computing precisely the bounds of the error zone is extremely complex in practice, as the formula would be different and more complex as the network size grows.

\begin{figure}
    \centering
    \begin{subfigure}[b]{0.49\linewidth}
        \centering
        \includegraphics[width=\textwidth]{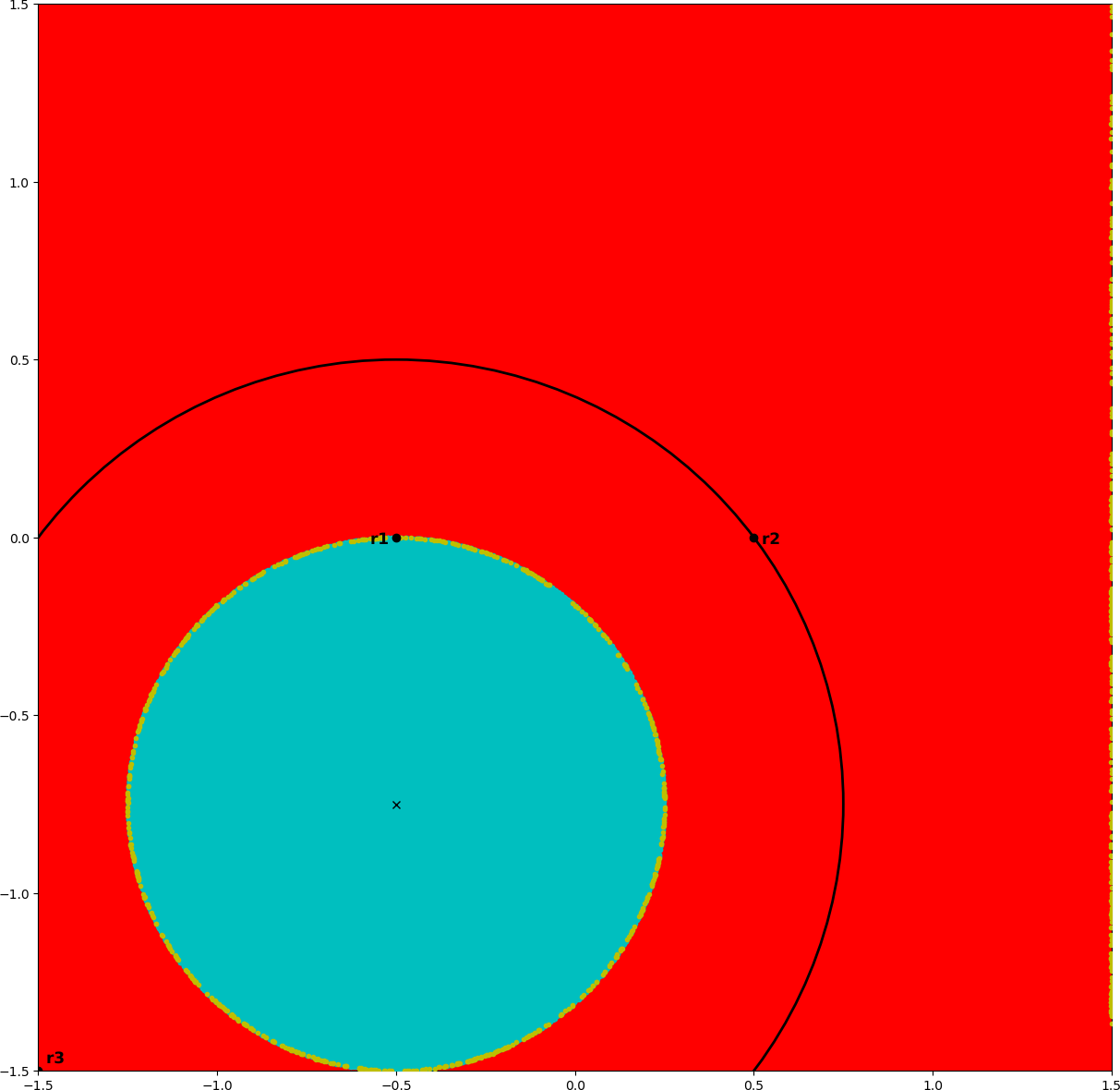}
    \end{subfigure}
    \hfill
    \begin{subfigure}[b]{0.49\linewidth}
        \centering
        \includegraphics[width=\textwidth]{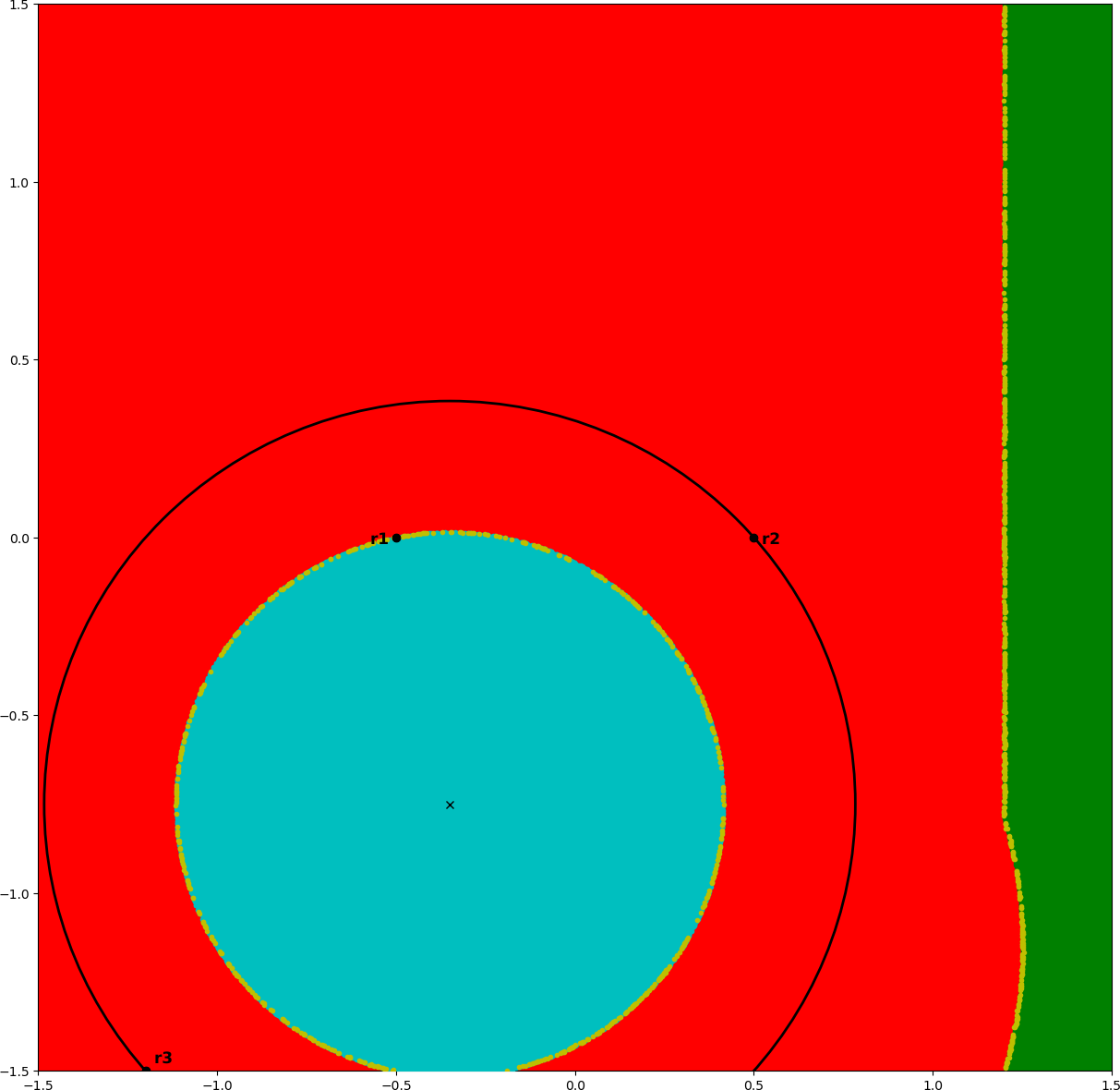}
    \end{subfigure}
    \begin{subfigure}[b]{0.49\linewidth}
        \centering
        \includegraphics[width=\textwidth]{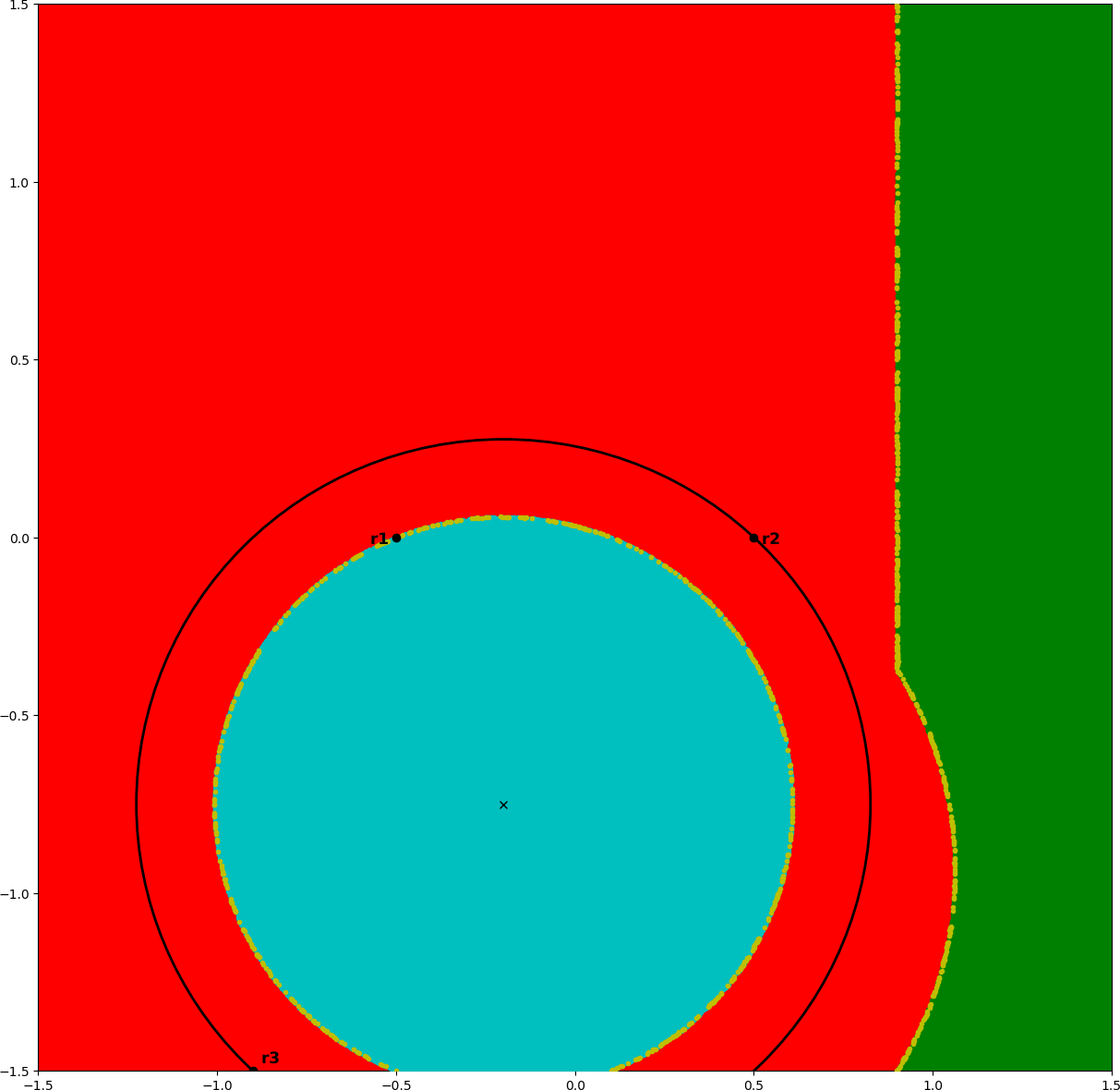}
    \end{subfigure}
    \hfill
    \begin{subfigure}[b]{0.49\linewidth}
        \centering
        \includegraphics[width=\textwidth]{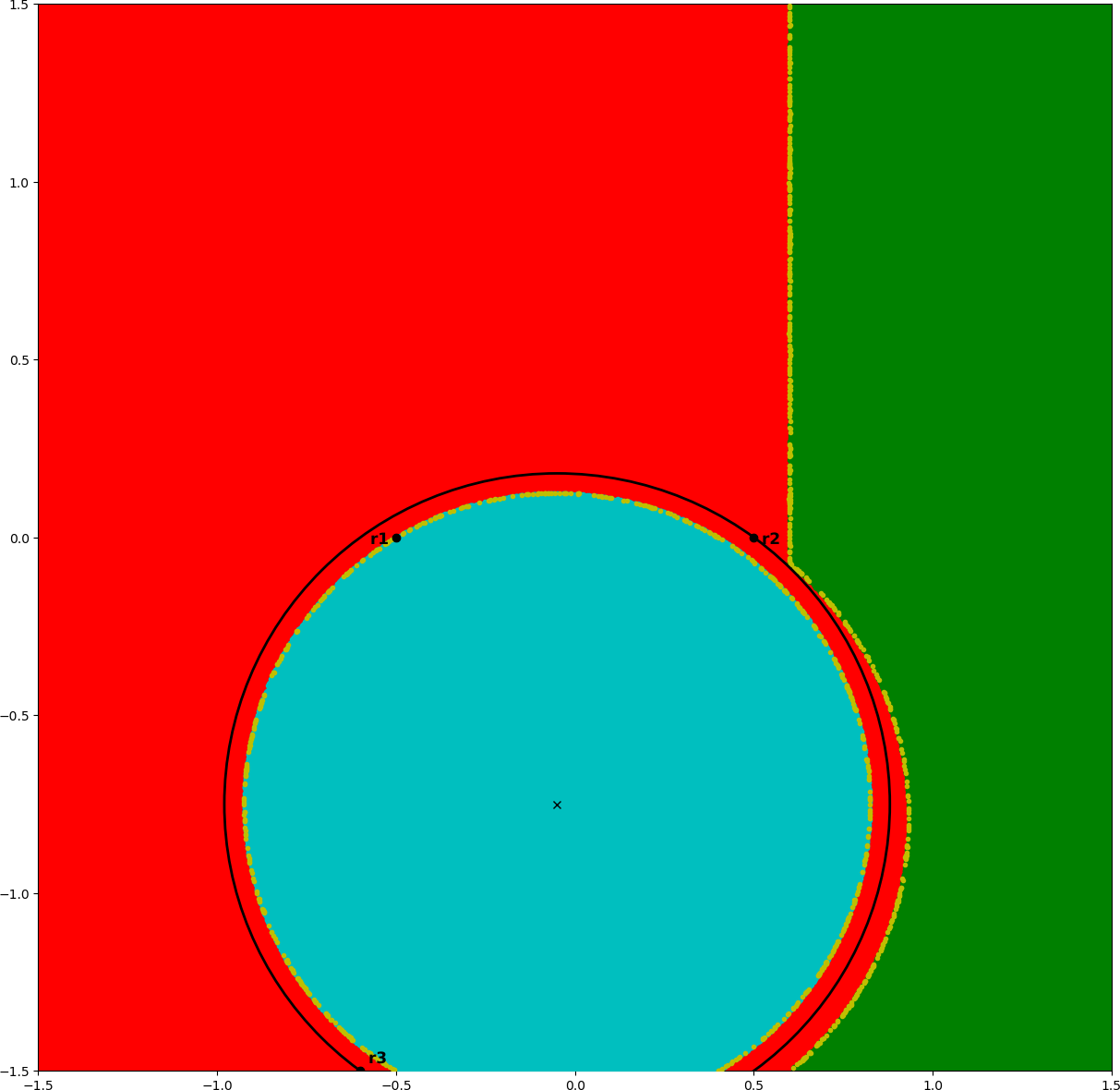}
    \end{subfigure}
    \begin{subfigure}[b]{0.49\linewidth}
        \centering
        \includegraphics[width=\textwidth]{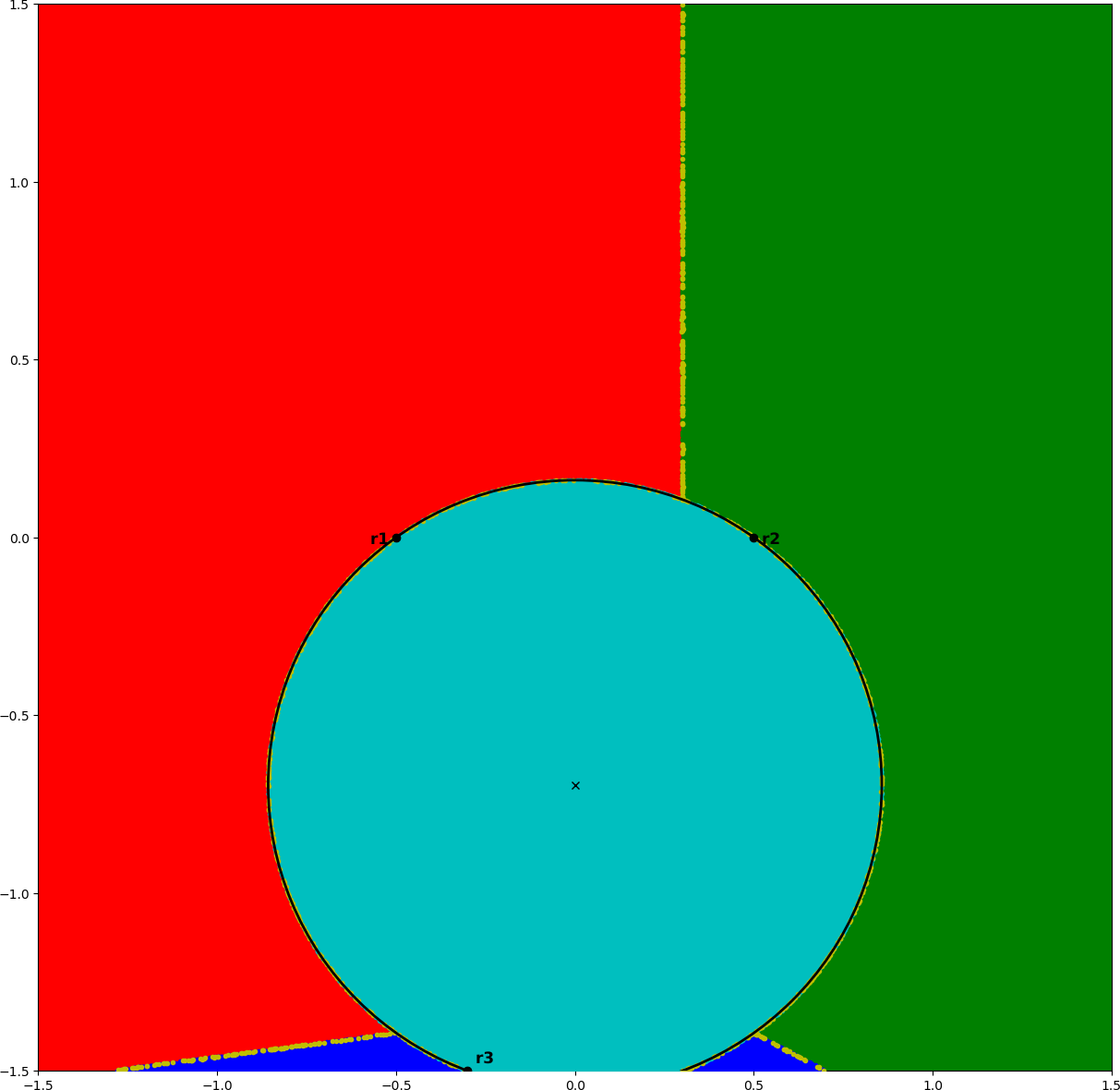}
    \end{subfigure}
    \hfill
    \begin{subfigure}[b]{0.49\linewidth}
        \centering
        \includegraphics[width=\textwidth]{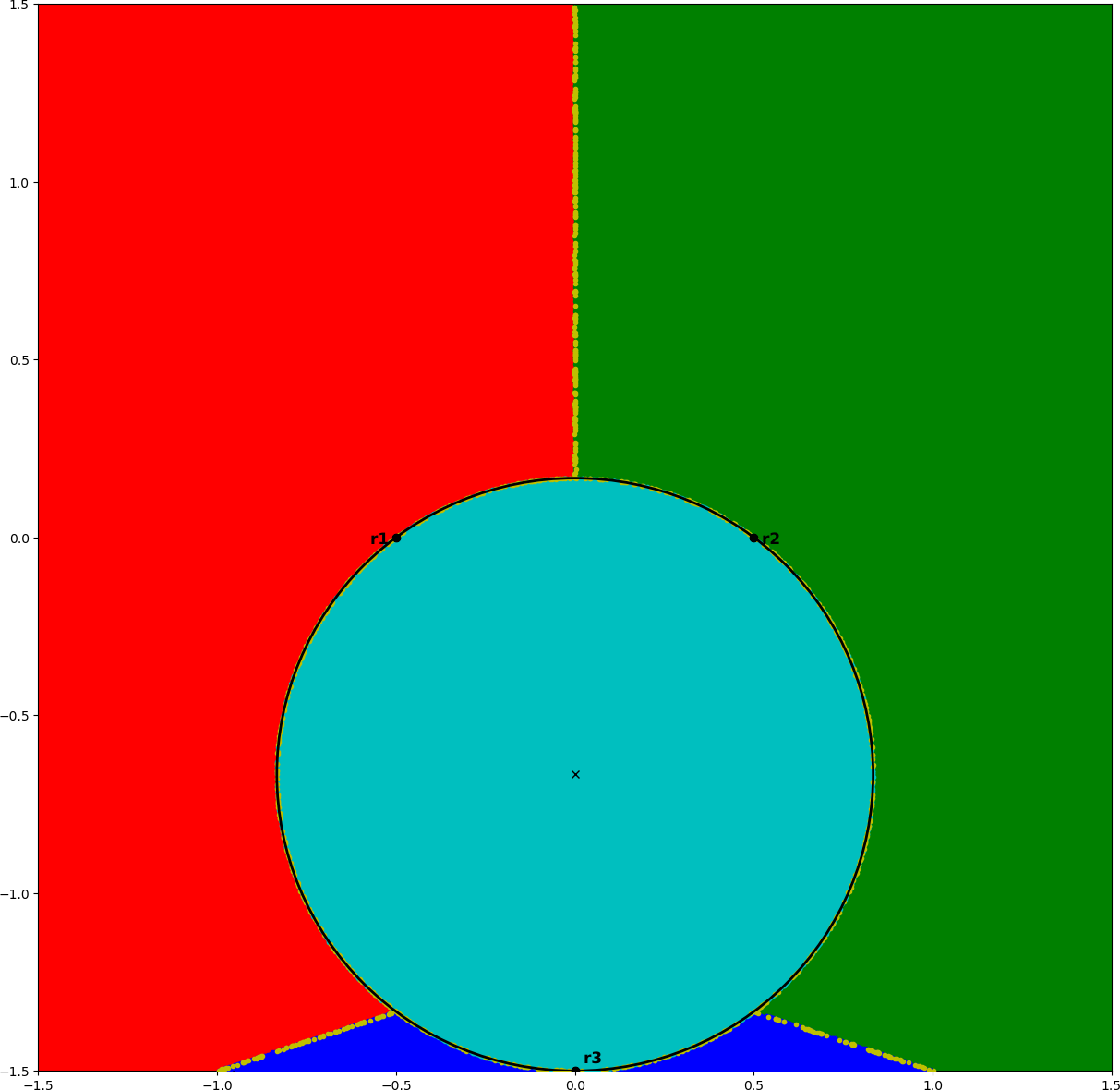}
    \end{subfigure}
    \caption{}
    \label{fig:lead4_0}
\end{figure}

\begin{figure}
    \centering
    \begin{subfigure}[b]{0.49\linewidth}
        \centering
        \includegraphics[width=\textwidth]{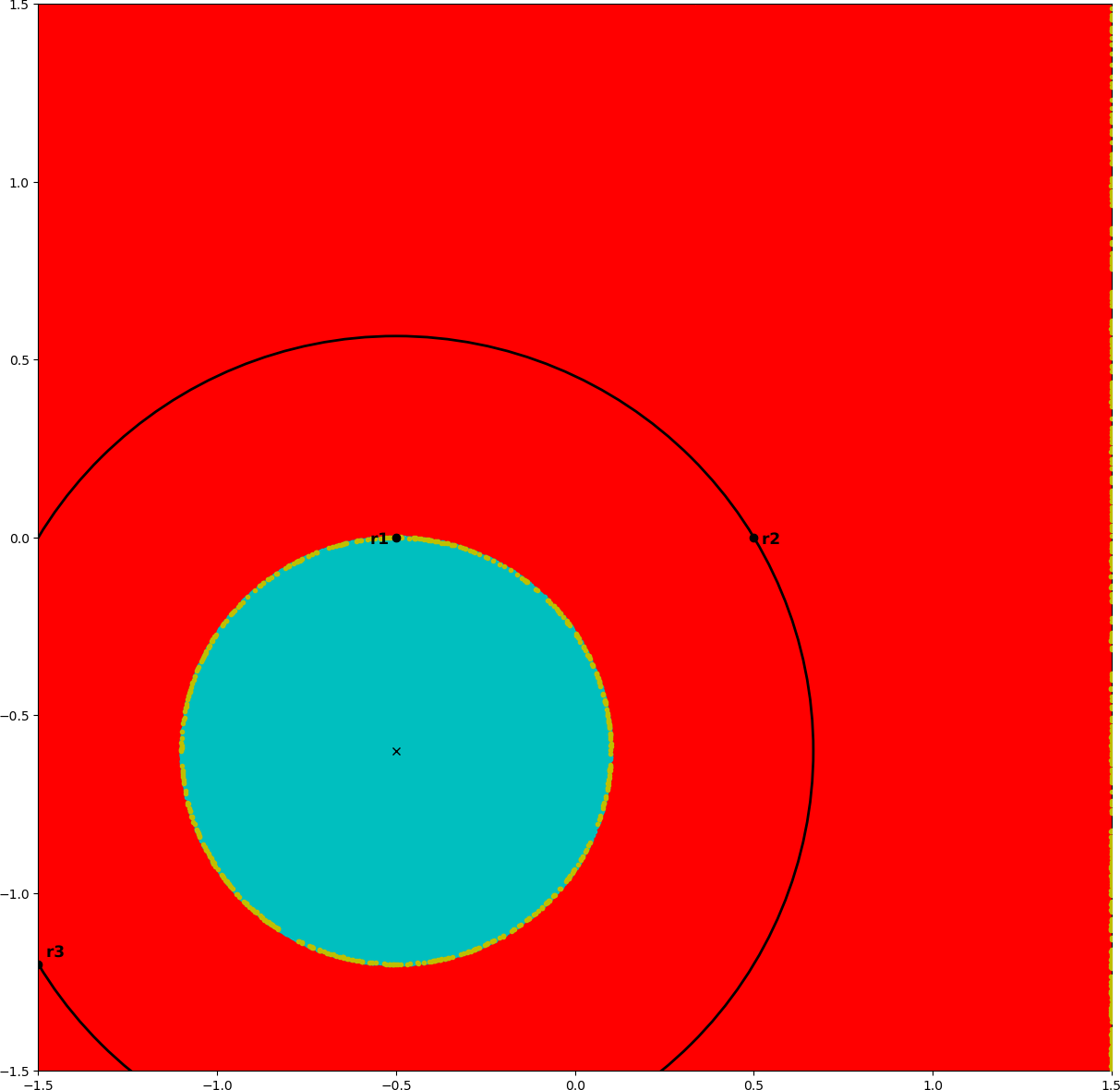}
    \end{subfigure}
    \hfill
    \begin{subfigure}[b]{0.49\linewidth}
        \centering
        \includegraphics[width=\textwidth]{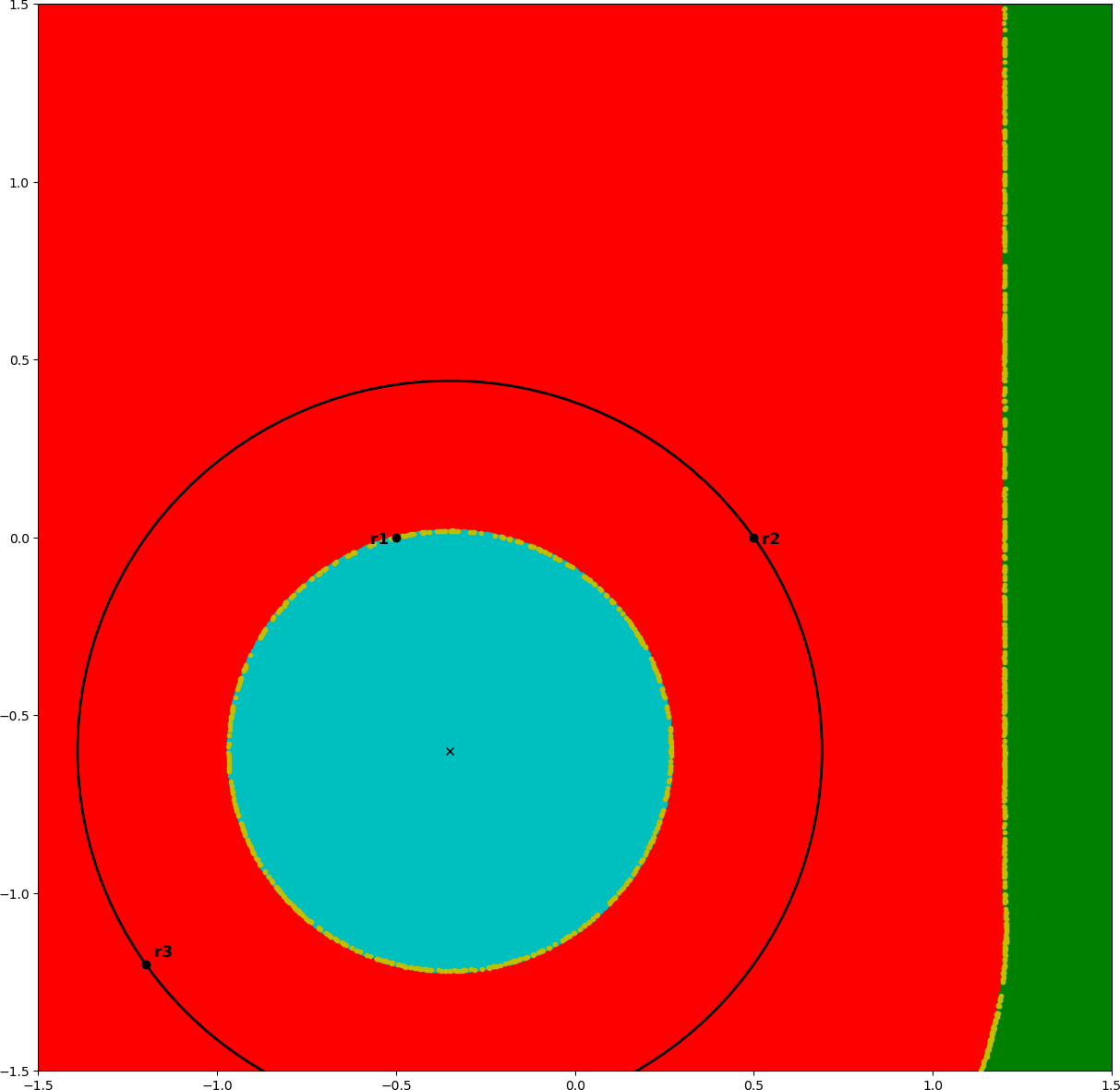}
    \end{subfigure}
    \begin{subfigure}[b]{0.49\linewidth}
        \centering
        \includegraphics[width=\textwidth]{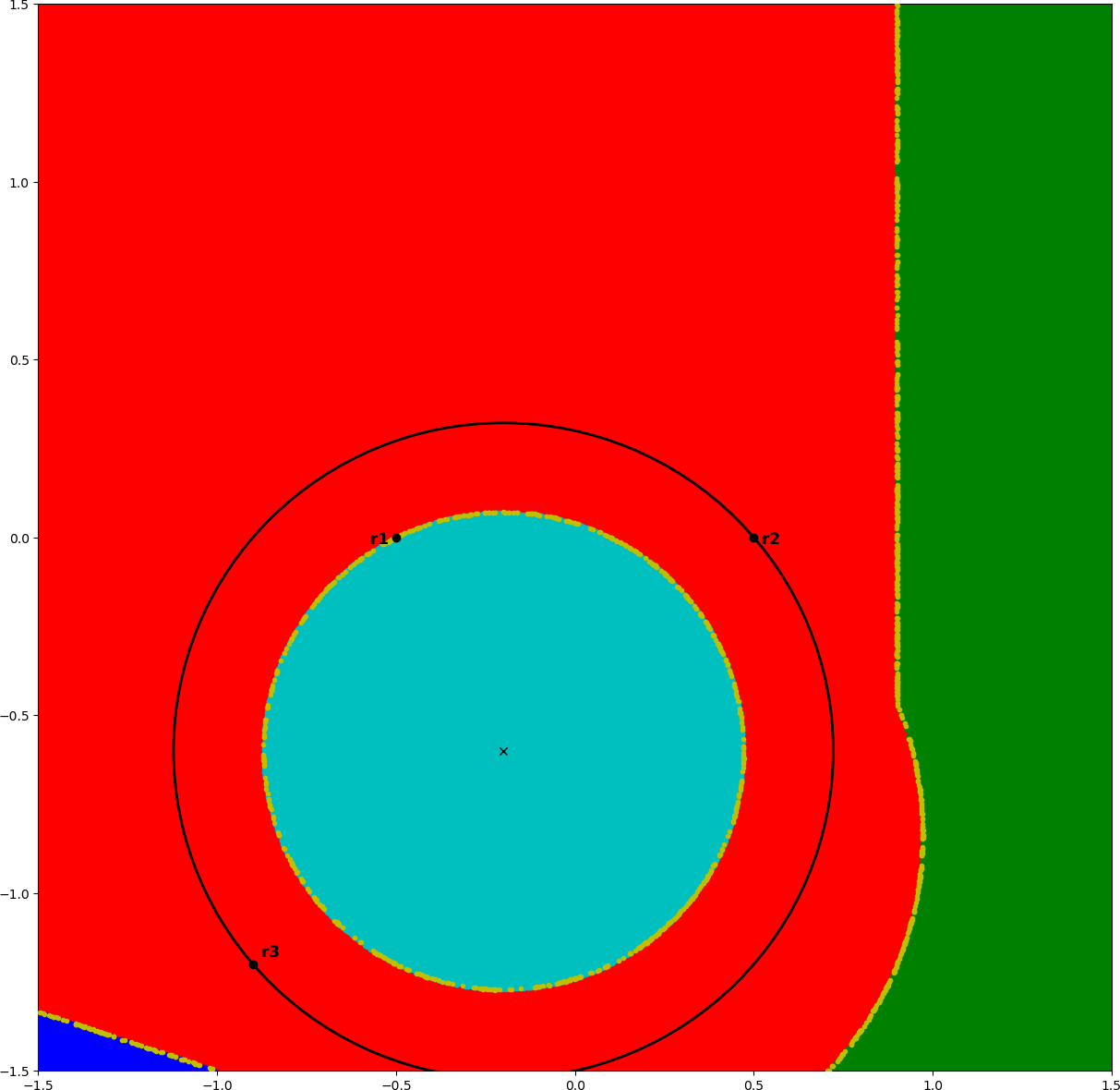}
    \end{subfigure}
    \hfill
    \begin{subfigure}[b]{0.49\linewidth}
        \centering
        \includegraphics[width=\textwidth]{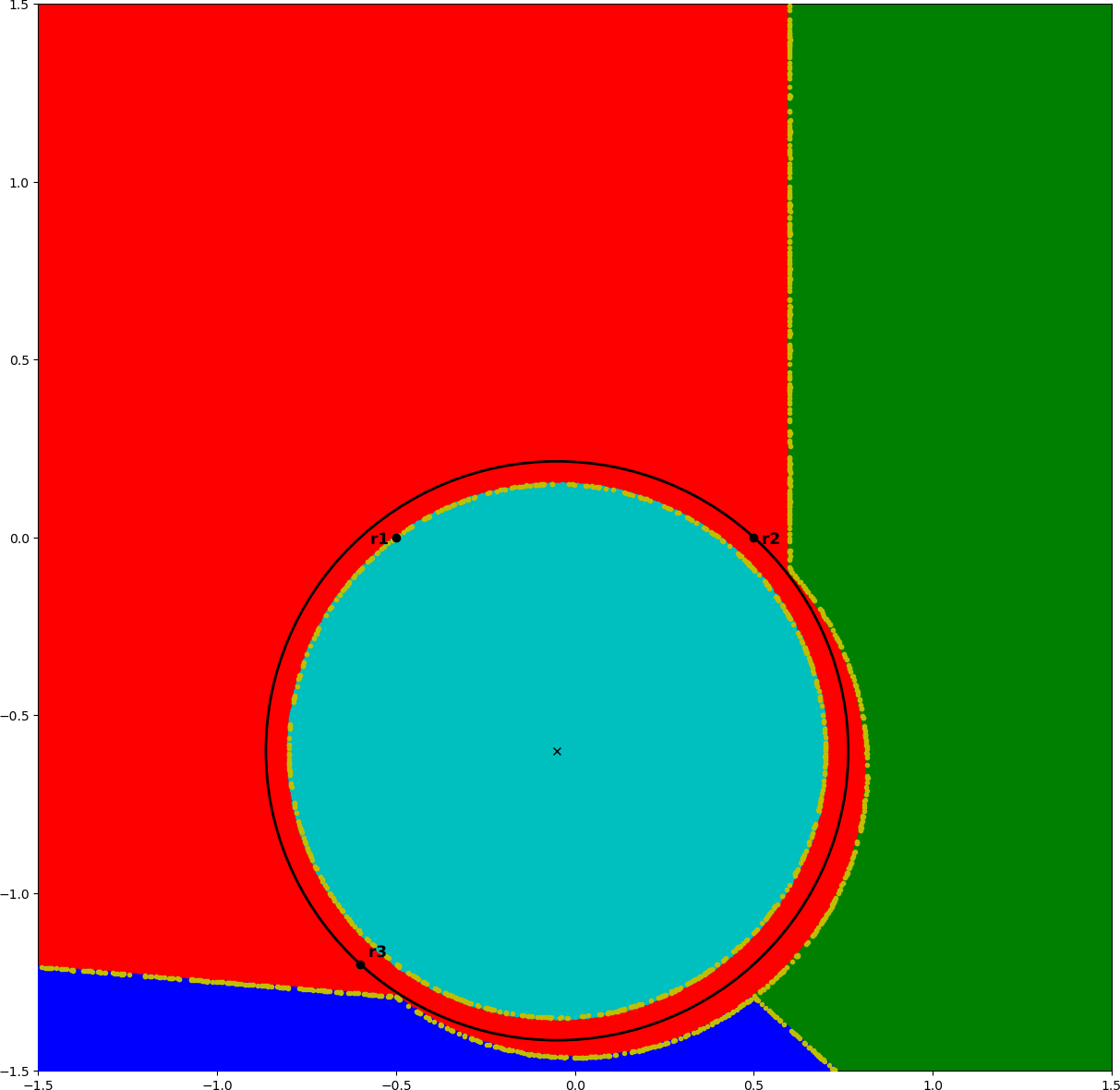}
    \end{subfigure}
    \begin{subfigure}[b]{0.49\linewidth}
        \centering
        \includegraphics[width=\textwidth]{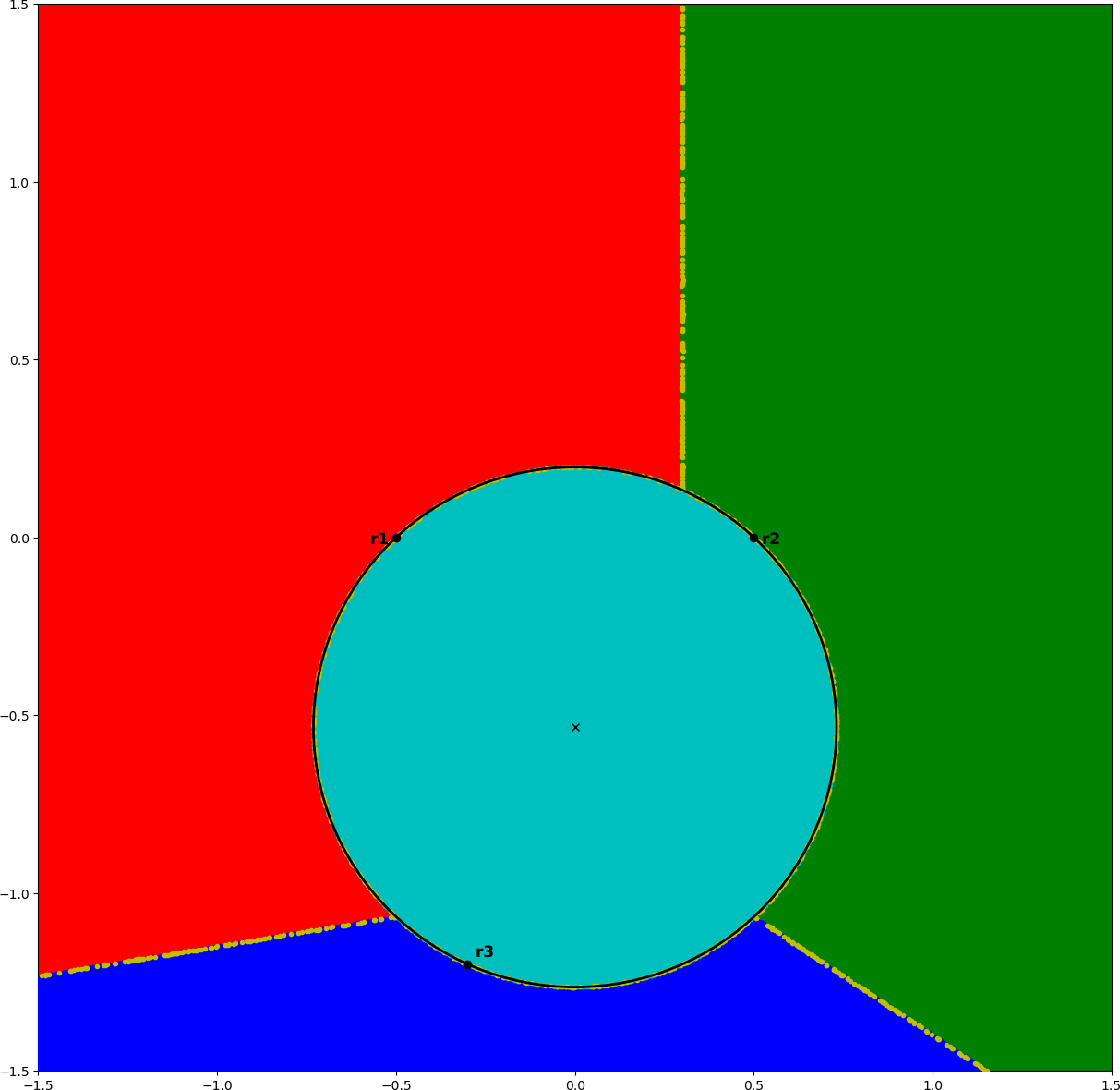}
    \end{subfigure}
    \hfill
    \begin{subfigure}[b]{0.49\linewidth}
        \centering
        \includegraphics[width=\textwidth]{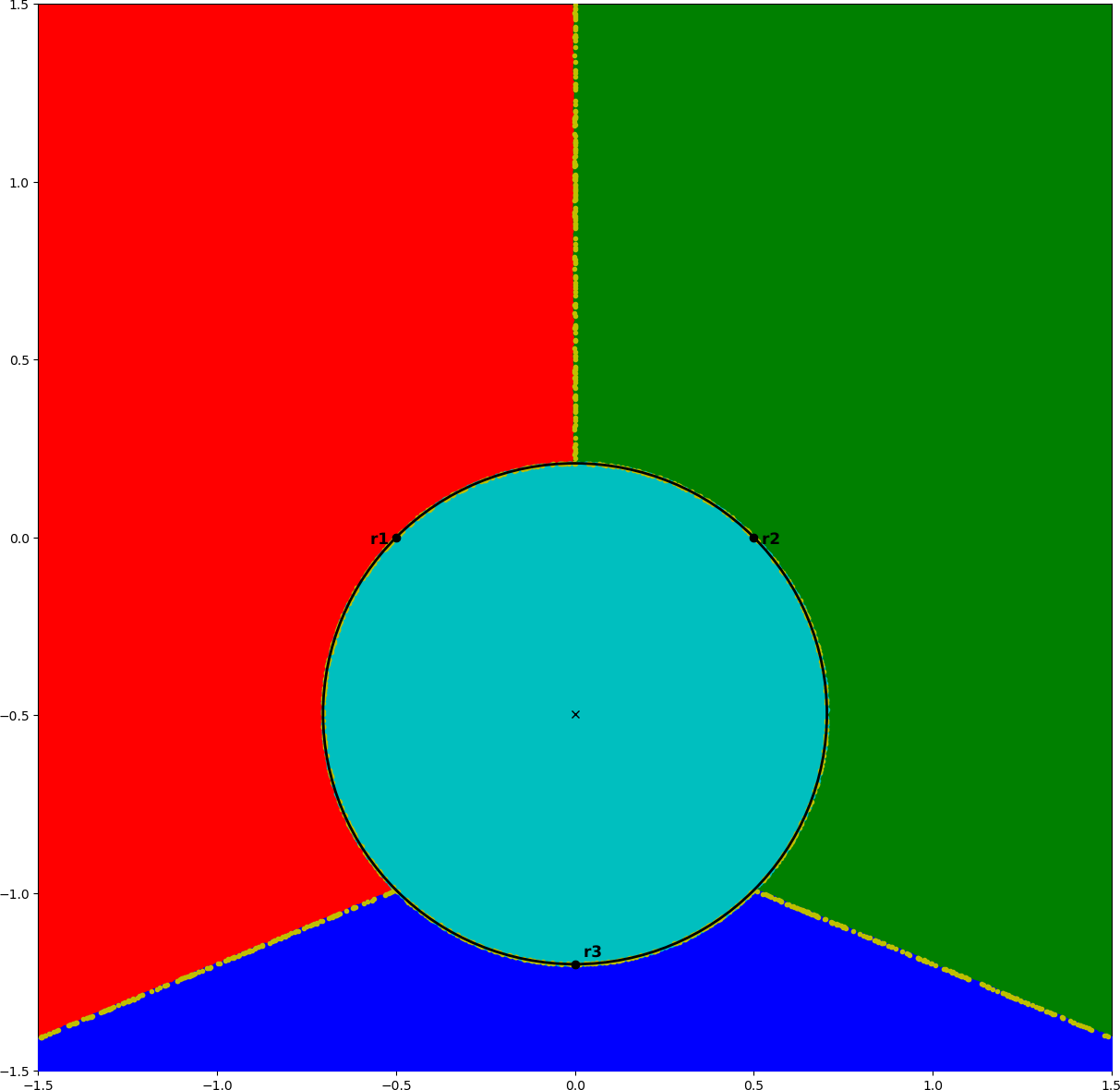}
    \end{subfigure}
    \caption{}
    \label{fig:lead4_1}
\end{figure}

\begin{figure}
    \centering
    \begin{subfigure}[b]{0.49\linewidth}
        \centering
        \includegraphics[width=\textwidth]{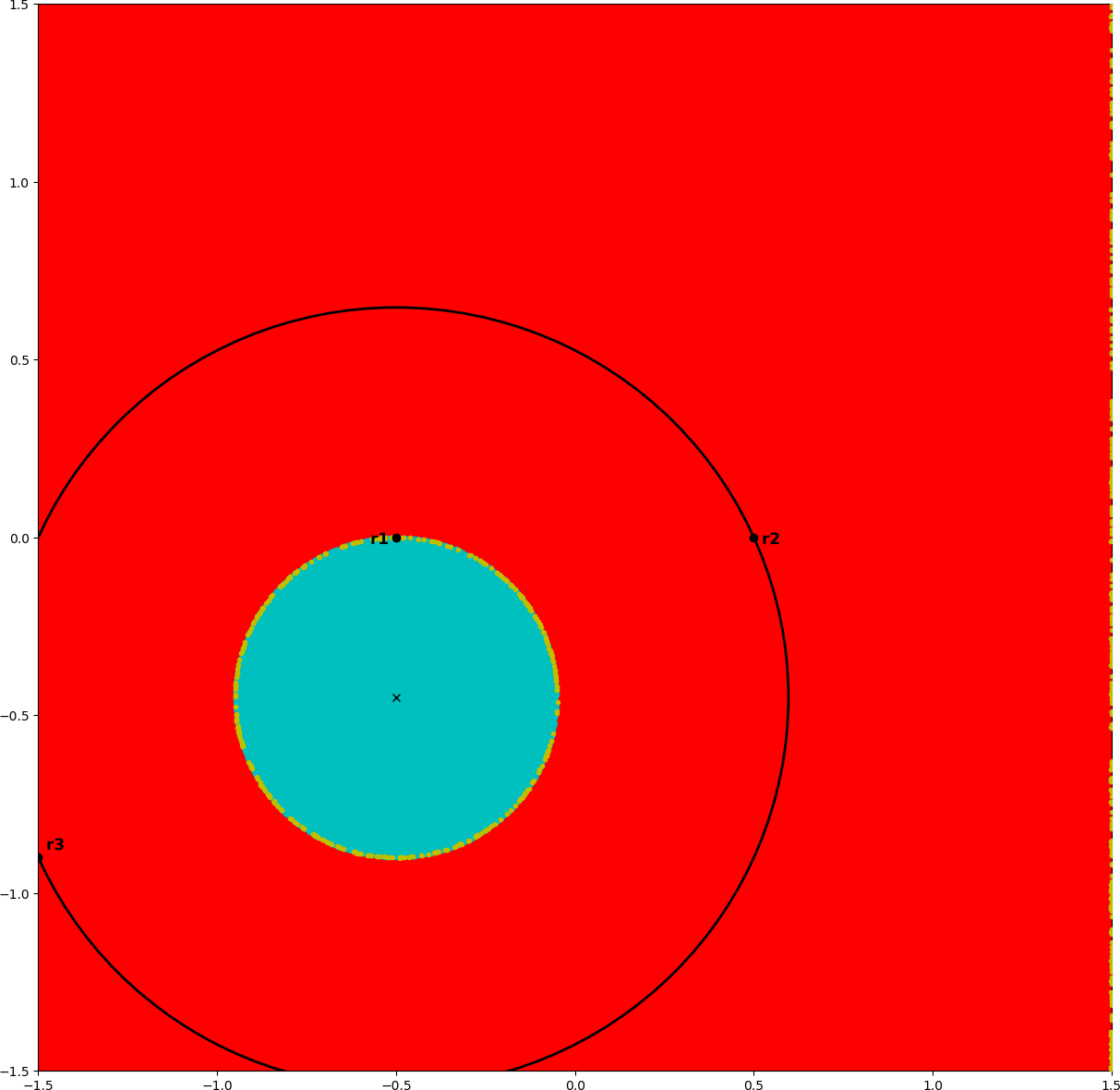}
    \end{subfigure}
    \hfill
    \begin{subfigure}[b]{0.49\linewidth}
        \centering
        \includegraphics[width=\textwidth]{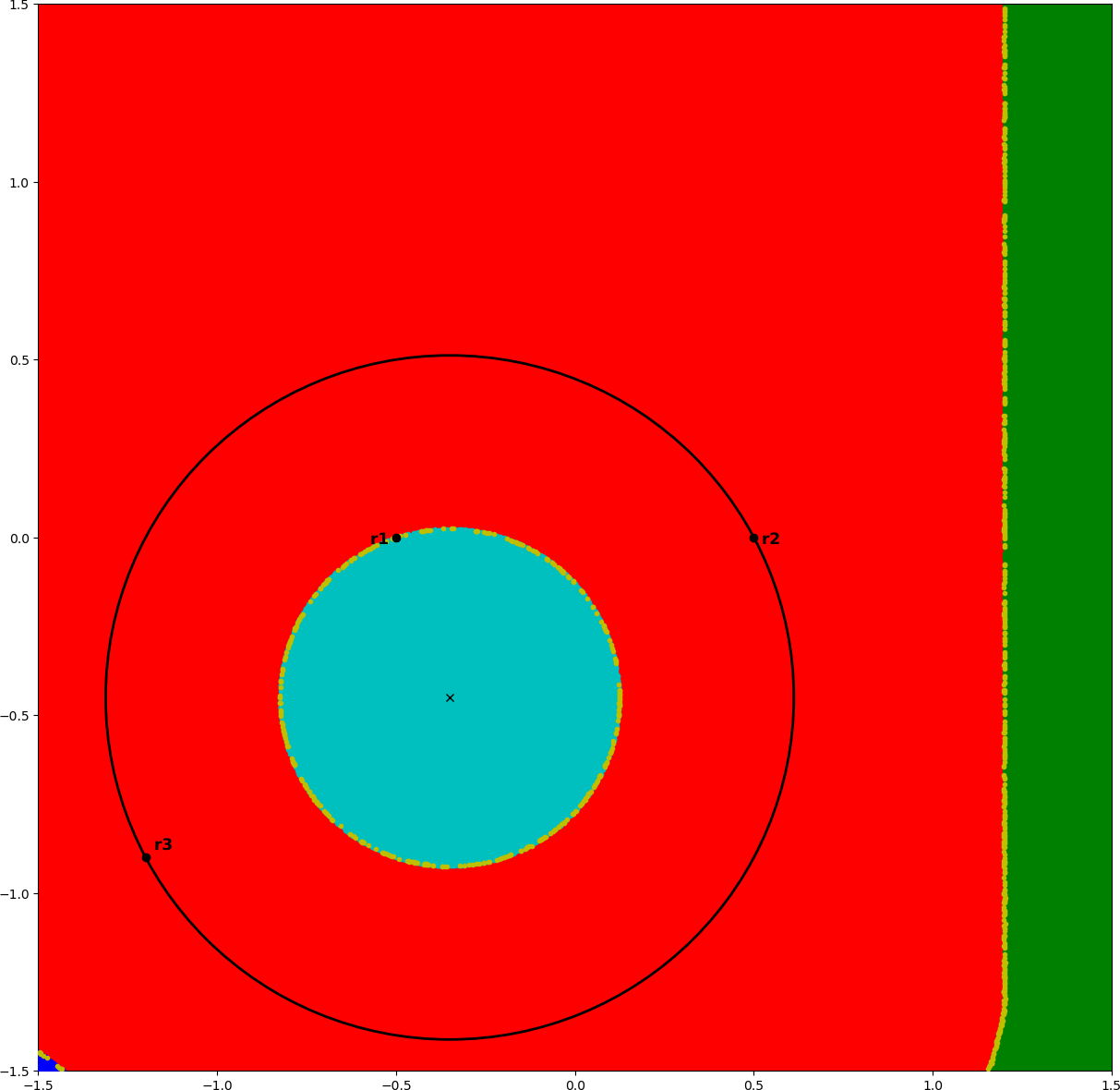}
    \end{subfigure}
    \begin{subfigure}[b]{0.49\linewidth}
        \centering
        \includegraphics[width=\textwidth]{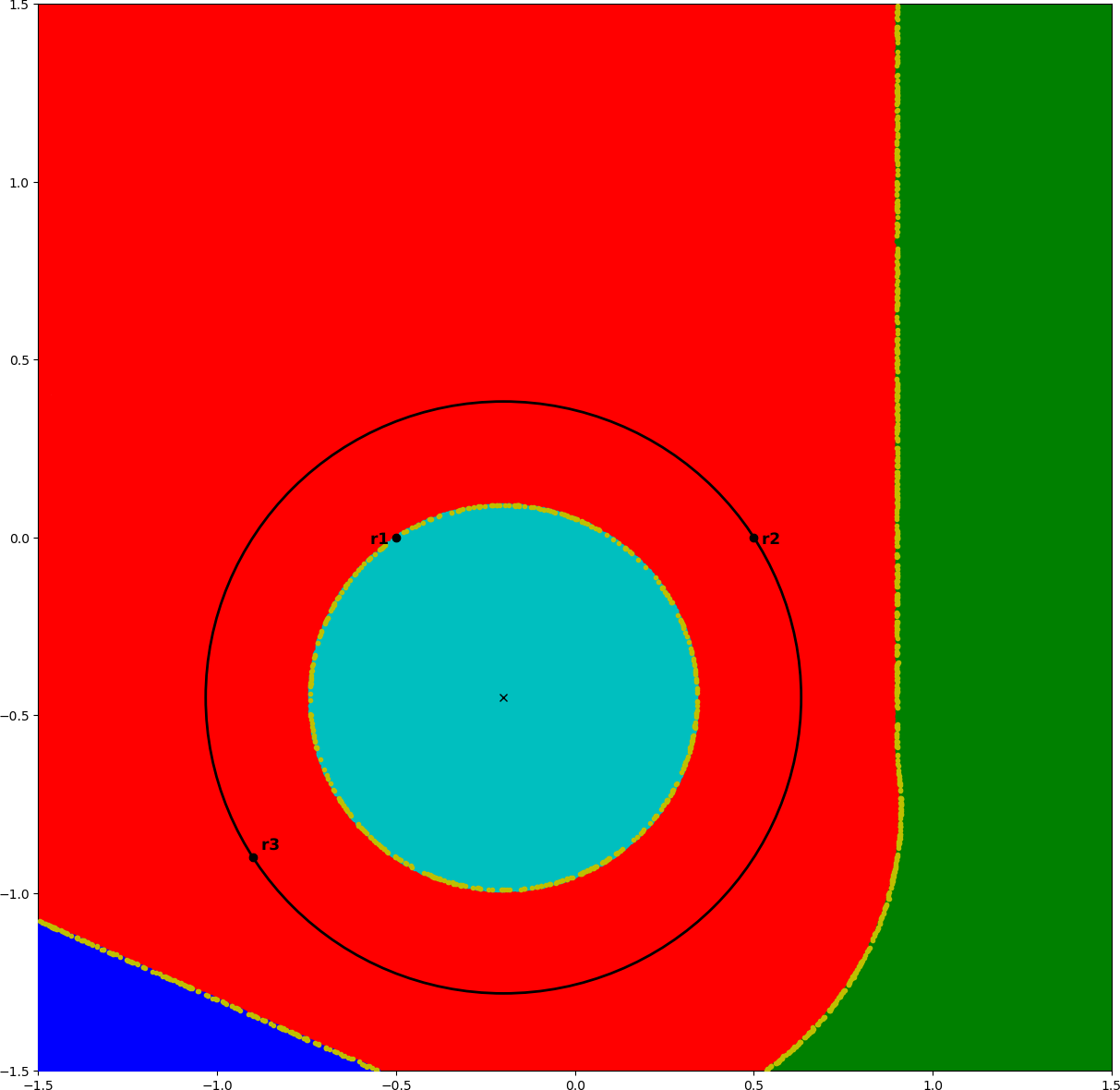}
    \end{subfigure}
    \hfill
    \begin{subfigure}[b]{0.49\linewidth}
        \centering
        \includegraphics[width=\textwidth]{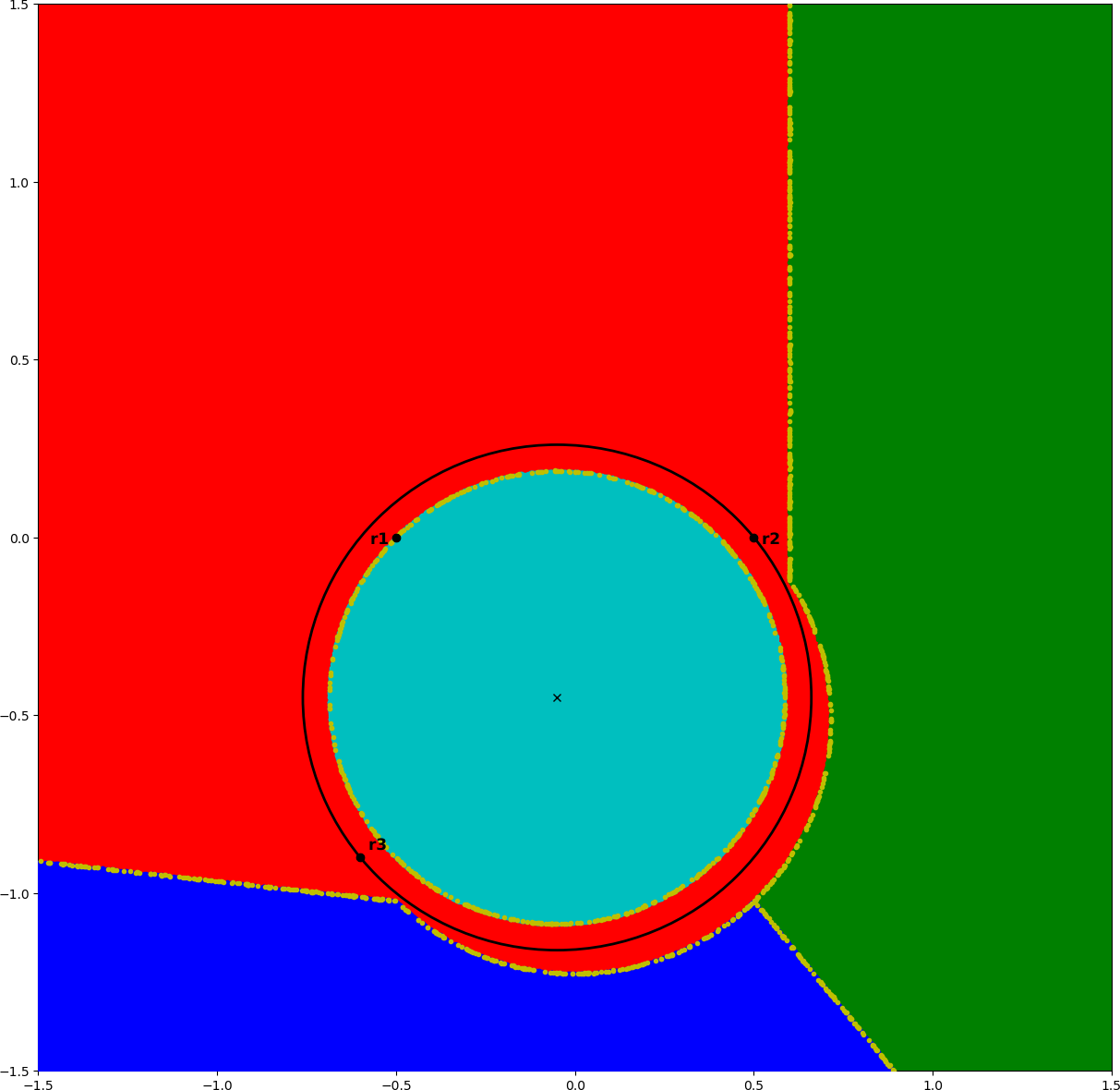}
    \end{subfigure}
    \begin{subfigure}[b]{0.49\linewidth}
        \centering
        \includegraphics[width=\textwidth]{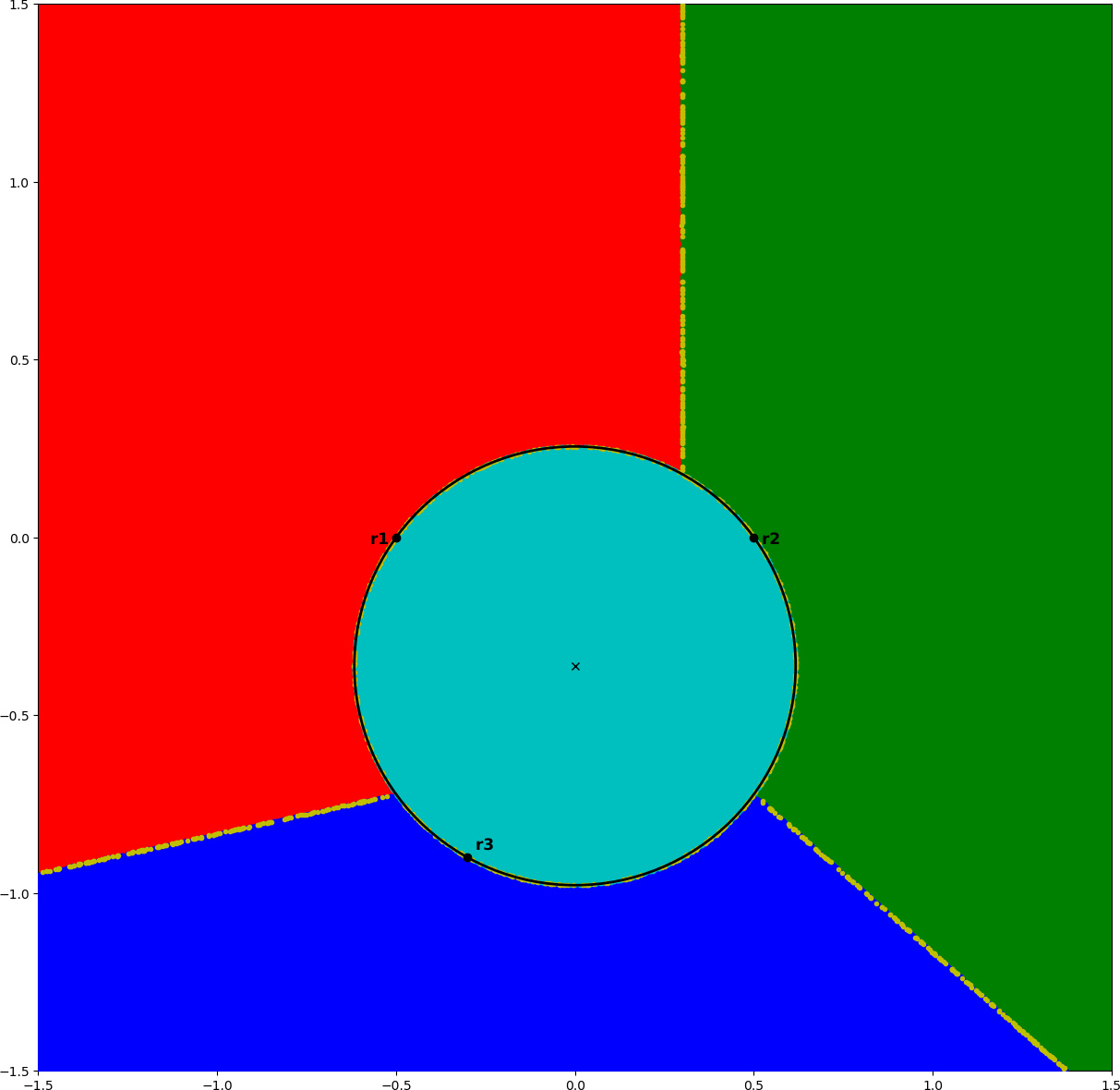}
    \end{subfigure}
    \hfill
    \begin{subfigure}[b]{0.49\linewidth}
        \centering
        \includegraphics[width=\textwidth]{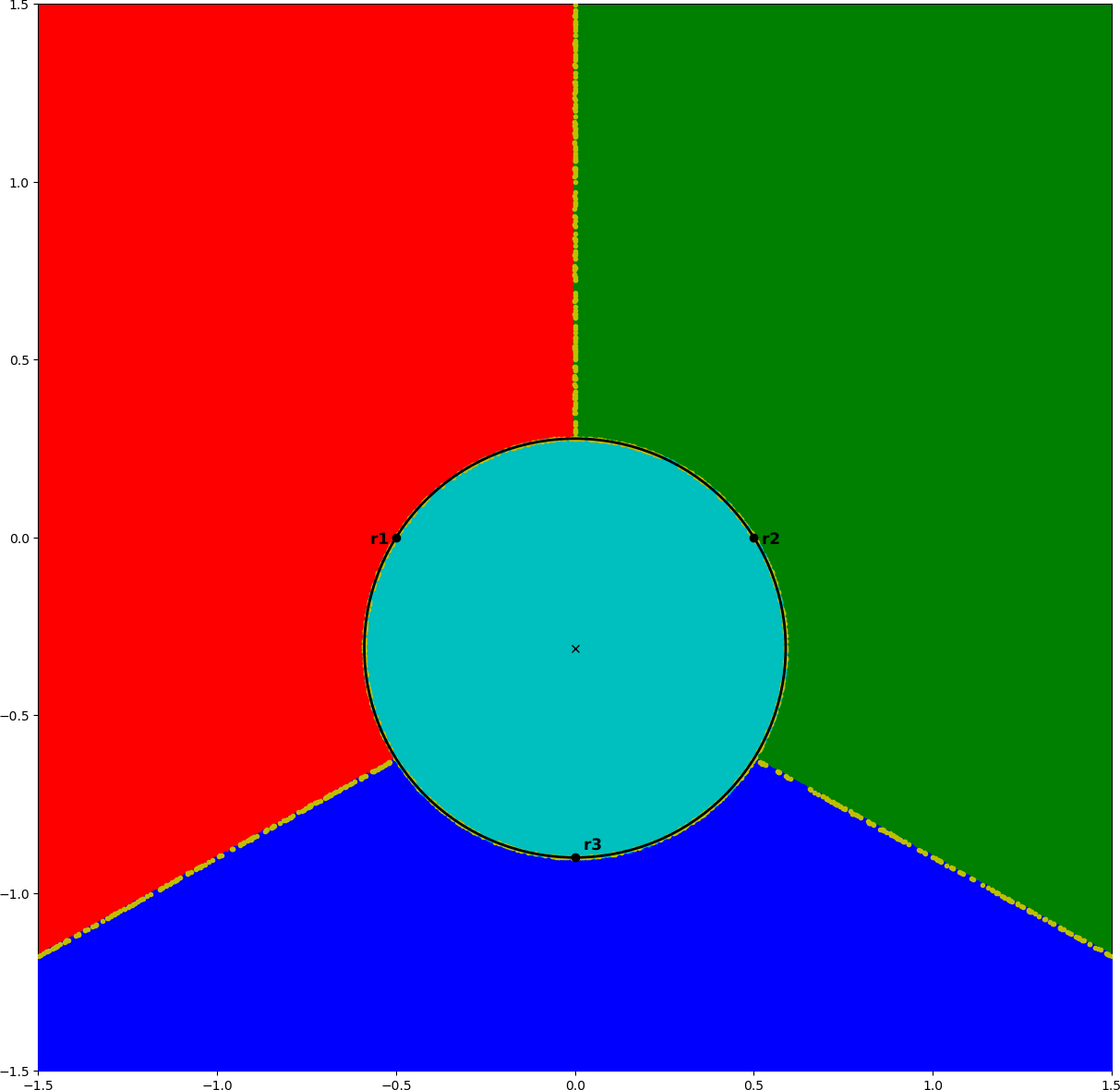}
    \end{subfigure}
    \caption{}
    \label{fig:lead4_2}
\end{figure}

\begin{figure}
    \centering
    \begin{subfigure}[b]{0.49\linewidth}
        \centering
        \includegraphics[width=\textwidth]{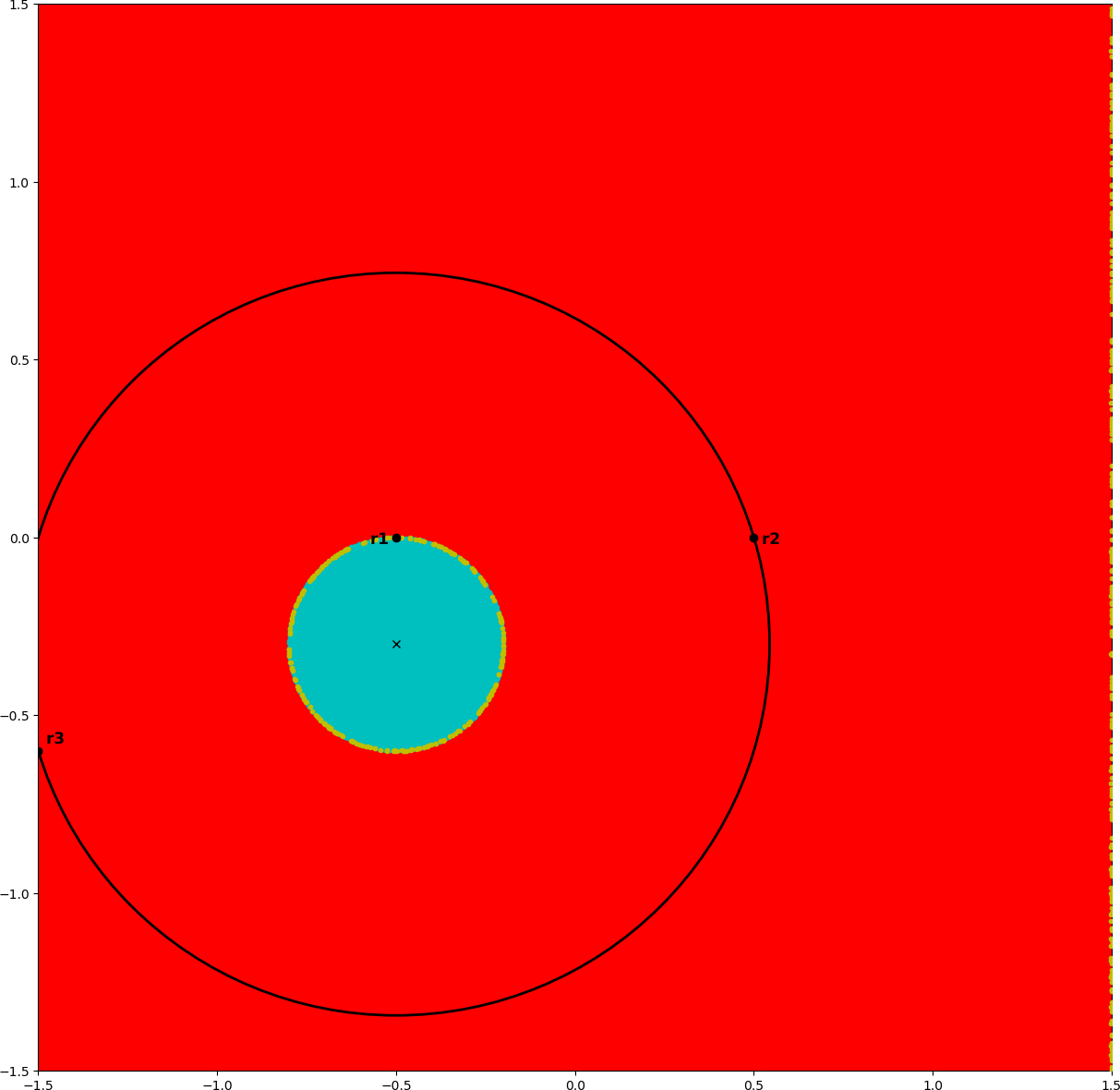}
    \end{subfigure}
    \hfill
    \begin{subfigure}[b]{0.49\linewidth}
        \centering
        \includegraphics[width=\textwidth]{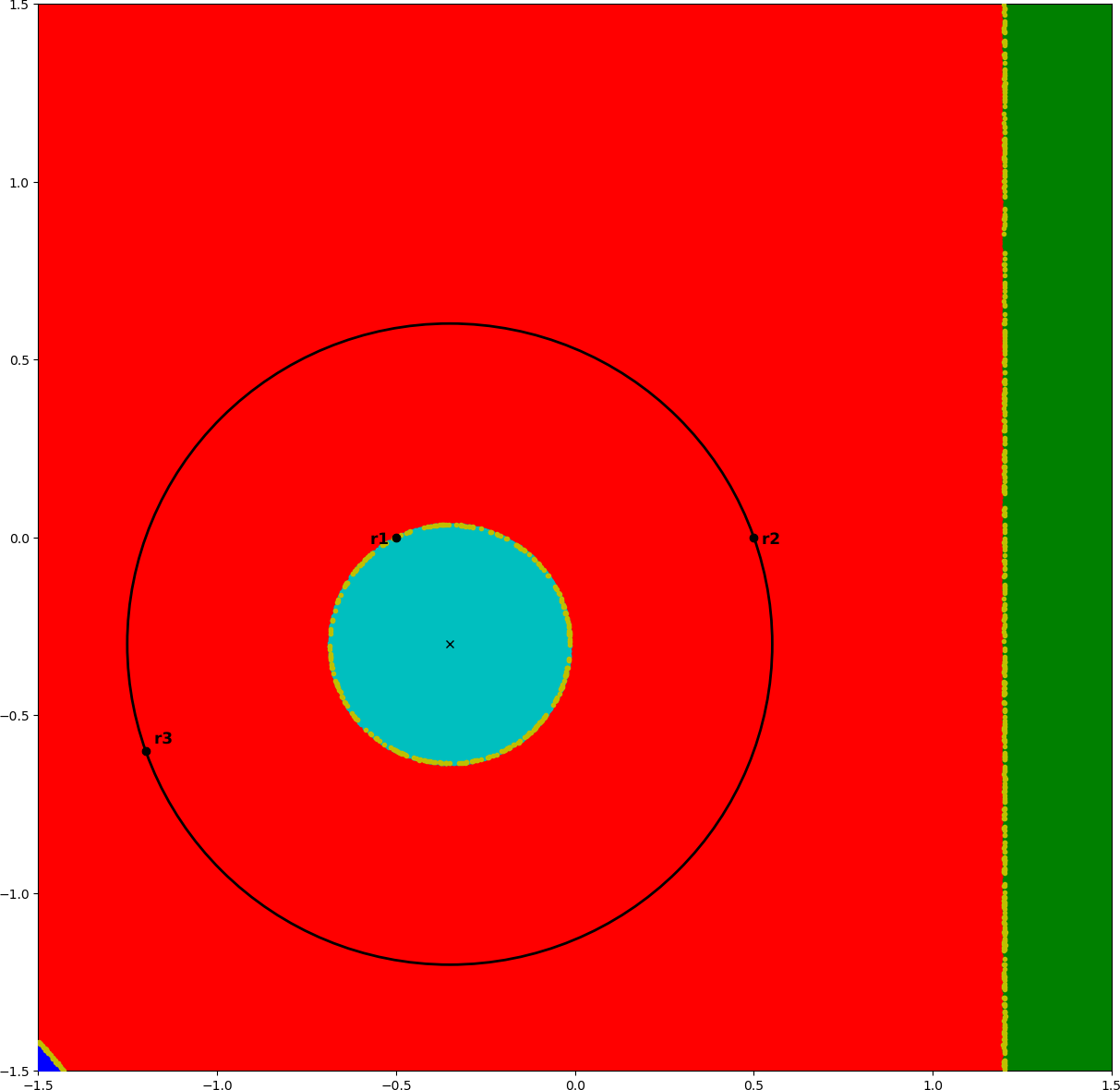}
    \end{subfigure}
    \begin{subfigure}[b]{0.49\linewidth}
        \centering
        \includegraphics[width=\textwidth]{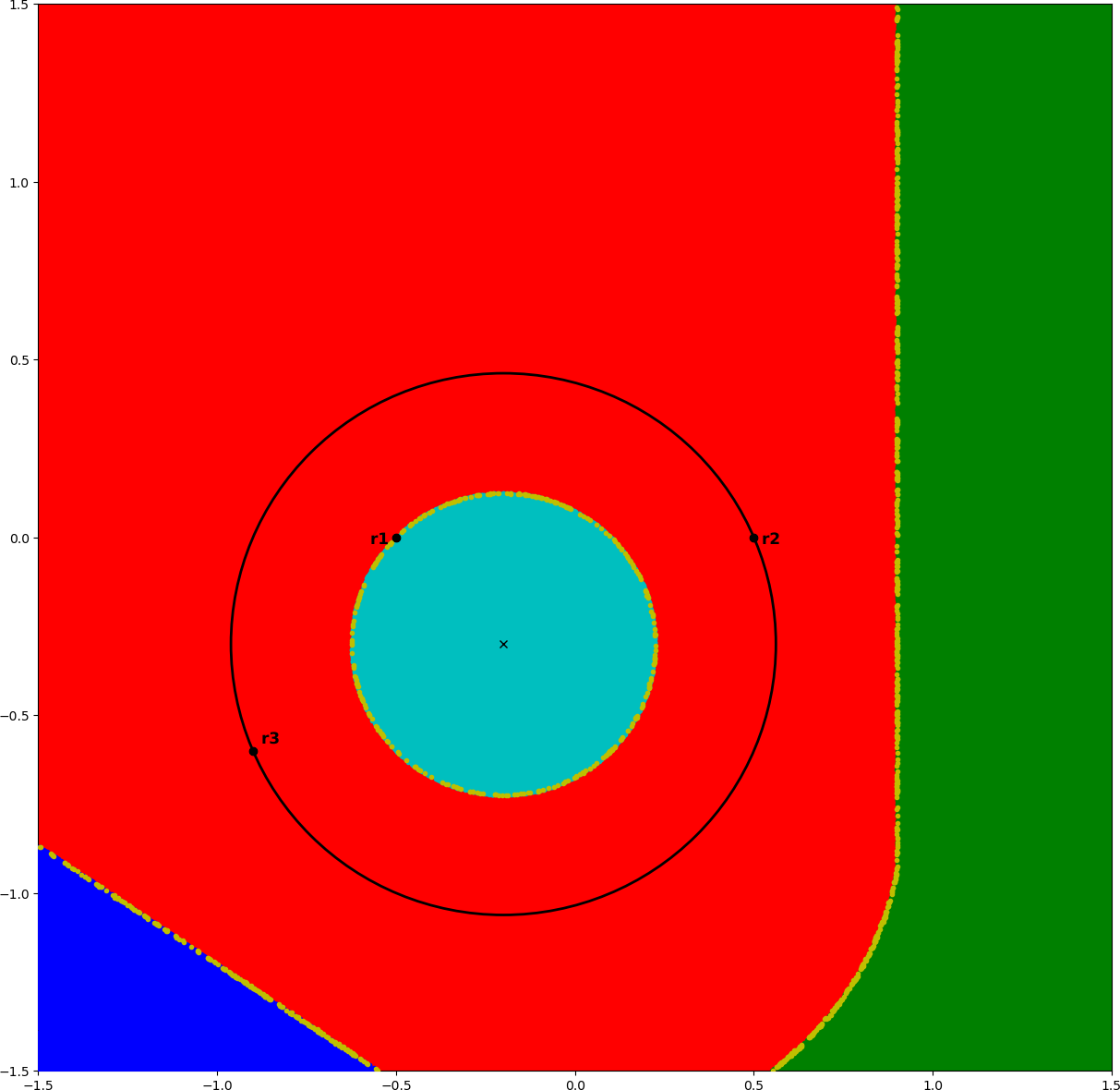}
    \end{subfigure}
    \hfill
    \begin{subfigure}[b]{0.49\linewidth}
        \centering
        \includegraphics[width=\textwidth]{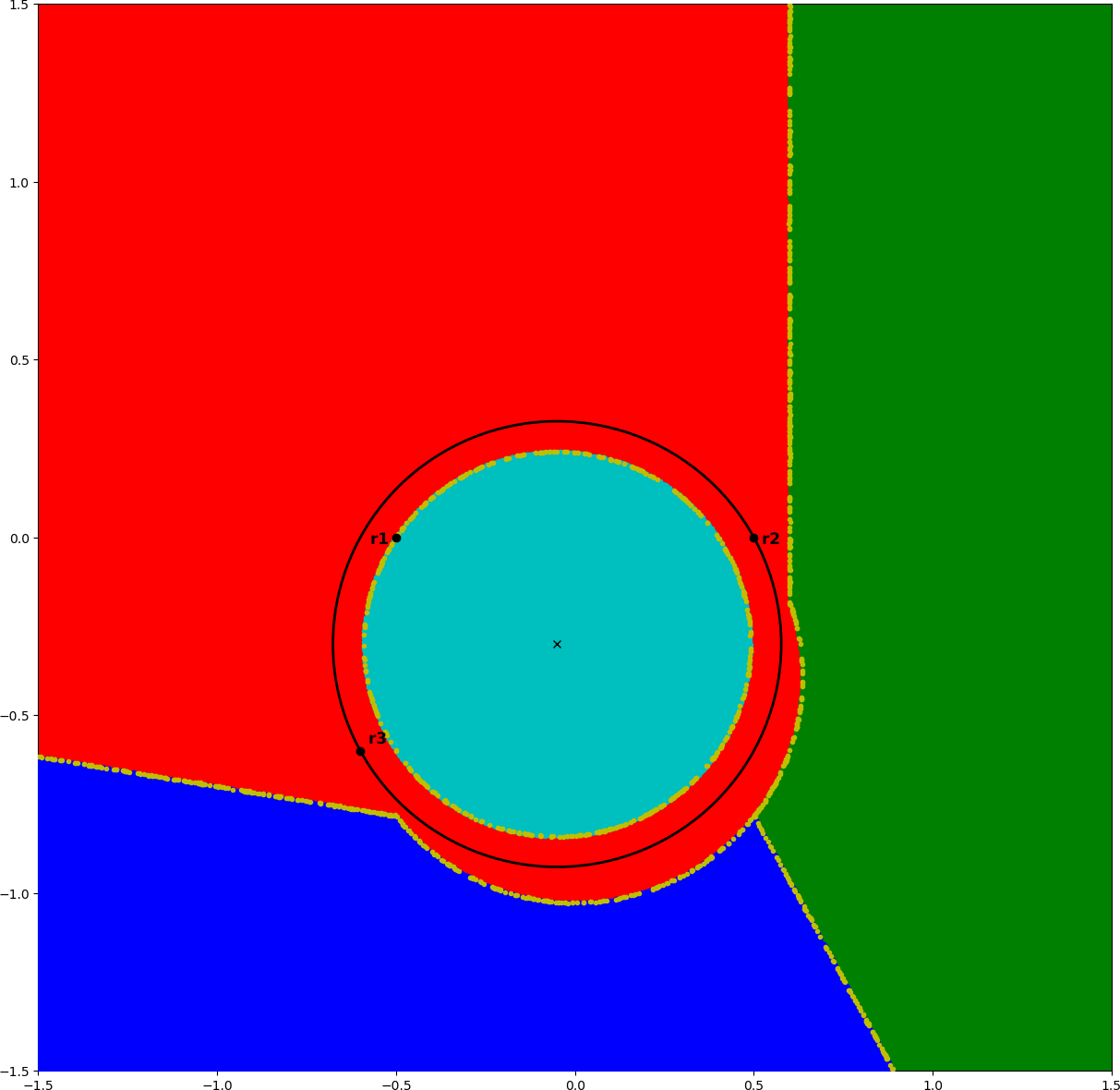}
    \end{subfigure}
    \begin{subfigure}[b]{0.49\linewidth}
        \centering
        \includegraphics[width=\textwidth]{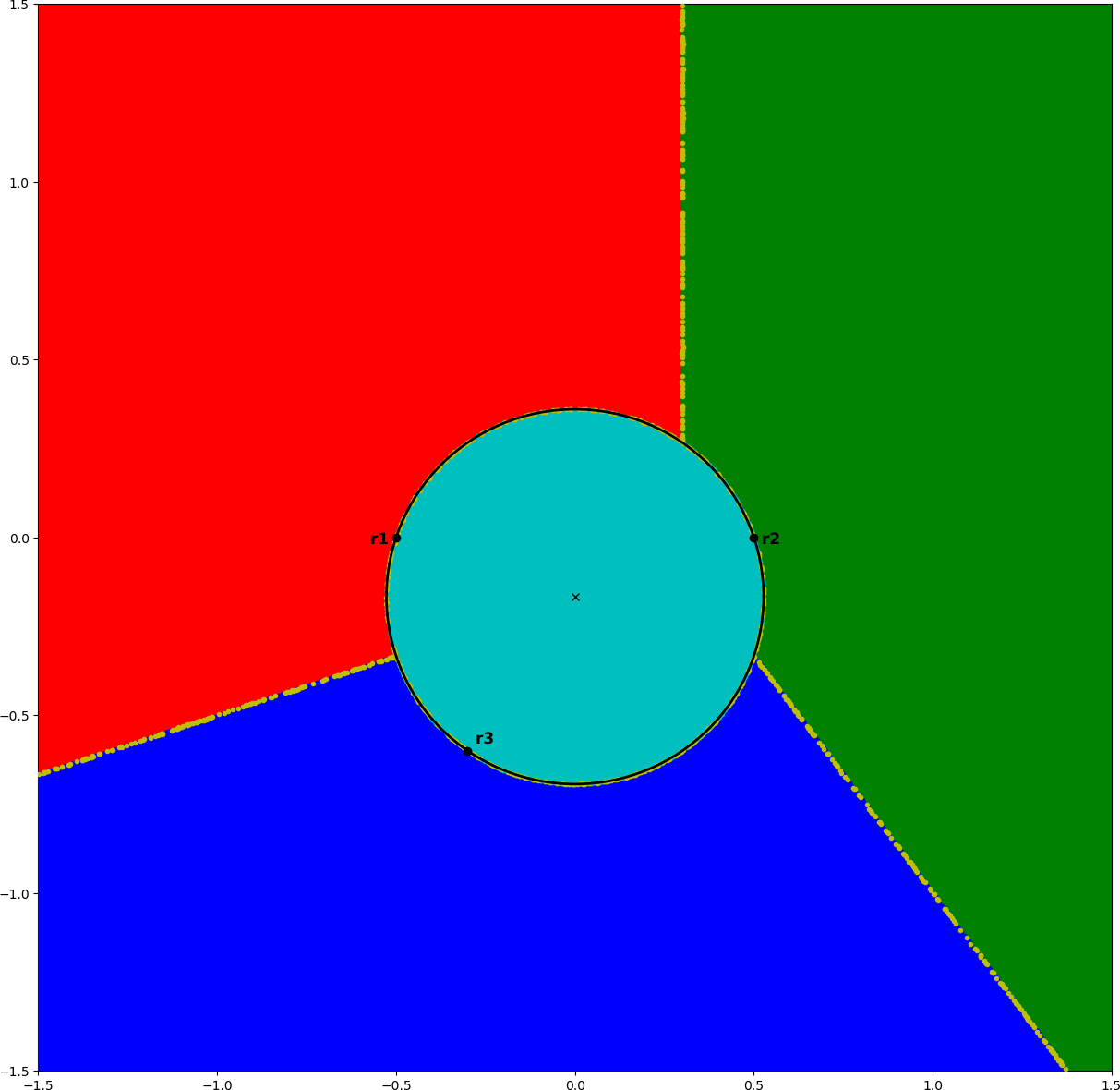}
    \end{subfigure}
    \hfill
    \begin{subfigure}[b]{0.49\linewidth}
        \centering
        \includegraphics[width=\textwidth]{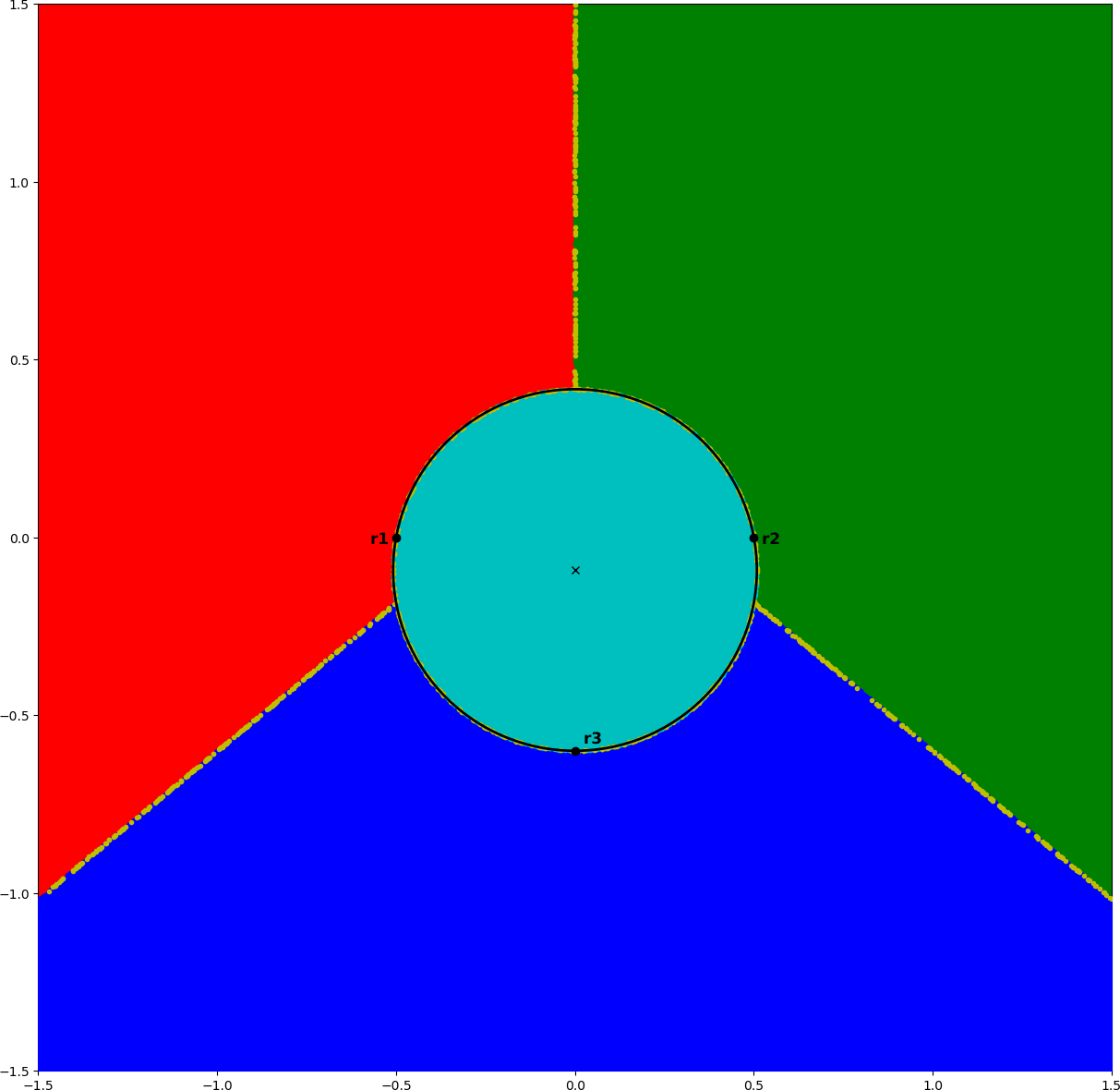}
    \end{subfigure}
    \caption{}
    \label{fig:lead4_3}
\end{figure}

\begin{figure}
    \centering
    \begin{subfigure}[b]{0.49\linewidth}
        \centering
        \includegraphics[width=\textwidth]{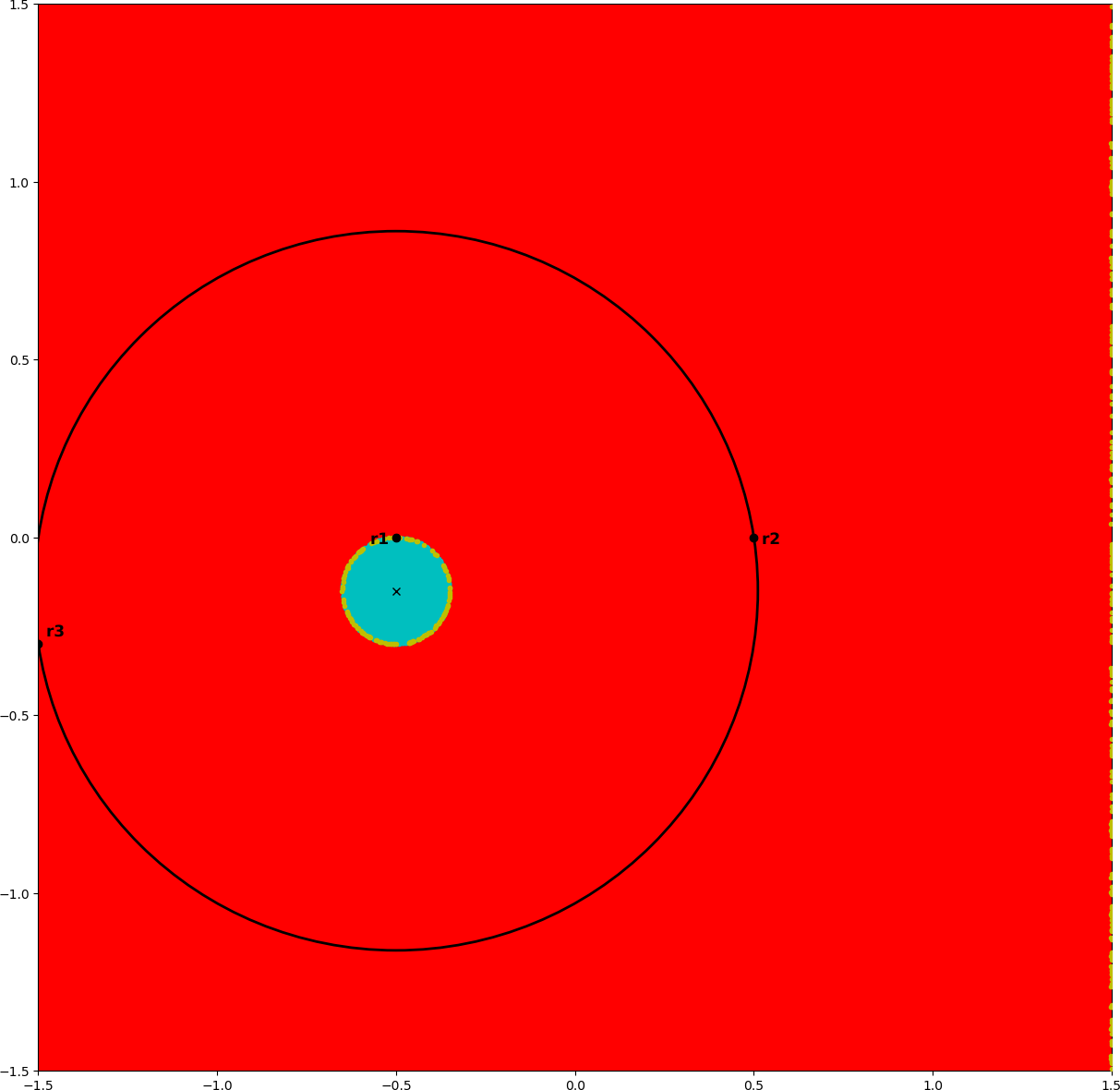}
    \end{subfigure}
    \hfill
    \begin{subfigure}[b]{0.49\linewidth}
        \centering
        \includegraphics[width=\textwidth]{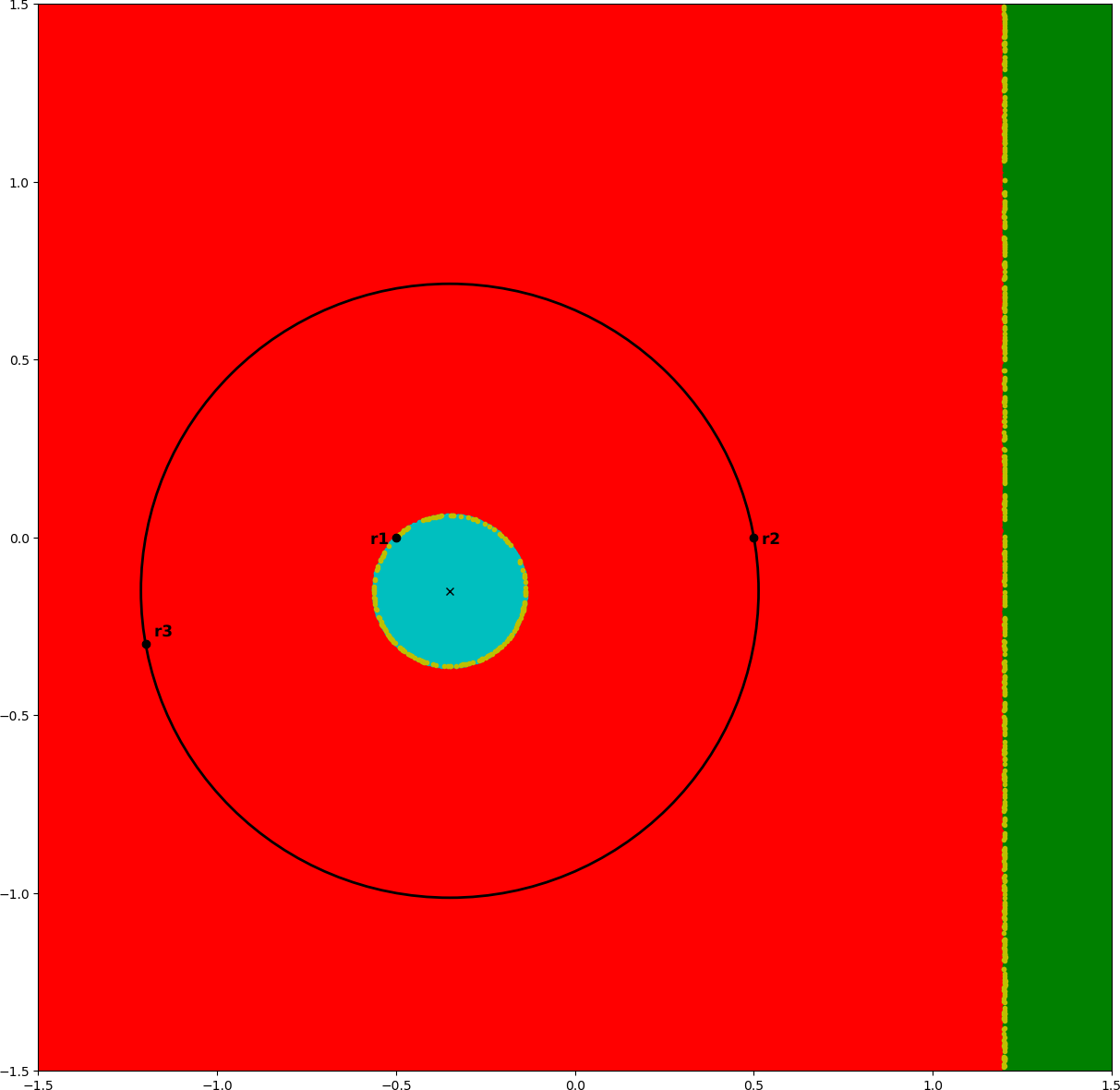}
    \end{subfigure}
    \begin{subfigure}[b]{0.49\linewidth}
        \centering
        \includegraphics[width=\textwidth]{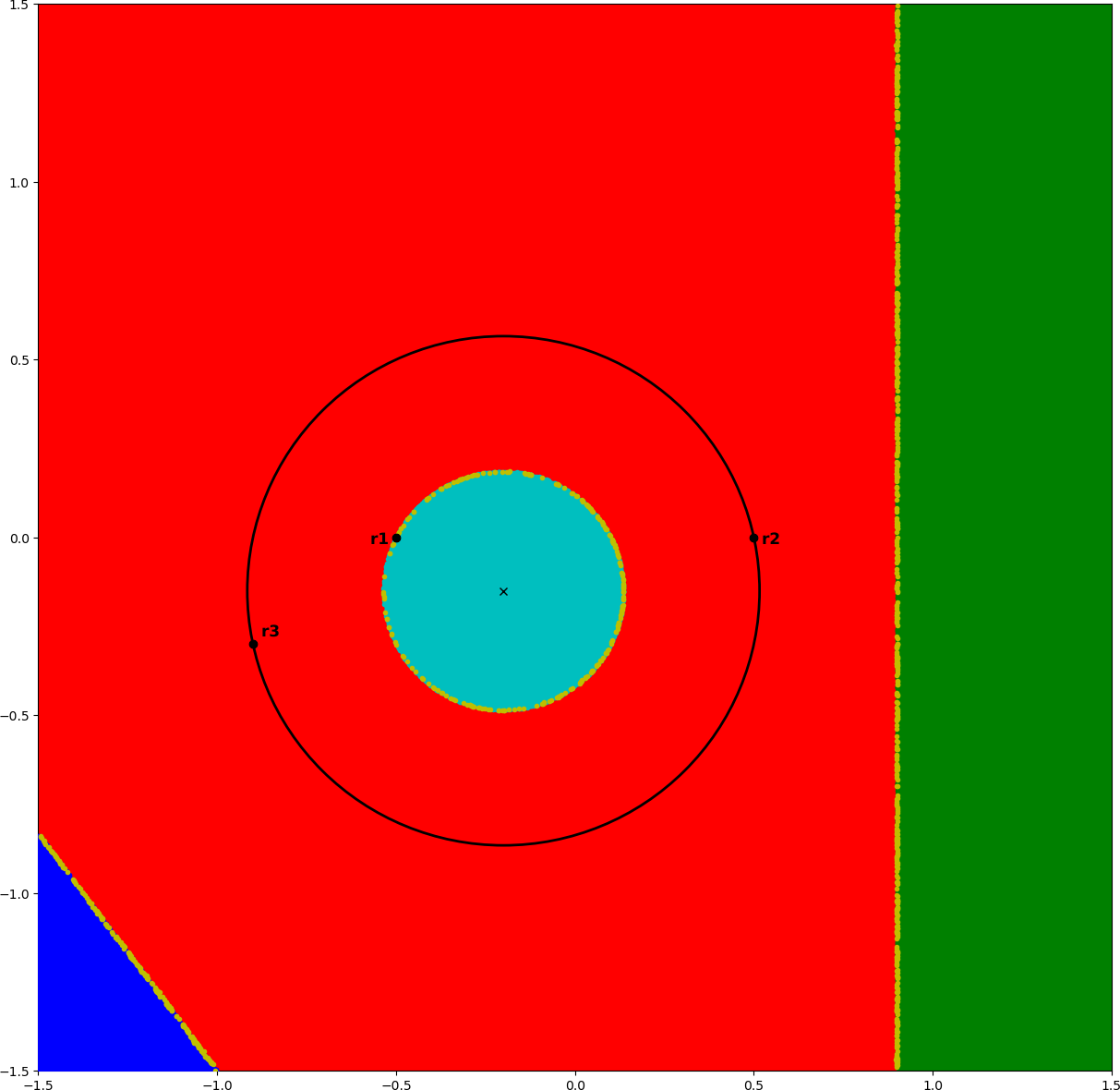}
    \end{subfigure}
    \hfill
    \begin{subfigure}[b]{0.49\linewidth}
        \centering
        \includegraphics[width=\textwidth]{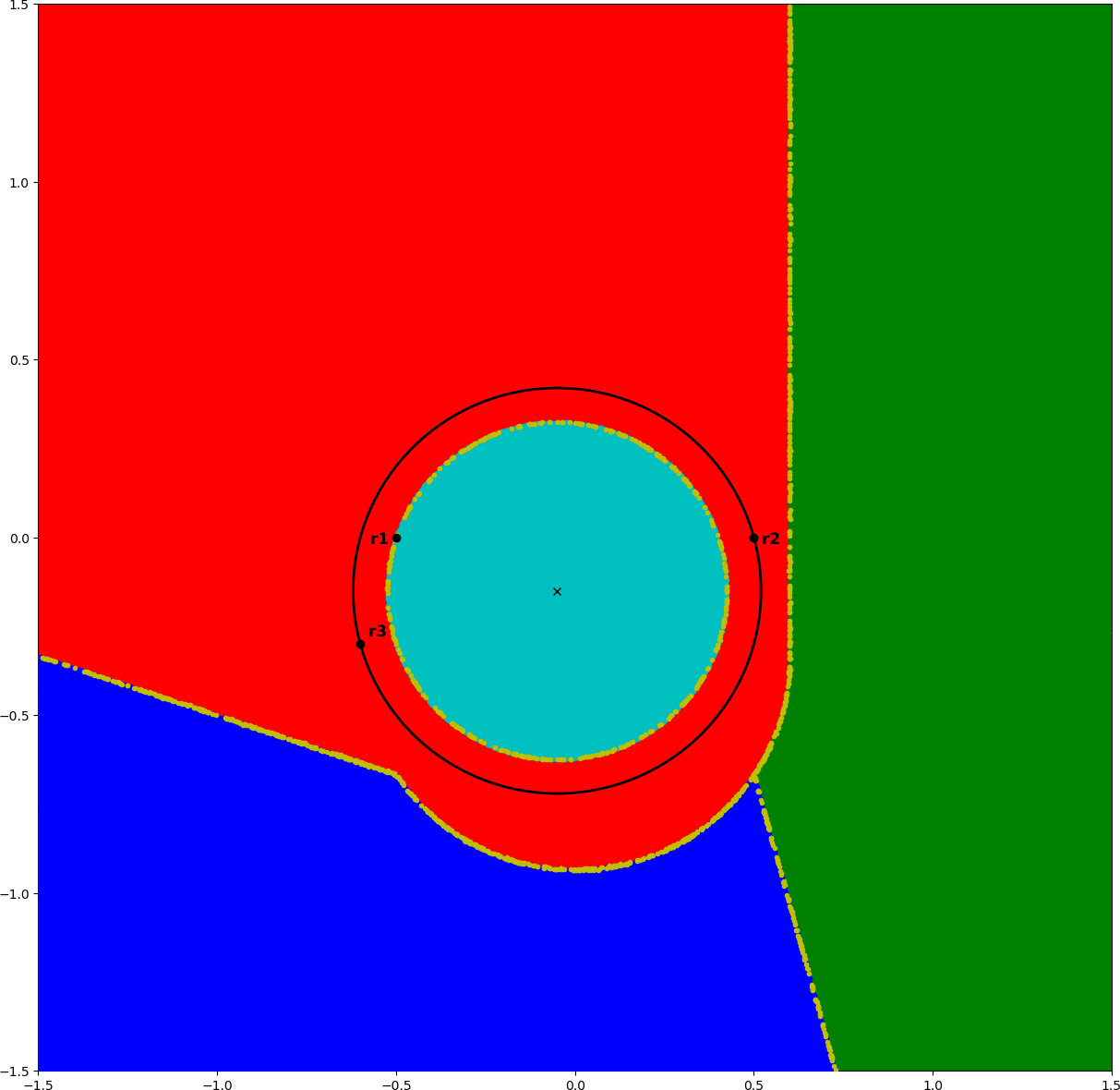}
    \end{subfigure}
    \begin{subfigure}[b]{0.49\linewidth}
        \centering
        \includegraphics[width=\textwidth]{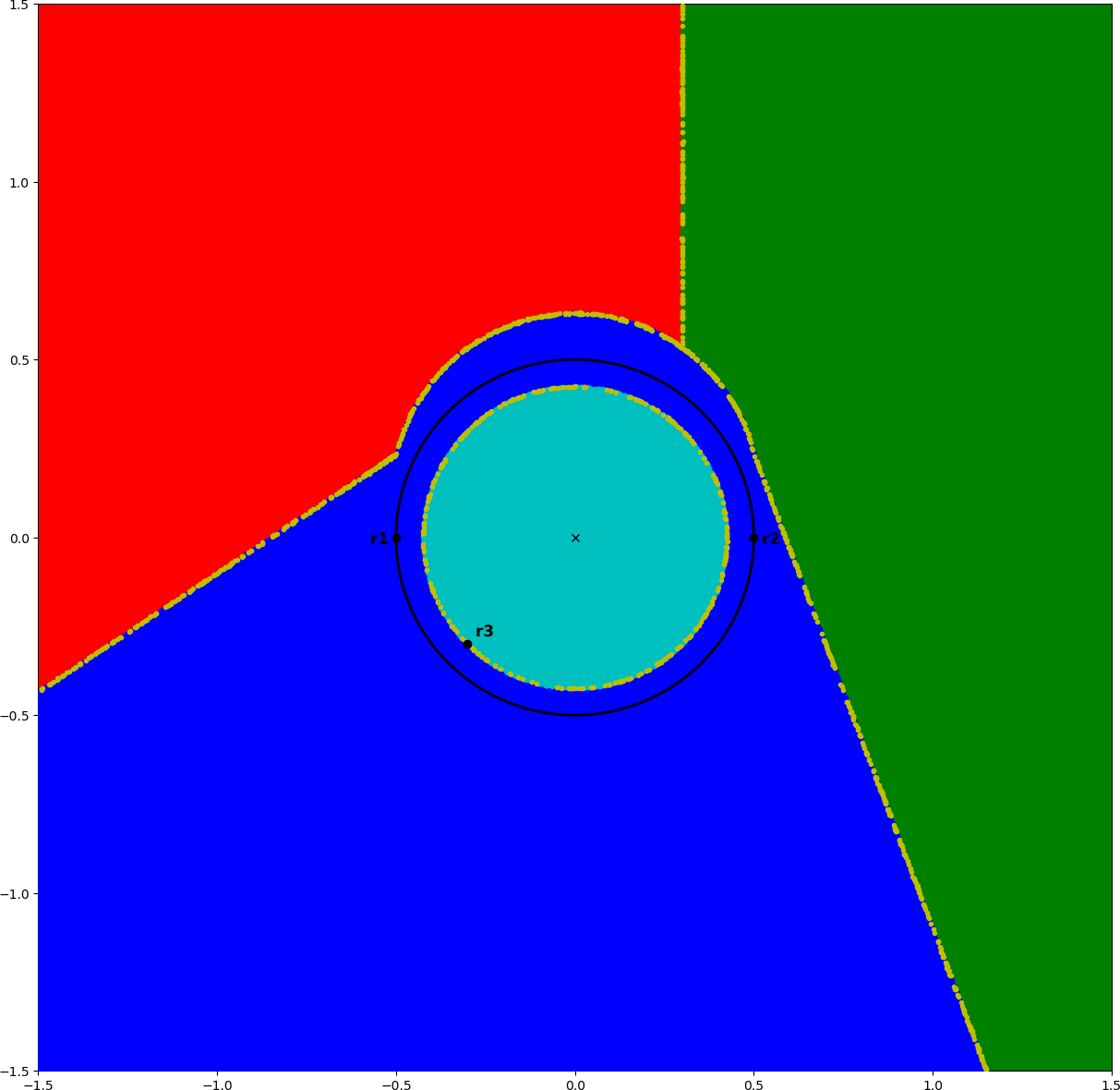}
    \end{subfigure}
    \hfill
    \begin{subfigure}[b]{0.49\linewidth}
        \centering
        \includegraphics[width=\textwidth]{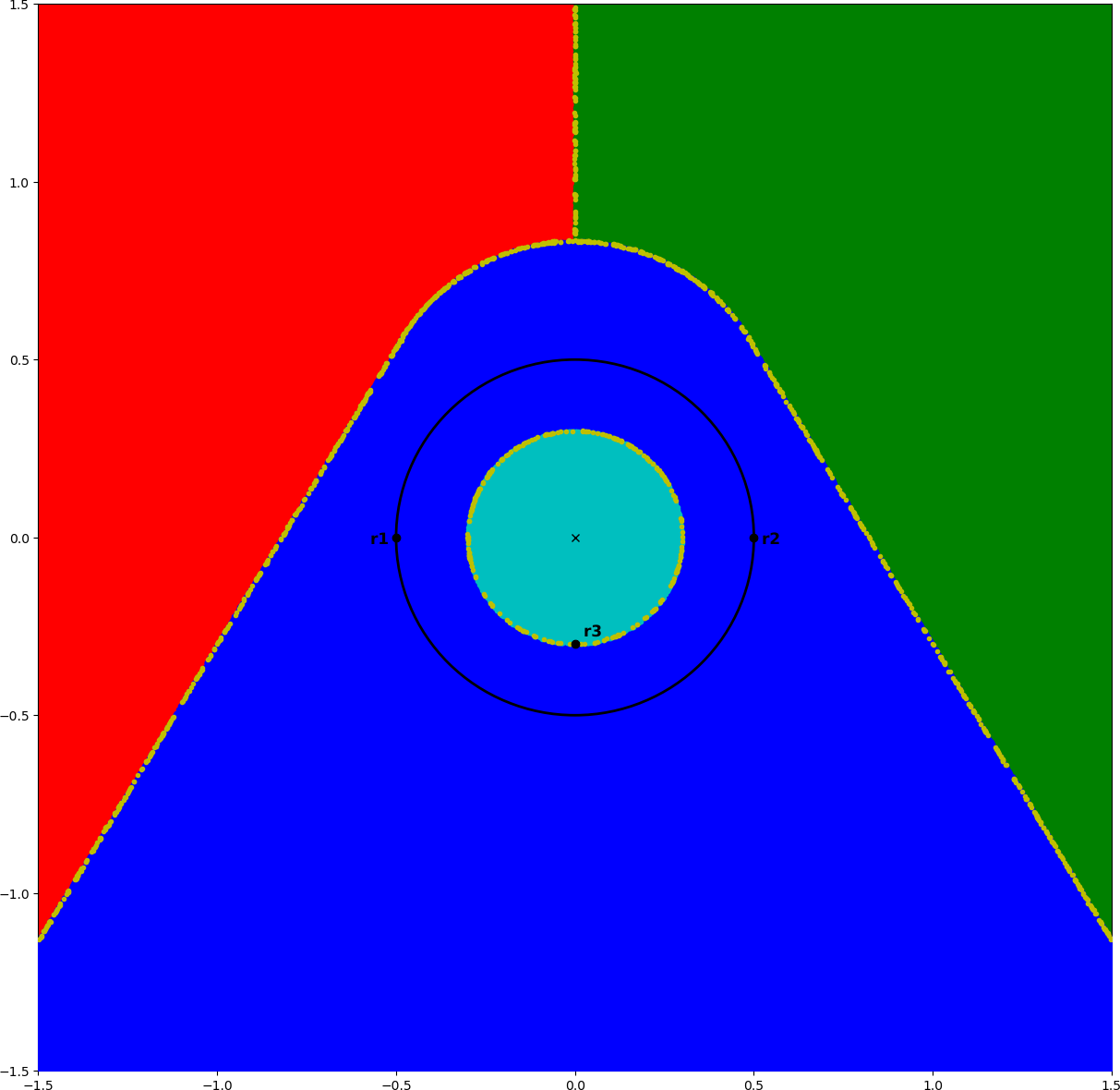}
    \end{subfigure}
    \caption{}
    \label{fig:lead4_4}
\end{figure}

\begin{figure}
    \centering
    \begin{subfigure}[b]{0.49\linewidth}
        \centering
        \includegraphics[width=\textwidth]{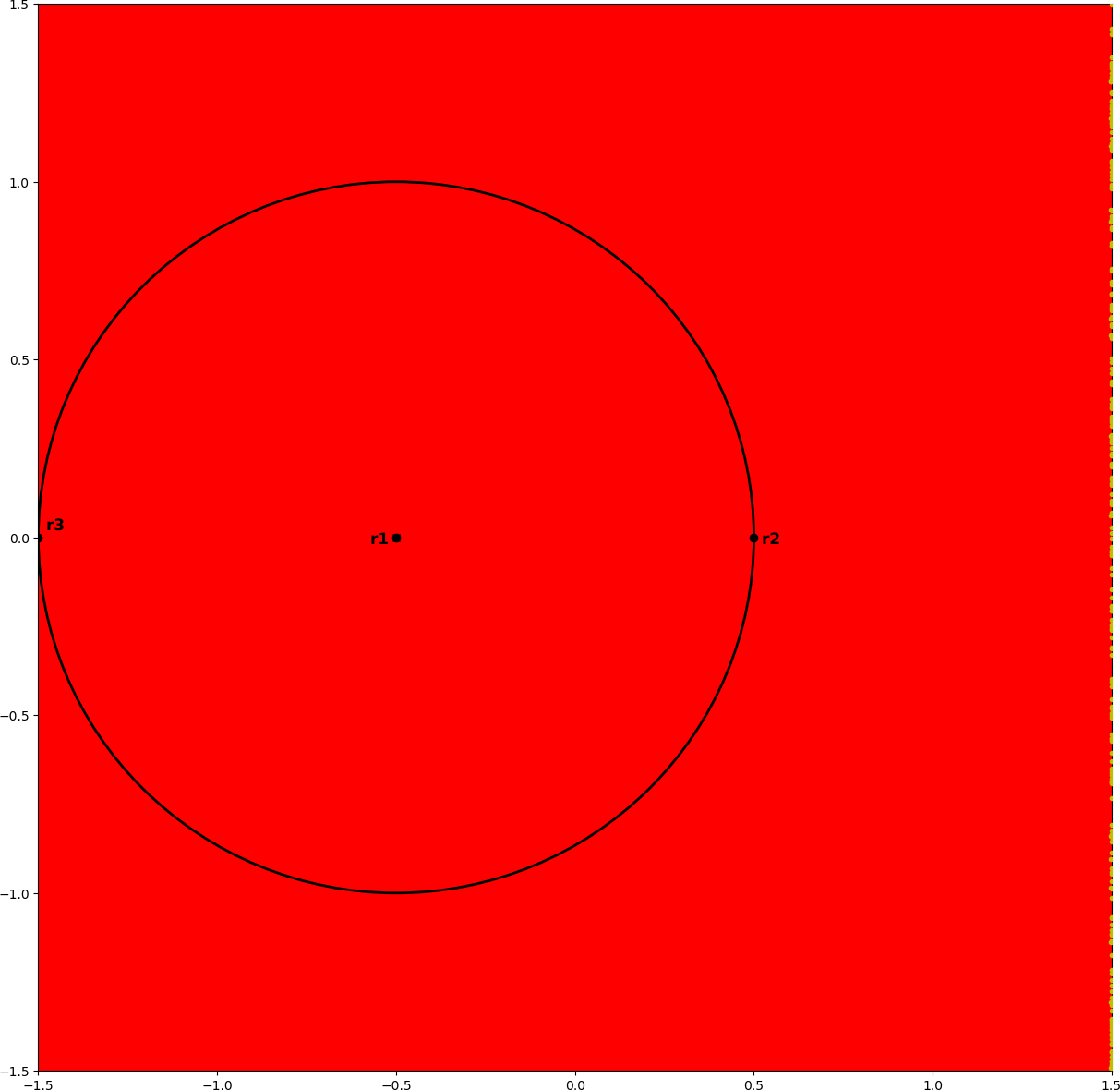}
    \end{subfigure}
    \hfill
    \begin{subfigure}[b]{0.49\linewidth}
        \centering
        \includegraphics[width=\textwidth]{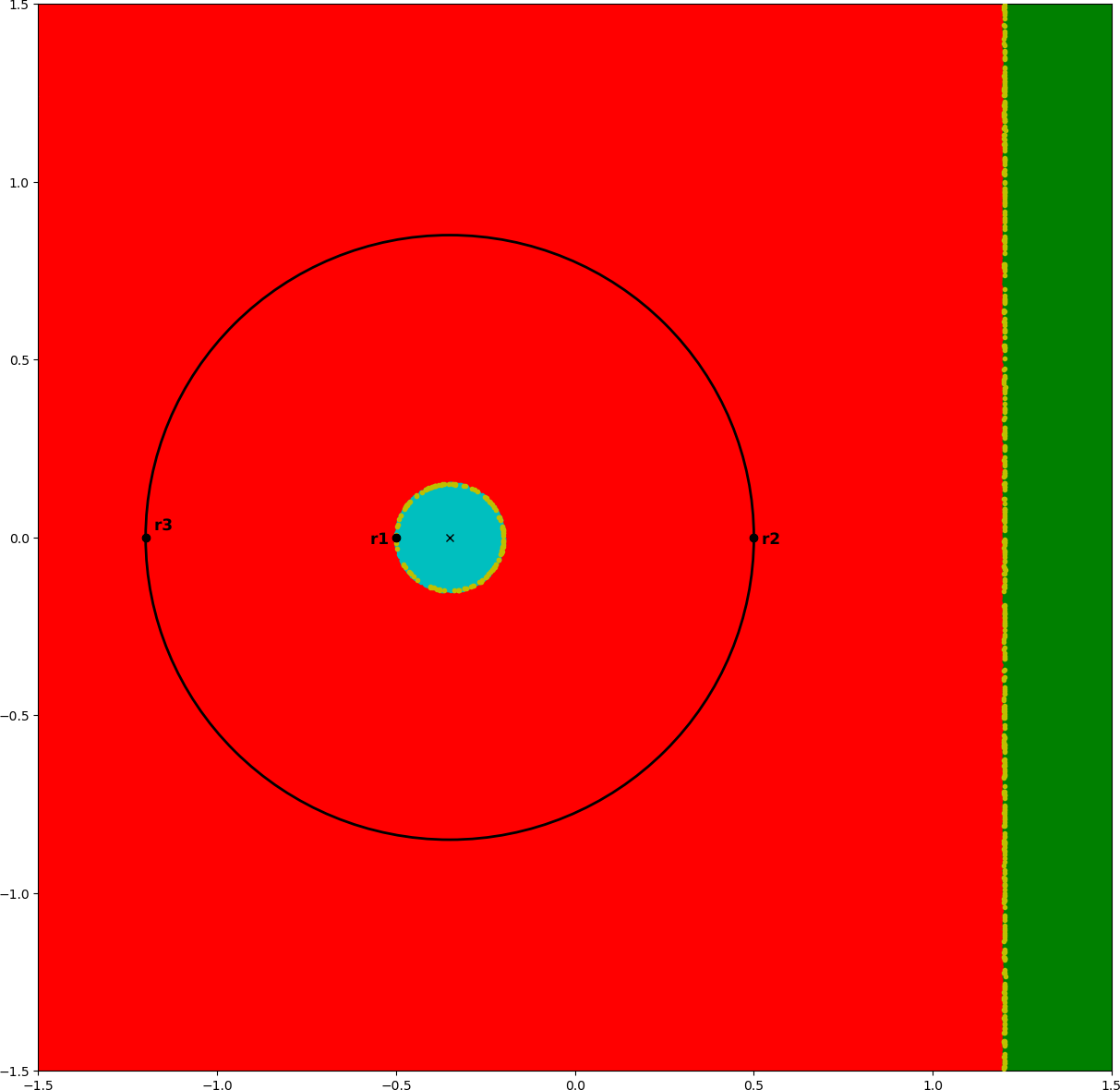}
    \end{subfigure}
    \begin{subfigure}[b]{0.49\linewidth}
        \centering
        \includegraphics[width=\textwidth]{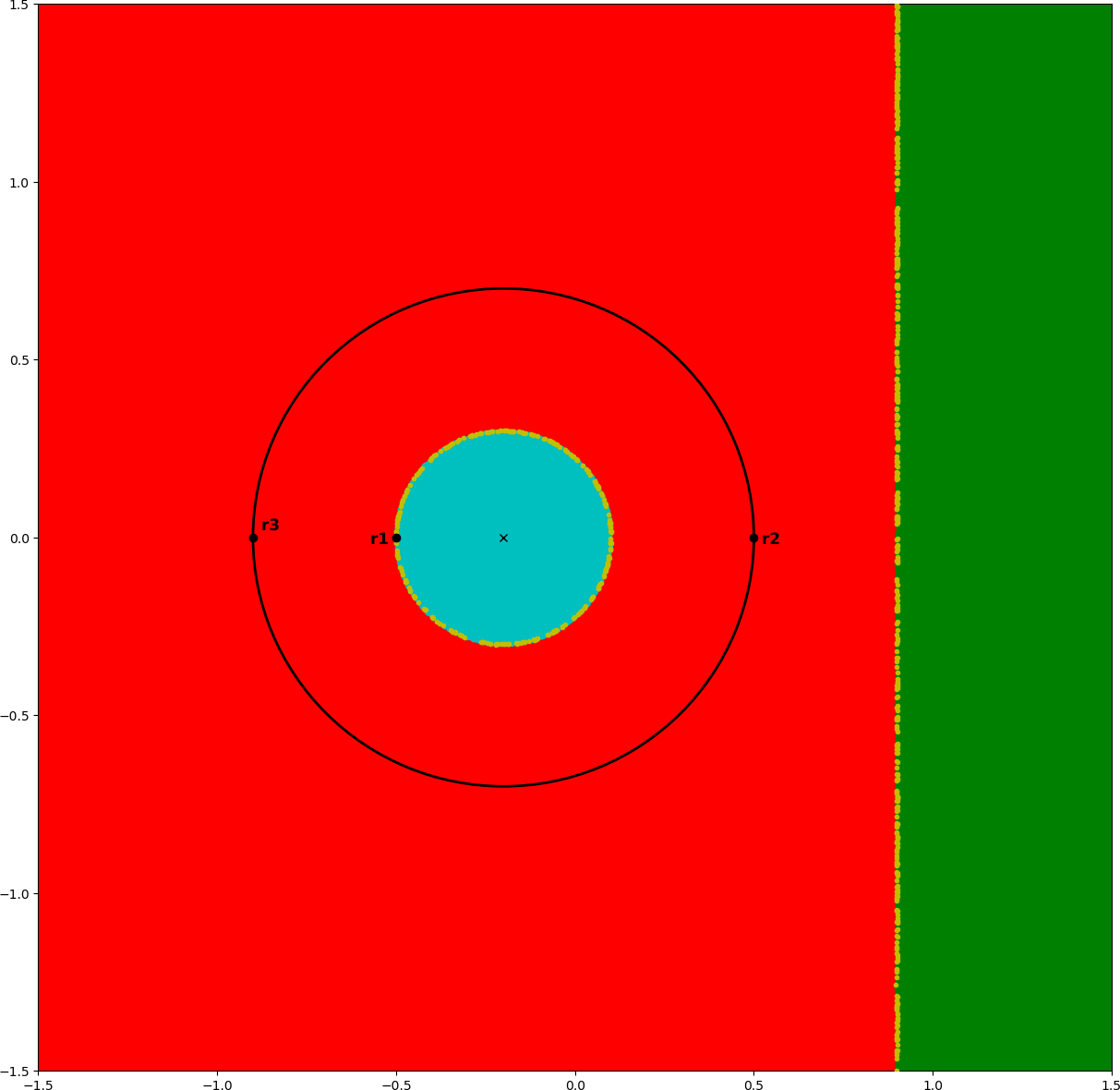}
    \end{subfigure}
    \hfill
    \begin{subfigure}[b]{0.49\linewidth}
        \centering
        \includegraphics[width=\textwidth]{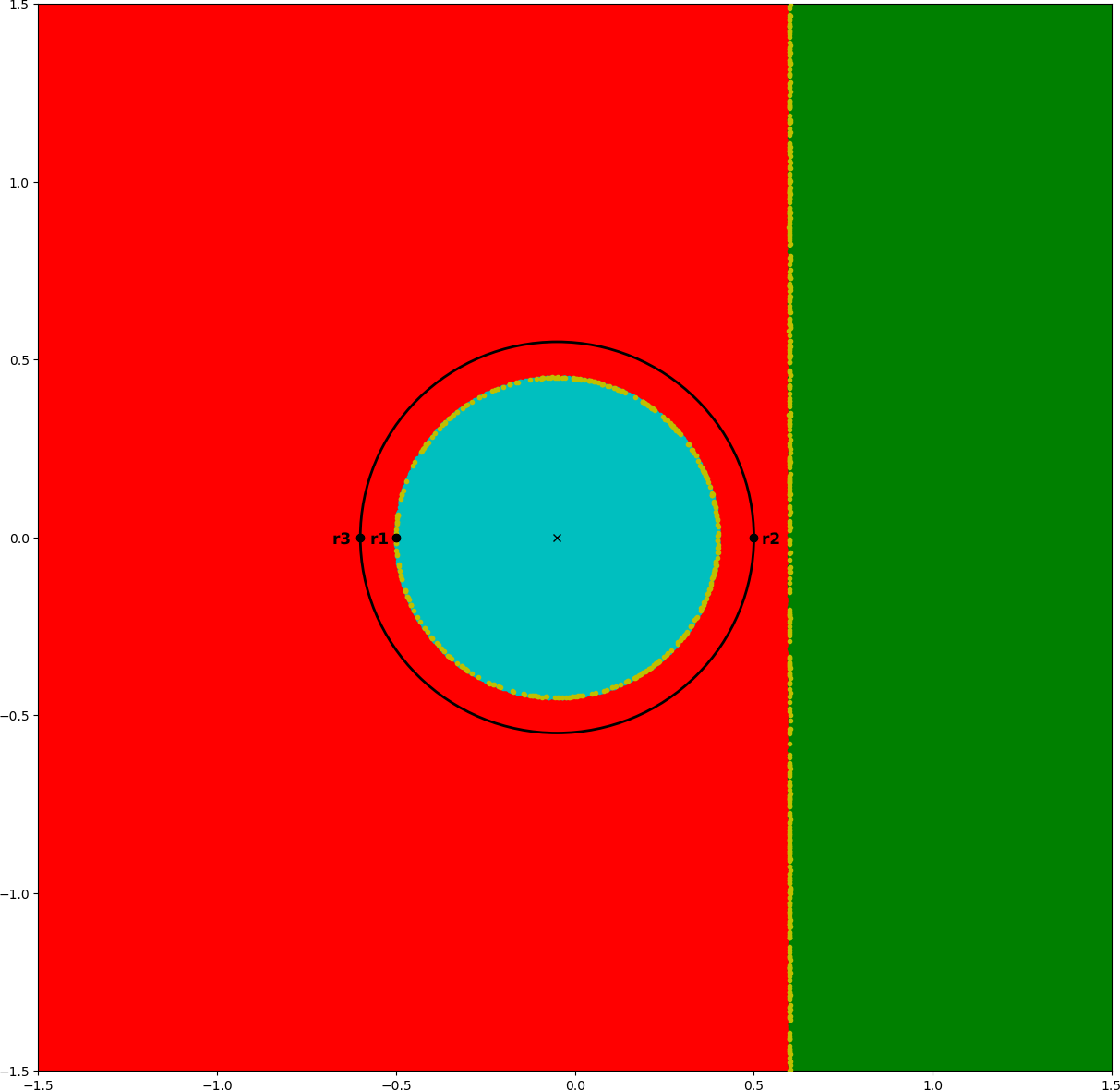}
    \end{subfigure}
    \begin{subfigure}[b]{0.49\linewidth}
        \centering
        \includegraphics[width=\textwidth]{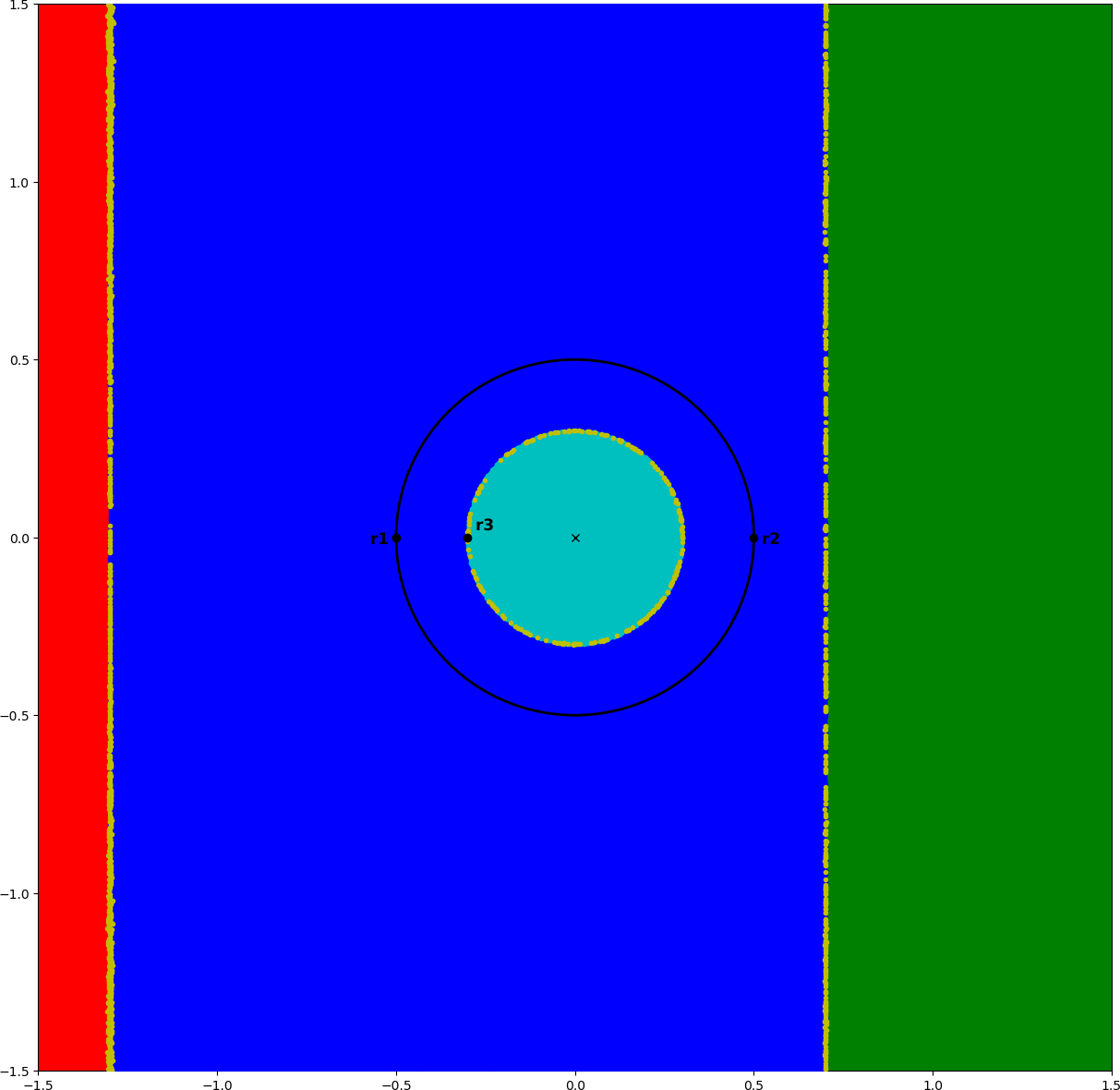}
    \end{subfigure}
    \hfill
    \begin{subfigure}[b]{0.49\linewidth}
        \centering
        \includegraphics[width=\textwidth]{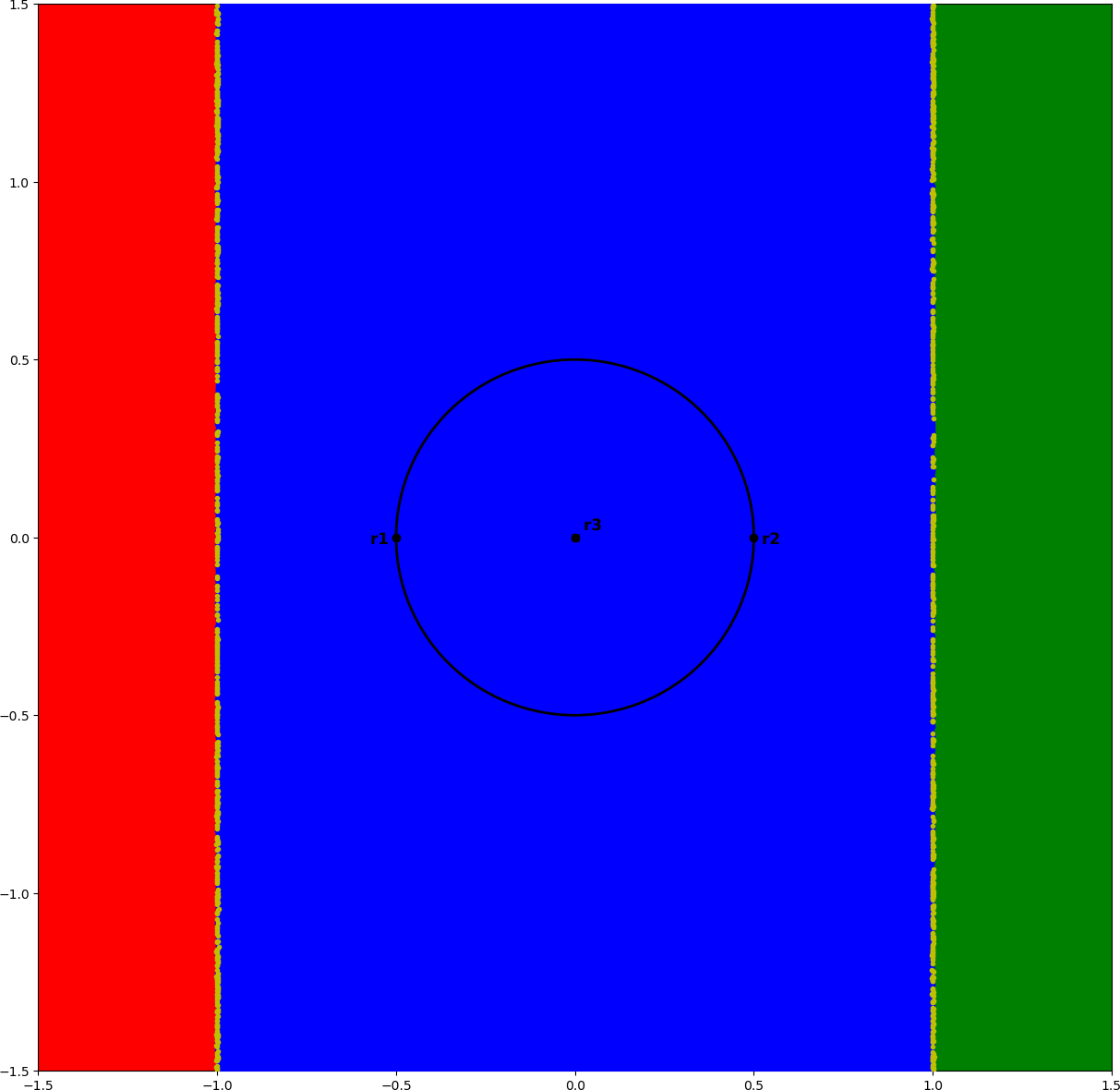}
    \end{subfigure}
    \caption{}
    \label{fig:lead4_5}
\end{figure}

\end{document}